\documentclass{article}
\usepackage[utf8]{inputenc}
\usepackage{amsmath,amssymb,amsfonts,amsthm}
\usepackage[english]{babel}
\usepackage{parskip}
\usepackage{hyperref}
\usepackage{tikz}
\usepackage{graphicx}
\usepackage{float}
\usepackage{subcaption}
\usepackage{fancyhdr} 
\usepackage{geometry}
\usepackage{setspace}
\usepackage{csquotes}
\usepackage{comment}
\usepackage{authblk}

\newtheorem{definition}{Definition}[section]
\newtheorem{lemma}{Lemma}[section]
\newtheorem{theorem}{Theorem}[section]
\newtheorem{corollary}{Corollary}[section]
\newtheorem{remark}{Remark}[section]

\newcommand{\keywords}[1]{%
\par\noindent\textbf{Keywords: }#1
}

\renewcommand{\P}{\textbf{P}}
\newcommand{\NC}{\textbf{NC}}

\usepackage{hyperref}
\hypersetup{
  colorlinks=true,
  linkcolor=black,
  citecolor=black,
  filecolor=black,
  urlcolor=black,
}

\usepackage[shortlabels]{enumitem} 
\usepackage[capitalize]{cleveref} 
\usepackage{mdframed}

\usepackage{comment}

\usepackage{newpxtext}

\let\H\relax

\usepackage{ammacros}

\usepackage{microtype}

\usepackage[author=]{fixme}
\fxusetheme{color}

\usepackage{tikz}
\usetikzlibrary {arrows.meta}
\newcounter{row}
\newcounter{col}
\newcounter{rows}
\newcounter{cols}

\newcommand\setrow[1]{
  \setcounter{col}{1}
  \foreach \n in {#1} {
    \edef\x{\value{col} - 0.5+\value{cols}}
    \edef\y{-\value{row} - 0.5 + \value{rows}}
    \node[anchor=center] at (\x, \y) {\n};
    \stepcounter{col}
  }
  \stepcounter{row} 
}

\title{Embedding Boolean circuits into fungal automata with arbitrary update sequences}
\author{E.~Goles\thanks{Facultad de Ingeniería y Ciencias, Universidad Adolfo Ibáñez, Santiago, Chile. E-mail: eric.chacc@uai.cl}
,
A.~Modanese\thanks{Aalto University, Espoo, Finland. E-mail: augusto.modanese@aalto.fi}
,
M.~Ríos-Wilson\thanks{Facultad de Ingeniería y Ciencias, Universidad Adolfo Ibáñez, Santiago, Chile. E-mail: martin.rios@uai.cl}
,
D.~Ruiz-Tala\thanks{Departamento de Ingeniería Matemática, Facultad de Ingeniería y
Ciencias, Universidad de Chile, Santiago, Chile. E-mail: domruiz123@gmail.com}
~and 
T.~Worsch\thanks{Karlsruhe Institute of Technology (KIT), Karlsruhe, Germany (retired). Email: thomas.worsch@posteo.de}}
\date{February 2025}

\begin{document}
\maketitle

\begin{abstract}
  The sandpile automata of Bak, Tang, and Wiesenfeld [Phys.\ Rev.\ Lett., 1987]
  are a simple model for the diffusion of particles in space.
  A fundamental problem related to the complexity of the model is predicting its
  evolution in the parallel setting.
  Despite decades of effort, a classification of this problem for
  two-dimensional sandpile automata remains outstanding.
  Fungal automata were recently proposed by Goles et al.~[Phys.\ Lett.\ A, 2020]
  as a spin-off of the model in which diffusion occurs either in horizontal
  ($H$) or vertical ($V$) directions according to a so-called update scheme.
  Goles et al. proved that the prediction problem for this model with the update
  scheme $H^4V^4$ is $\P$-complete.
  This result was subsequently improved by Modanese and Worsch [Algorithmica,
  2024], who showed the problem is $\P$-complete also for the simpler update
  scheme $HV$.
  In this work, we fill in the gaps and prove that the prediction problem is
  $\P$-complete for any update scheme that contains both $H$ and $V$ at least
  once.
\end{abstract}

\fxnote{Can we please use the standalone package so that tikz doesn't need to
compile every figure from scratch every time?}
\fxnote{There are a lot of warnings in code for figures. Please fix.}

\section{Introduction}

The \emph{sandpile model} of Bak, Tang, and Wiesenfeld \cite{bak1987self} is a
classical discrete model for particle diffusion.
Concretely, one might picture \emph{grains of sand} that are distributed across
a lattice of \emph{cells} that partitions Euclidean space.
Each cell contains a certain number of grains of sand and, when this number
exceeds a fixed threshold value, the cell \emph{fires}, equally distributing
grains to its neighbors. 

The sandpile model has been extensively studied from the lenses of cellular
automata theory, where it is interpreted as a class of \emph{sandpile automata}.
In the conventional model of sandpile automata in two dimensions, the neighbors
of a cell are fixed to be the ones immediately located along its four cardinal
directions.
The states are simply numbers from $0$ to $7$ indicating the number of grains in
a cell, and the threshold for firing is $4$.
\fxnote{Would be nice to have a figure}
The complexity of predicting the behavior of sandpile automata in two dimensions
is a long-standing open problem in the field \cite{formenti2020hard}.

Recently, Goles et al.~\cite{goles2020computational} proposed a variant of
sandpile automata that was in turn inspired by another class of automata called
\emph{fungal automata} \cite{adamatzky2023fungal} (see also
\cite{sepulveda2023exploring}).
These are a nature-inspired model of computing that mimics the behavior of septa
in fungi. 
Septa are compartments that are located inside the cytoplasm of fungi organisms
and that communicate flow through structures called Woronin bodies, which change
their behavior depending on external conditions to be closed (preventing the
flow from one septum to another) or open. 

Concretely, the \emph{fungal sandpile automata} defined by Goles et
al.~\cite{goles2020computational} are (two-dimensional) sandpile automata where
the neighborhood of a cell is not fixed but dynamically changes over time
according to a so-called \emph{update scheme}.
This is a word $Z \in \{H,V\}^+$ that defines the neighborhood of cells at every
time step following a cyclic rotation.
The $H$ and $V$ symbols in $Z$ indicate whether at the corresponding time point
the cells distribute grains to their horizontal (i.e., east and west) or
vertical (i.e., north and south) neighbors.

\subsection{The Problem}

Given a complex discrete dynamical system, a fundamental question is the
complexity of \emph{predicting} its behavior.
That is, given an initial configuration, an objective entity in the system
(e.g., the state of a cell in a cellular automaton), and a bound on the
observation time, will the entity assume some state of interest or not?
In the case of sandpiles, we are particularly interested in determining if a
cell will ever contain a grain of sand.
For the specific variant of fungal sandpiles, our \emph{fungal sandpile
prediction} problem is formally defined as follows.

\begin{mdframed}
  \begin{description}
    \item[Parameters:] A word $Z \in \{H,V\}^+$.
    \item[Input:] The initial configuration $c$ for a finite grid $R$ of cells,
    the index of a cell $x$ (encoded in binary), and a positive time bound $T$
    (encoded in unary).
    \item[Question:] Suppose the initial configuration of the fungal sandpile is
    such that we place $c$ in $R$ surrounded by cells all in the state $0$.
    Will cell $x$ assume a non-zero state at some time point $t \le T$?
  \end{description}
\end{mdframed}

Characterizing the complexity of predicting cellular automata and related models
such as automata networks is a recurring object of study \cite{goles2016pspace,%
  concha-vega25_sandpiles_natcomput,%
  ollinger22_freezing_dmtcs,nakar25_characterizing_arxiv,%
  goles2020complexity,goles2021complexity,rios2024intrinsic1,%
  rios2024intrinsic2,RIOSWILSON2024114890},
and the corresponding prediction problem for sandpile automata has been the
subject of intensive efforts as well \cite{formenti2020hard}.
Letting $d$ denote the number of dimensions of the sandpile automaton, for $d =
1$ it is known that the prediction problem can be solved in a parallel setting
as it is contained in a subclass of $\NC$ \cite{miltersen2007computational}.
Conversely, for $d \ge 3$, the problem has been shown to be $\P$-complete
\cite{moore1999computational}. 
As already mentioned, the complexity of the problem for $d = 2$ has been open
for very long.
One of the main motivations for studying fungal sandpile automata (or other
variants of sandpile automata) is to explore and develop techniques that could
shed more light into the case of ordinary two-dimensional sandpiles.

In their original paper, Goles et al.~\cite{goles2020computational} showed that,
if we set $Z = H^4 V^4$, then the fungal sandpile prediction problem associated
with $Z$ is $\P$-complete.
This result was subsequently improved by Modanese and Worsch
\cite{modanese2024embedding}, who introduced several novel techniques in order
to prove that $\P$-completeness is preserved when $Z = HV$, which is the
simplest non-trivial update scheme possible (see also \cref{sec:intro-tech}
further below).
The fundamental approach in both papers is by a reduction from the circuit
value problem (CVP), which is well-known to be $\P$-complete, but relying on
different strategies to embed Boolean circuits in fungal sandpiles.

\subsection{Our Contribution}

In this paper, we settle an open question of Modanese and Worsch
\cite{modanese2024embedding} and show that the fungal sandpile prediction
problem is $P$-complete \emph{whenever $Z \in \{H,V\}^+$ is a non-trivial update
scheme}, that is, whenever it contains $H$ and $V$ at least once.
This gives a full characterization of the problem since it is clear that, if $Z$
consists of only $H$ or only $V$, then it is equivalent to prediction on
ordinary sandpiles in the one-dimensional case
\cite{miltersen2007computational}.
More precisely, we have the following result:

\begin{theorem}
  If $Z \in \{ H,V \}^+$ is any word that contains at least one $H$ and one $V$,
  then the fungal sandpile prediction problem associated with $Z$ is
  $\P$-complete.
\end{theorem}

To show this result, we adapt the modular approach of Modanese and Worsch
\cite{modanese2024embedding} and show it is possible to implement the same
underlying set of gadgets when we have such a word $Z$.
The main challenge to overcome is coping with the fact of $Z$ no longer being
fixed.
This means the entire description and proof of correctness for the gadgets must
be generic enough to handle any such given $Z$.
In particular, the structure of the gadgets is no longer static and we can no
longer refer to them in a direct manner.
See \cref{sec:intro-tech} for further discussion.

\subsection{Implications and Future Work}
Before we turn to the technicalities involved in proving the result, let us
briefly discuss what it means for the case of ordinary sandpiles.
If one could show how to implement fungal sandpiles in ordinary sandpiles, then
of course our result would immediately imply hardness in that case as well.
Unfortunately, however, no reduction between the two problems is presently
known.

That being said, there are still meaningful inferences to be made.
We focus here on points that were not already made by Modanese and Worsch in
their work \cite{modanese2024embedding}.
More specifically, we have two main observations to make:
\begin{enumerate}
  \item (As is the case with the work of Modanese and Worsch) the construction
  relies heavily on the fact that we have alternations between horizontal and
  vertical phases, which lets us induce \emph{signals} of different polarities
  (positive and negative).
  This phenomenon is not only interesting as a feature of fungal sandpiles but
  also fundamental for the construction to function since the signals behave
  distinctly.
  In particular, when signals of inverse polarity reach the same cell, the
  respective sets of cells that are potentially visited next are disjoint.
  This seems to be an inherent feature of fungal sandpiles and it is not
  reasonable to expect the same would be possible in ordinary sandpiles.
  \item We have a sharp threshold between the one-dimensional case that is
  embedded in the problem (i.e., the case where $Z$ contains only $H$ or $V$)
  and the more complex two-dimensional one.
  This could be seen as an indication that a sharp threshold in the case of
  ordinary sandpiles is to be expected as well.
\end{enumerate}
One possibility for these two observations to be consistent with each another is
that the characterizing class for the prediction problem on ordinary sandpiles
lies in-between $\P$ and $\NC$.  
The reason for believing this is that we should also see a threshold from the
one- to the two-dimensional case (hence we are beyond $\NC$) but, as already
mentioned, it is perhaps not possible to obtain $\P$-completeness---which would
also explain why all previous approaches at proving it have failed.
Hence one promising direction of study would be to try and embed simpler classes
of circuits.
It is already known how to embed circuits that are monotone and planar (but not
just one of the two, as that would already give us a proof of
$\P$-completeness), so one could consider similarly restricted classes of
circuits.

\subsection{Technical Discussion}
\label{sec:intro-tech}

We start by recalling the necessary details from the construction of Modanese
and Worsch \cite{modanese2024embedding}, where the same result for the special
case $z = HV$ was proven.
As already mentioned above, this is done by a reduction from the circuit value
problem (CVP), which we restate here for the sake of concreteness.

\begin{mdframed}
  \begin{description}
    \item[Input:] The description of a single-output Boolean circuit $C$ with
    bounded fan-in as well as an input $x \in \{ 0,1 \}^+$ to the circuit.
    \item[Question:] What is the output of $C$ on $x$?
  \end{description}
\end{mdframed}

In order to obtain a sensible problem, some care must be taken regarding the
format of the description in which $C$ is given in.
The problem is well-known and these details are not relevant in the present
paper, and so we omit a more detailed treatment.
It suffices to say that the description must be structured in such a way that it
can be processed gate-by-gate by a logspace machine.
It is a well-known result that the CVP is $\P$-complete, which means the problem
is in $\P$ and also that any problem in $\P$ can be reduced to it by a logspace
reduction.

Let us return to the construction of Modanese and Worsch.
It is structured into multiple \emph{layers}:
\begin{description}
  \item[Layer 0:] the underlying model of fungal automata
  \item[Layer 1:] partitioning of the grid into blocks of size $2 \times 2$
  \item[Layer 2:] polarized circuits
  \item[Layer 3:] Boolean circuits with delay
  \item[Layer 4:] embedding of a Boolean circuit in logspace
\end{description}
Fortunately, due to this modular structure, we do not need to redo the entire
construction from scratch.
Rather, we only need to (heavily) adapt layers 1 and 2 and take care of some
details in layer 4.
Layer 3 can be left untouched, and hence we do not review it here; the
interested reader may consult \cite{modanese2024embedding} for the details.

In order to better structure this discussion, we describe the required
adaptations layer by layer while at the same time reviewing necessary details
and terminology from the original construction \cite{modanese2024embedding}.

\paragraph{Layer 0: fungal automata.}
We provide a slightly more technical definition of fungal automata than in the
preceding discussions.
(The reader may consult \cref{sec:prel} for the full definition.)
The underlying space is the lattice $\Z^2$. 
Each cell in $\Z^2$ may assume a state from a set $Q = \{ 0, \dots, 5 \}$
through a configuration $c\colon \Z^2 \to Q$, where $c(x)$ indicates the number
of grains in cell $x \in \Z^2$.
There are two local rules $H$ and $V$ that may be applied:
\begin{itemize}
  \item $H$ \emph{fires} cells with state $4$ or $5$, reducing their state by
  $2$ (i.e., mapping state $4$ to state $2$ and $5$ to $3$). 
  Meanwhile cells that have neighbors along a same row in firing states gain one
  unit per such neighbor (e.g., if a cell is in state $3$ and one neighbor along the same row in state $4$ and one in state $5$, then it turns into a $5$).
  \item $V$ has the same behavior as $H$ but considers neighbors along columns
  instead of rows.
\end{itemize}
The global transition function $F$ applies either rule everywhere simultaneously
to obtain the next configuration from the previous one.
The choice of local rule follows an update scheme, which is a word $Z \in \{ H,V
\}^+$.
Starting from the initial configuration, the $i$-th rule applied by $F$ is 
$Z_{(i-1) \bmod \abs{Z}}$ (e.g., if $Z = HVHV$, then the seventh rule applied by
$F$ is $Z_{(7-1) \bmod 4} = Z_2 = H$).

As already stated multiple times, the original result addresses the case $Z =
HV$, whereas we would like to prove the same for arbitrary $Z$ (that does
consist of only $H$ or only $V$).
Using simple observations, without restriction we can assume that $Z$ starts
with $H$ and ends with $V$.
In addition, as otherwise some features of the construction are degenerate, we
also assume for this discussion that $Z$ contains at least two $H$ and two $V$
symbols.

\paragraph{Layer 1: block partitioning.}
The first layer consists of partitioning the grid into blocks of $2 \times 2$
cells.
This allows better organization of \emph{signals}, which are cells with a state
$4$ or $5$.
They are so called because they induce a self-propagating behavior when placed
on a path of cells with state $3$ (see \cref{fig:ciclo}).

In our generic construction, the block sizes are determined by the composition
of $Z$.
More specifically, letting $h$ and $v$ be the counts of the symbols $H$ and $V$
in $Z$, respectively, our blocks are of size $h \times v$.
The blocks in our construction are inspired by the original one but they 
actually constitute a more fine-grained division than in the original.
To be precise, we use $2 \times 2$ blocks for update schemes such as $Z =
HHVV$, whereas for $Z = HV$ we would actually have a degenerate partition into
$1 \times 1$ blocks.
(This is the reason for us skipping the cases $h = 1$ and $v = 1$ in this
high-level overview.)

The general approach is to have every block simulate a single cell in the
original $HV$ construction.
Assuming this can be realized correctly, this allows us to apply the higher
layers of abstraction without having to alter them too much.
That being said, it is not possible to obtain such a simple simulation directly;
in fact, it is necessary for a block to use also the cells that are adjacent to
its boundaries.
Hence it is not trivial how to \enquote{glue} blocks together and also not at
all clear how to handle corners, that is, cells where four distinct blocks
intersect.
This is the first technical difficulty to be overcome and requires extra care
throughout the construction.
It is also the reason why we need to adapt the next layer and redo the
constructions for the gadgets in it from scratch.

\paragraph{Layer 2: polarized circuits.}
As already foreshadowed, this layer is where the bulk of our technical effort
goes into.
In the original construction, a set of gadgets is defined that is used to build
a limited class of circuits that operates on \emph{polarized signals}, that is,
signals that may be \emph{positive} or \emph{negative}, which is determined by
their relative placement inside a block.
It is a feature of the fungal automata model that the parity of these signals is
conserved (regardless of the update scheme used), and we exploit this in the
construction.

We provide a brief overview of the set of gadgets required:
\begin{itemize}
  \item Basic wire operations to merge and duplicate signals of the same
  polarity.
  \item Monotone Boolean gates (i.e., AND and OR) that operate on signals of the
  same polarity.
  \item \emph{Semi-crossings}, which are a rudimentary form of crossing
  \emph{between signals of opposite polarity} that is destroyed upon use.
  \item Operators called \emph{switches} that allow a certain logical operation
  similar to that of a transistor between signals of opposite polarities.
  \item \emph{Diodes} to ensure signals traverse a wire in a single direction.
  \item \emph{Retarders} to allow delaying a signal by a variable amount of
  time.
\end{itemize}
Using the original implementations as a blueprint, we construct the same set of
components, and hence the result is fully compatible with the subsequent layer.
Nevertheless, we must construct these all from scratch and also provide details
as to how to implement them for arbitrary $Z$.

The key notion that we employ to realize all these gadgets while coping with a
variable $Z$ are so-called \emph{bridges}.
Their core functionality is to convey signals from one block to the other.
Each bridge is implemented using the inside of a block as well as the one-cell
neighborhood around it, which we refer to as the \emph{closure} of the block.
This means that adjacent bridges potentially use \emph{overlapping cells}, and
hence technical care is needed to ensure there are no side effects.

Since we use bridges throughout to realize the more complex components of layer
2, one can also think of them as an \enquote{intermediate layer} between blocks
and polarized circuits.
The advantage of this approach is that we must only define a \emph{single}
gadget (namely the bridge) based on $Z$ and from which the construction of all
layer 2 components follows in a generic manner.

\paragraph{Layer 3: Boolean circuits with delay.}
As already mentioned, we do not need to change anything in this layer.
Hence we skip describing it and reference the interested reader to
\cite{modanese2024embedding}.

\paragraph{Layer 4: logspace embedding.}
The final layer in the original construction is the \emph{logspace constructor}
that is in charge of taking a CVP instance and actually producing the
description of the embedding associated with it.
Since this layer sits on top of the previous one, not much adaptation is needed
compared to the original construction, though one needs to take care of a few
details since now the constructor depends on $Z$ and needs to take into
consideration additional factors incurred by the use of bridges.

\subsection{Structure of the Article}

The rest of the article is structured as follows: 
In \cref{sec:prel} we present basic definitions and notation.
In \cref{sec:main} we set up the groundwork for the circuit embedding.
Following this, in \cref{sec:bridges} we introduce our core gadget, bridges.
We have two appendices:
\Cref{sec:layer-2} covers the reimplementation of gadgets in layer 2, while in
\cref{sec:layer-4} we discuss the adaptations needed for layer 4.

\section{Preliminaries} \label{sec:prel}

The set of integers is denoted by $\Z$, that of positive integers by $\N_+$, and
that of non-negative integers by $\N_0$.

A two-dimensional \emph{cellular automaton} is a tuple $(Q,N,f)$ where $Q$ is a finite set called the alphabet, $N$ is finite subset of $\Z^2$ called neighborhood and $f:Q^N\mapsto Q$ is called the \emph{local rule} of the cellular automaton. An assignation of states for each cell $c: \Z^2 \mapsto Q$ is called a \emph{configuration}. If $S\subseteq \Z^2$ we denote by $c|_{S}$ to the restriction of the configuration $c$ to the set $S.$ The local function $f$ of a cellular automaton induces a global function $F:Q^{\Z^2} \mapsto  Q^{\Z^2}$  given by $F(c)(x) = f(c |_{x+N}).$ 

A fungal automaton can be seen as a two-dimensional cellular automaton equipped
with the von Neumman neighborhood, i.e., $N = \{(a, b) \in \mathbb{Z} \mid
|a| + |b| \leq 1\}$ and set of states $Q=\{ q \in \N_0 \mid q \le 5 \}$.  
Based on $N$, given a cell $x \in \Z^2$ we let $N_h(x) = \{x,x+(1,0),x-(1,0)\}$
be the \emph{horizontal neighborhood} and $N_v(x) = \{x,x+(0,1),x-(0,1)\}$ the \emph{vertical
neighborhood} of $x$.
Cells are updated according to a two entry update function
$F:Q^{\mathbb{Z}^2}\times \{H,V\} \to S^{\mathbb{Z}^2}$, where for $x
\in\mathbb{Z}^2$:
\begin{align*}
F(c,H)(x) &= c(x) - 2[c(x)\geq4] + \sum_{y \in N_h(x)} [c(y) \geq 4], \\
F(c,V)(x) &= c(x) - 2[c(x)\geq4] + \sum_{y \in N_v(x)} [c(y) \geq 4],
\end{align*}
with $[ \cdot]$ being the usual characteristic function, that is, $[Z(s) = z] =
1$ if $Z(s) = z$ and $0$ otherwise.

In \Cref{fig:ciclo}, we show an example of the dynamics of a fungal automata for a word $Z = HVVHHHV.$  All the empty  cells (cells with $0$ grains) are represented by blank spaces. In addition, we are assuming that all the cells that are not part of the figure ($\Z^2$ is an infinite grid) have $0$ grains.


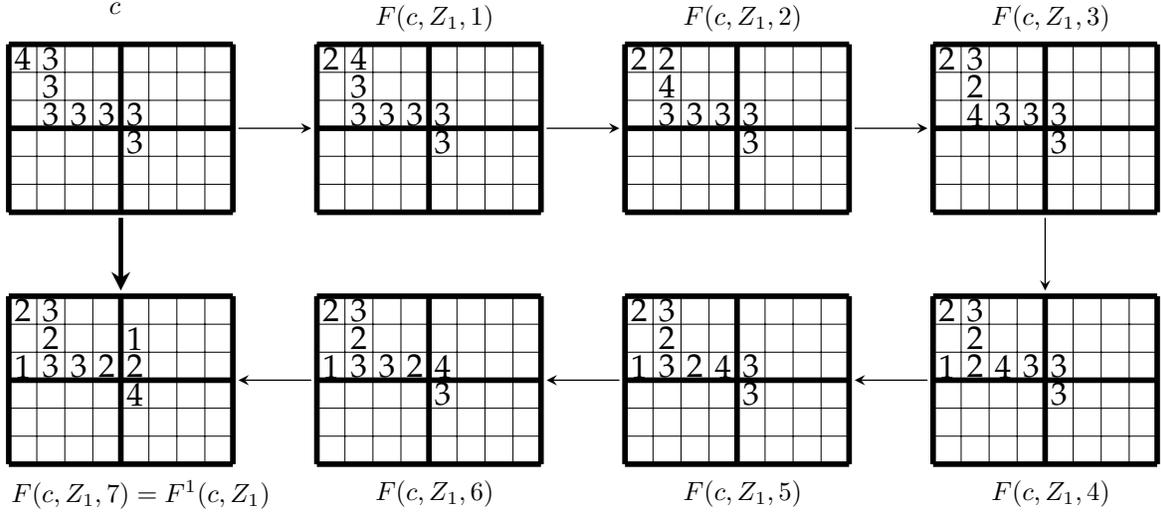
\begin{figure}
\centering
\resizebox{\textwidth}{!}{%
\begin{tikzpicture}[scale=.3]

    \draw[-{Stealth[width=2mm,length=2mm]}, line width=0.5mm] (4,8.8) -- (4,6.2);

    \draw[line width=0.1mm] (0, 9) grid (8, 15);
    \draw[line width=0.6mm] (0,9) -- (0,15); 
    \draw[line width=0.6mm] (4,9) -- (4,15); 
    \draw[line width=0.6mm] (8,9) -- (8,15); 
    \draw[line width=0.6mm] (0,9) -- (8,9);
    \draw[line width=0.6mm] (0,12) -- (8,12);
    \draw[line width=0.6mm] (0,15) -- (8,15);
    
    \setcounter{rows}{15}
    \setcounter{row}{0}
    \setrow {4,3, , , }
    \setrow { ,3, , , }
    \setrow { ,3,3,3,3}
    \setrow { , , , ,3};
    
    \draw[- stealth] (8.2,12) -- (10.8,12);

    \draw[line width=0.1mm] (11, 9) grid (19, 15);
    \draw[line width=0.6mm] (11,9) -- (11,15); 
    \draw[line width=0.6mm] (15,9) -- (15,15); 
    \draw[line width=0.6mm] (19,9) -- (19,15); 
    \draw[line width=0.6mm] (11,9) -- (19,9);
    \draw[line width=0.6mm] (11,12) -- (19,12);
    \draw[line width=0.6mm] (11,15) -- (19,15);

    \setcounter{cols}{11}
    \setcounter{rows}{15}
    \setcounter{row}{0}
    \setrow {2,4, , , }
    \setrow { ,3, , , }
    \setrow { ,3,3,3,3}
    \setrow { , , , ,3};

    \draw[- stealth] (19.2,12) -- (21.8,12);

    \draw[line width=0.1mm] (22, 9) grid (30, 15);
    \draw[line width=0.6mm] (22,9) -- (22,15); 
    \draw[line width=0.6mm] (26,9) -- (26,15); 
    \draw[line width=0.6mm] (30,9) -- (30,15); 
    \draw[line width=0.6mm] (22,9) -- (30,9);
    \draw[line width=0.6mm] (22,12) -- (30,12);
    \draw[line width=0.6mm] (22,15) -- (30,15);

    \setcounter{cols}{22}
    \setcounter{rows}{15}
    \setcounter{row}{0}
    \setrow {2,2, , , }
    \setrow { ,4, , , }
    \setrow { ,3,3,3,3}
    \setrow { , , , ,3};

    \draw[- stealth] (30.2,12) -- (32.8,12);

    \draw[line width=0.1mm] (33, 9) grid (41, 15);
    \draw[line width=0.6mm] (33,9) -- (33,15); 
    \draw[line width=0.6mm] (37,9) -- (37,15); 
    \draw[line width=0.6mm] (41,9) -- (41,15); 
    \draw[line width=0.6mm] (33,9) -- (41,9);
    \draw[line width=0.6mm] (33,12) -- (41,12);
    \draw[line width=0.6mm] (33,15) -- (41,15);

    \setcounter{cols}{33}
    \setcounter{rows}{15}
    \setcounter{row}{0}
    \setrow {2,3, , , }
    \setrow { ,2, , , }
    \setrow { ,4,3,3,3}
    \setrow { , , , ,3};

    \draw[- stealth] (37,8.8) -- (37,6.2);

    \draw[line width=0.1mm] (0, 0) grid (8, 6);
    \draw[line width=0.6mm] (0,0) -- (0,6); 
    \draw[line width=0.6mm] (4,0) -- (4,6); 
    \draw[line width=0.6mm] (8,0) -- (8,6); 
    \draw[line width=0.6mm] (0,0) -- (8,0);
    \draw[line width=0.6mm] (0,3) -- (8,3);
    \draw[line width=0.6mm] (0,6) -- (8,6);

    \setcounter{cols}{33}
    \setcounter{rows}{6}
    \setcounter{row}{0}
    \setrow {2,3, , , }
    \setrow { ,2, , , }
    \setrow {1,2,4,3,3}
    \setrow { , , , ,3};

    \draw[- stealth] (32.8,3) -- (30.2,3);

    \draw[line width=0.1mm] (11, 0) grid (19, 6);
    \draw[line width=0.6mm] (11,0) -- (11,6); 
    \draw[line width=0.6mm] (15,0) -- (15,6); 
    \draw[line width=0.6mm] (19,0) -- (19,6); 
    \draw[line width=0.6mm] (11,0) -- (19,0);
    \draw[line width=0.6mm] (11,3) -- (19,3);
    \draw[line width=0.6mm] (11,6) -- (19,6);

    \setcounter{cols}{22}
    \setcounter{rows}{6}
    \setcounter{row}{0}
    \setrow {2,3, , , }
    \setrow { ,2, , , }
    \setrow {1,3,2,4,3}
    \setrow { , , , ,3};

    \draw[- stealth] (21.8,3) -- (19.2,3);

    \draw[line width=0.1mm] (22, 0) grid (30, 6);
    \draw[line width=0.6mm] (22,0) -- (22,6); 
    \draw[line width=0.6mm] (26,0) -- (26,6); 
    \draw[line width=0.6mm] (30,0) -- (30,6); 
    \draw[line width=0.6mm] (22,0) -- (30,0);
    \draw[line width=0.6mm] (22,3) -- (30,3);
    \draw[line width=0.6mm] (22,6) -- (30,6);

    \setcounter{cols}{11}
    \setcounter{rows}{6}
    \setcounter{row}{0}
    \setrow {2,3, , , }
    \setrow { ,2, , , }
    \setrow {1,3,3,2,4}
    \setrow { , , , ,3};

    \draw[- stealth] (10.8,3) -- (8.2,3);

    \draw[line width=0.1mm] (33, 0) grid (41, 6);
    \draw[line width=0.6mm] (33,0) -- (33,6); 
    \draw[line width=0.6mm] (37,0) -- (37,6); 
    \draw[line width=0.6mm] (41,0) -- (41,6); 
    \draw[line width=0.6mm] (33,0) -- (41,0);
    \draw[line width=0.6mm] (33,3) -- (41,3);
    \draw[line width=0.6mm] (33,6) -- (41,6);

    \setcounter{cols}{0}
    \setcounter{rows}{6}
    \setcounter{row}{0}
    \setrow {2,3, , , }
    \setrow { ,2, , ,1}
    \setrow {1,3,3,2,2}
    \setrow { , , , ,4};

\setcounter{cols}{0}

    \draw (3.5,16.6) node [anchor=north west,scale=0.8][inner sep=0.75pt]    {$c$};
     
    \draw (0,-0.5) node [anchor=north west,scale=0.8][inner sep=0.75pt]    {$F(c,Z_1,7) = F^1(c,Z_1)$};
   
    \draw (13,-0.5) node [anchor=north west,scale=0.8][inner sep=0.75pt]    {$F(c,Z_1,6)$};

    \draw (24,-0.5) node [anchor=north west,scale=0.8][inner sep=0.75pt]    {$F(c,Z_1,5)$};
  
    \draw (35,-0.5) node [anchor=north west,scale=0.8][inner sep=0.75pt]    {$F(c,Z_1,4)$};

    \draw (35,16.5) node [anchor=north west,scale=0.8][inner sep=0.75pt]    {$F(c,Z_1,3)$};

    \draw (13,16.5) node [anchor=north west,scale=0.8][inner sep=0.75pt]    {$F(c,Z_1,1)$};

    \draw (24,16.5) node [anchor=north west,scale=0.8][inner sep=0.75pt]    {$F(c,Z_1,2)$};
    
\end{tikzpicture}
}
\caption{Updates according to the word $Z_1 = HVVHHHV \in {H,V}^*$. The initial configuration $c$ is shown in the top-left. The arrows, starting to the right, indicate the next configuration after each corresponding update. The configuration immediately below the initial one corresponds to the state after completing the entire cycle of the word.}
\label{fig:ciclo}
\end{figure}


Let $Z  \in \{H,V\}^k$. 
We abuse notation and define the global transition function induced by $Z$ as the
function $F(c, Z,s) = F(F(c, Z,s-1), Z(s))$ where $F(c,Z,0) = c$. 
The sequence of configurations $F(c,Z,0), \dots, F(c,Z,k)$ is a \emph{cycle} of
$Z$. 
In addition, we denote $F(c,Z,k)$ by simply $F(c,Z)$ and $F(c,Z,ik)$ by
$F^i(c,Z)$.
In \cref{fig:ciclo} is an example of a complete cycle. 
A fungal automaton is a tuple $\mathcal{F} = (Q,N,Z,F(\cdot,Z))$ where $Q, N,Z$ and $F(\cdot,Z))$ are defined as above.


We write $Z \in \{H,V\}^{h,v,k}$ whenever $Z$ is such that the symbol $H$
appears $h$ times and $V$ appears $v$ times in $Z$. 
In the rest of the paper we sometimes use the identification $H = (1,0)$ and $V
= (0,1)$ whenever the context is clear and write $|w|_a := |\{i : w(i) = a\}|$
for any $w \in \{H,V\}^*, a \in \{H,V\}$. 

For any configuration $c \in Q^{\Z^2}$, a cell $x \in \Z^2$ with $c(x) \geq 4$ is called a \emph{signal}.

\section{Circuit embedding} \label{sec:main}

We start by describing the main structure on which our construction is based. This structure is called a block. Let $Z \in \{H,V\}^{h,v,k}$ a word.  A $Z$-block $B$ is a grid of size $h\times v.$ To avoid cluttering the notation, we refer to them as blocks when the context is clear.

Given a $Z$-block $B$, $B^{+}$ and $B^{-}$ are the cells in the upper- and
lower-left corner of $B$, respectively. 
More precisely, if $x=(a,b)$ are the coordinates of the cell in the upper-left
corner of $B$ then $B^+ = x$ and $B^- = x + (0,(v-1))$. 
We call $B^+$ (resp., $B^-$) the \emph{positive} (resp., \emph{negative})
\emph{source cell} of $B$.


\begin{figure}
\centering
\resizebox{.7\textwidth}{!}{%
\begin{tikzpicture}[scale=.4]

    \draw  [draw opacity=0][fill={rgb, 255:red, 184; green, 233; blue, 134 }  ,fill opacity=0.43 ] (0,0) -- (4,0) -- (4,3) -- (0,3) -- cycle ;
    \draw  [draw opacity=0][fill={rgb, 255:red, 184; green, 233; blue, 134 }  ,fill opacity=0.43 ] (8,0) -- (12,0) -- (12,3) -- (8,3) -- cycle ;
    \draw  [draw opacity=0][fill={rgb, 255:red, 184; green, 233; blue, 134 }  ,fill opacity=0.43 ] (0,6) -- (4,6) -- (4,9) -- (0,9) -- cycle ;
    \draw  [draw opacity=0][fill={rgb, 255:red, 184; green, 233; blue, 134 }  ,fill opacity=0.43 ] (8,6) -- (12,6) -- (12,9) -- (8,9) -- cycle ;
    \draw  [draw opacity=0][fill={rgb, 255:red, 184; green, 233; blue, 134 }  ,fill opacity=0.43 ] (4,3) -- (8,3) -- (8,6) -- (4,6) -- cycle ;

    \draw[line width=0.1mm] (0, 0) grid (12, 9);
    \draw[line width=0.6mm] (0,0) -- (0,9);
    \draw[line width=0.6mm] (4,0) -- (4,9); 
    \draw[line width=0.6mm] (8,0) -- (8,9); 
    \draw[line width=0.6mm] (12,0) -- (12,9); 
    \draw[line width=0.6mm] (0,0) -- (12,0);
    \draw[line width=0.6mm] (0,3) -- (12,3);
    \draw[line width=0.6mm] (0,6) -- (12,6);
    \draw[line width=0.6mm] (0,9) -- (12,9);
    
    \setcounter{cols}{0}
    \setcounter{rows}{9}
    \setcounter{row}{0}
    \setrow {\small{$+$}, , , ,\small{$+$}, , , ,\small{$+$}}
    \setrow { , , , }
    \setrow {\small{$-$}, , , ,\small{$-$}, , , ,\small{$-$}}
    \setrow {\small{$+$}, , , ,\small{$+$}, , , ,\small{$+$}}
    \setrow { , , , }
    \setrow {\small{$-$}, , , ,\small{$-$}, , , ,\small{$-$}}
    \setrow {\small{$+$}, , , ,\small{$+$}, , , ,\small{$+$}}
    \setrow { , , , }
    \setrow {\small{$-$}, , , ,\small{$-$}, , , ,\small{$-$}}
    ;
    
\setcounter{cols}{0}

    \draw  [draw opacity=0][fill={rgb, 255:red, 0; green, 0; blue,  255 }  ,fill opacity=0.2 ] (17,2) -- (23,2) -- (23,7) -- (17,7) -- cycle ;

    \draw[line width=0.1mm] (14, 0) grid (26, 9);
    \draw[line width=0.6mm] (14,0) -- (14,9);
    \draw[line width=0.6mm] (18,0) -- (18,9); 
    \draw[line width=0.6mm] (22,0) -- (22,9); 
    \draw[line width=0.6mm] (26,0) -- (26,9); 
    \draw[line width=0.6mm] (14,0) -- (26,0);
    \draw[line width=0.6mm] (14,3) -- (26,3);
    \draw[line width=0.6mm] (14,6) -- (26,6);
    \draw[line width=0.6mm] (14,9) -- (26,9);
    
    \setcounter{cols}{14}
    \setcounter{rows}{9}
    \setcounter{row}{0}
    \setrow { , , , , , , , , }
    \setrow { , , , }
    \setrow { , , , , , , , , }
    \setrow { , , , ,\small{$B_+$}, , , , }
    \setrow { , , , }
    \setrow { , , , ,\small{$B_-$}, , , , }
    \setrow { , , , , , , , , }
    \setrow { , , , }
    \setrow { , , , , , , , , }
    \setcounter{cols}{0};

\end{tikzpicture}
}
\caption{On the left, blocks for $Z = HVVHHHV \in \{H,V\}^{4,3,7}$ divided with
thick black lines and with respective positive and negative sources marked with
$+,-$. 
In green are marked diagonally connected blocks (white $Z$-blocks are diagonally
connected as well). 
On the right, letting $B$ be the central block, $\overline{B}$ is marked in
blue.}
\label{fig:block}
\end{figure}
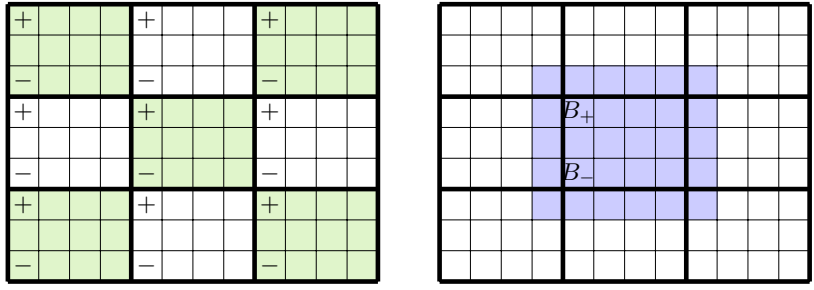

A $Z$-block subdivision of the grid is a subdivision of the grid in $Z$-blocks
as in \Cref{fig:block} for the word $Z = HVVHHHV \in \{H, V\}^{4,3,7}$. 
Also, observe that in \Cref{fig:ciclo} there are $4$ blocks forming a $2 \times
2$ shape for the word $Z_1$. 
In a $Z$-block subdivision, we say two blocks $B_1,B_2$ are \emph{vertically
adjacent} when $B_1$ is on top of $B_2$, that is, if $B_2^+ = B_1^++vV$ holds;
similarly, we say they are \emph{horizontally adjacent} when $B_1$ is at the
left of $B_2$, that is, if $B_2^+ = B_1^++hH$. 
Two blocks $B_1,B_2$ are \emph{diagonally connected} if there exists a
block $B$ horizontally adjacent to $B_1$ and vertically adjacent to $B_2$. 
In \Cref{fig:block}, green blocks are diagonally connected. 



Let $Z \in \{H,V\}^{h,v,k}$ and $B$ be a $Z$-block. The \emph{closure}
of $B$ is the set of cells $\overline{B}$ containing $B$ and all its vertical,
horizontal and diagonal neighbors, formally $\overline{B} = B \cup \{y \in \Z^2:
y = x \pm H \pm V, x \in B\}.$ In \Cref{fig:block}, the closure of the center block
are the cells marked in blue.

\begin{definition} \label{def:defbridge}
    Given $Z \in \{H,V\}^{h,v,k}$, we say $P=(P(0),P(1),\hdots,P(l))$ where
    $P(i) \in \Z^2$ and $l\leq k$ is a $Z$-path there exists $\alpha,\beta \in
    \{-1,1\}$ such that:
\[
P(s) = x_s = \begin{cases}
     x \in \Z^2 & \text{ if } s = 0,\\
     \alpha Z(s) [Z(s) = H]  + \beta Z(s) [Z(s) = V] + P(s-1) & \text{otherwise.}
\end{cases}
\]
\end{definition}

Observe that a $Z$-path is characterized only by its starting cell $P(0)$ and $\alpha,\beta$. We say that a $Z$-path $P$ connects two cells $x,y \in \Z^2$ if $P(0) = x$ and $P(l)=y.$

\begin{lemma}\label{rem:pathconect} Let $Z \in \{H,V\}^{h,v,k}$ and consider diagonally connected $Z$-blocks $B_1,B_2.$ For each $x,y\in \{(B_1^+,B_2^+), (B_1^-,B_2^-)\}$, there exists a $Z$-path $P = \{P(0),...,P(k)\}$ connecting $x$ and $y.$
\end{lemma}

\begin{proof}
Let $x,y$. We need to show that there exists a $Z$-path $P = \{P(0),...,P(k)\}$ starting from $P(0)=x$ such that $P(k) =y.$ In other words, since $Z$-paths can be defined by $P(0)$ and some $\alpha,\beta$ we need to show that it is possible to choose $\alpha,\beta \in \{-1,1\}$ such that $P(k) = y$. 

First, observe that for every $Z$-path $P=(P(0),\hdots,P(l))$, we have that: 
\[
  P(k) = P(0)+\sum_{i=0}^l Z(i)(\alpha[Z(i) = H] + \beta[Z(i) = V]).
\]
On the other hand, we have that cells $x$ and $y$ satisfy
$y=x + \alpha'h H  + \beta'v V$
for some $\alpha,\beta\in \{-1,1\}$.
Choosing $\alpha,\beta = \alpha',\beta'$, it follows
\begin{align*}
  P(k) &=  P(0)+\sum_{i=0}^k Z(i)(\alpha[Z(i) = H] + \beta[Z(i) = V]) \\
  &=  x +\alpha h H  + \beta v V = x +\alpha' h H  + \beta' v V \\
  &=  y.
\end{align*}
Thus, $P$ connects $x$ and $y.$ The lemma holds.
\end{proof}



\paragraph{Restricting update schemes.}
During the proof, we assume that the update scheme $Z$ starts with \(H\) and
ends with \(V\).
We argue we can make this assumption without restriction since there are two
possibilities for $Z$ to deviate from this scheme:
\begin{enumerate}
  \item If $Z$ starts with \(V\) and ends with \(H\), then we can perform the
  same construction rotated 90 degrees.
  \item If $Z$ starts and ends with the same symbol, then its behavior is
  equivalent to that of a word that starts and ends with distinct ones.
  For instance, let \(Z = V^iHWV^j\) with \(W \in \{H,V\}^*\) and \(i, j \in
  \mathbb{N}\).
  Then notice that, when executing \(Z\), the dynamics after $i$ initial steps
  behave as if we were working with \(Z' = HWV^{j+i}\).
  Hence we need only add $i$ \enquote{waiting} steps at the beginning of the
  construction.
  This can be done simply by extending source bridges (Defined in the following section) by a wire that is $i$ cells long
  (and carries the initial signal in the correct direction).
  See \cref{fig:words} for an illustration.
\end{enumerate}

In addition, we can also assume that $Z$ contains at least two $H$ and two $V$
symbols each.
If it contains, say, a single $V$, then we can obtain the result by directly
adapting the construction from \cite{modanese2024embedding} with a similar
strategy based on \enquote{waiting} steps.
Let $i$ be the number of occurrences of $H$ in $Z$, that is, $Z = H^iV$.
Then, for every $2 \times 2$ block in the original construction, we maintain a
$(i+1) \times 2$ block where the $i$ additional columns are simply carrying the
signal from the leftmost to the rightmost column during the intermediate
applications of $H$.
(Alternatively, our construction also works for the case of a single $H$ or a
single $V$, but the terminology is then degenerate; e.g., then the positive and
negative sources of a block are actually the same cell.)
\fxnote{Here we could also use an illustration}

\begin{figure}[t]
\centering
\resizebox{.5\textwidth}{!}{%
\input{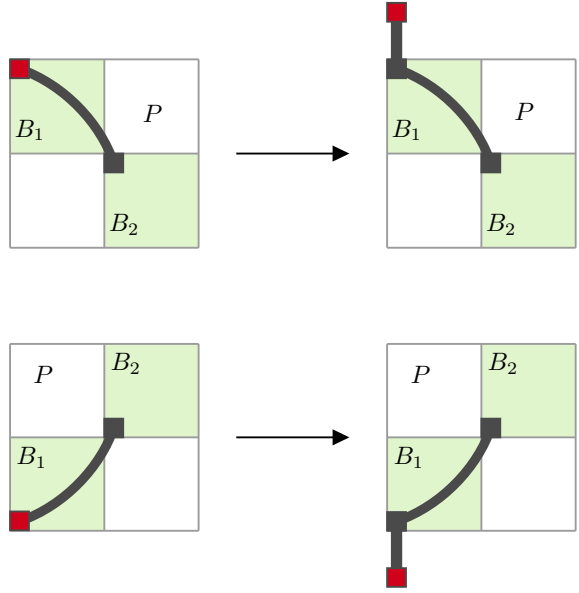}
}
\caption{Source bridges modification. At each source cell, we concatenate cells with value 3 either above or below it. This ensures that a signal reaches the source cell of the block, allowing it to start functioning as intended for the new, \emph{non-shifted} word.}
\label{fig:words}
\end{figure}

\section{Bridges} \label{sec:bridges}
Now we introduce the main element of our construction: \emph{bridges}.
From now on, for every word $Z$ we work with a $Z$-subdivision of the grid.

For a fixed word $Z$, a bridge is a structure defined by a particular $Z$-path and a configuration $c.$ Intuitively, the path marks a trail of cells with $3$-grains. These paths connect  positive (or negative) source cells of diagonal $Z$-blocks.
Thus, if we left $4$ grains in the first cell of the path, that will trigger an avalanche that will traverse the four blocks. We show an example for a word in \Cref{fig:ciclo}.



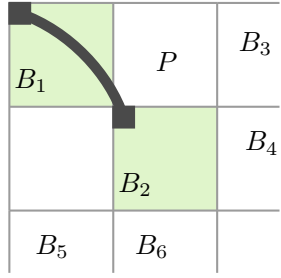
\begin{figure}[t]
\centering
\resizebox{.25\textwidth}{!}{%
\tikzset{every picture/.style={line width=0.75pt}} 

\begin{tikzpicture}[x=0.75pt,y=0.75pt,yscale=-1,xscale=1]

\draw  [draw opacity=0][fill={rgb, 255:red, 184; green, 233; blue, 134 }  ,fill opacity=0.43 ] (70,70) -- (120,70) -- (120,120) -- (70,120) -- cycle ;
\draw  [draw opacity=0][fill={rgb, 255:red, 184; green, 233; blue, 134 }  ,fill opacity=0.43 ] (20,20) -- (70,20) -- (70,70) -- (20,70) -- cycle ;
\draw [color={rgb, 255:red, 155; green, 155; blue, 155 }  ,draw opacity=1 ][line width=0.75]    (70,20) -- (70,150) ;
\draw [color={rgb, 255:red, 155; green, 155; blue, 155 }  ,draw opacity=1 ][line width=0.75]    (20,20) -- (20,120) ;
\draw [color={rgb, 255:red, 155; green, 155; blue, 155 }  ,draw opacity=1 ][line width=0.75]    (120,20) -- (120,120) ;
\draw [color={rgb, 255:red, 155; green, 155; blue, 155 }  ,draw opacity=1 ][line width=0.75]    (120,20) -- (20,20) ;
\draw [color={rgb, 255:red, 155; green, 155; blue, 155 }  ,draw opacity=1 ][line width=0.75]    (150,70) -- (20,70) ;
\draw [color={rgb, 255:red, 155; green, 155; blue, 155 }  ,draw opacity=1 ][line width=0.75]    (120,120) -- (20,120) ;
\draw  [color={rgb, 255:red, 74; green, 74; blue, 74 }  ,draw opacity=1 ][fill={rgb, 255:red, 74; green, 74; blue, 74 }  ,fill opacity=1 ] (70,70) -- (80,70) -- (80,80) -- (70,80) -- cycle ;
\draw [color={rgb, 255:red, 74; green, 74; blue, 74 }  ,draw opacity=1 ][line width=3.75]    (25,25) .. controls (49.2,33.5) and (68.2,54.5) .. (75,75) ;
\draw  [color={rgb, 255:red, 74; green, 74; blue, 74 }  ,draw opacity=1 ][fill={rgb, 255:red, 74; green, 74; blue, 74 }  ,fill opacity=1 ] (20,20) -- (30,20) -- (30,30) -- (20,30) -- cycle ;
\draw [color={rgb, 255:red, 155; green, 155; blue, 155 }  ,draw opacity=1 ][line width=0.75]    (120,120) -- (120,150) ;
\draw [color={rgb, 255:red, 155; green, 155; blue, 155 }  ,draw opacity=1 ][line width=0.75]    (120,120) -- (150,120) ;
\draw [color={rgb, 255:red, 155; green, 155; blue, 155 }  ,draw opacity=1 ][line width=0.75]    (120,20) -- (150,20) ;
\draw [color={rgb, 255:red, 155; green, 155; blue, 155 }  ,draw opacity=1 ][line width=0.75]    (20,120) -- (20,150) ;

\draw (71,100.4) node [anchor=north west][inner sep=0.75pt]    {$B_{2}$};
\draw (21,50.4) node [anchor=north west][inner sep=0.75pt]    {$B_{1}$};
\draw (89,42.4) node [anchor=north west][inner sep=0.75pt]    {$P$};
\draw (129,32.4) node [anchor=north west][inner sep=0.75pt]    {$B_{3}$};
\draw (132,80.4) node [anchor=north west][inner sep=0.75pt]    {$B_{4}$};
\draw (31,130.4) node [anchor=north west][inner sep=0.75pt]    {$B_{5}$};
\draw (79,130.4) node [anchor=north west][inner sep=0.75pt]    {$B_{6}$};

\end{tikzpicture}
}
\caption{Bridge notation for a positive bridge $T = (P, c)$ connecting blocks
$B_1$ and $B_2$, which are highlighted in green. 
Light gray lines indicate the division between $Z$-blocks while the bridge is
marked in dark gray.
The squares refer to the cells $P(0)$ and $P(k)$. 
The concavity of the bridge is intentional, highlighting that $T$ starts with a
horizontal step and ends with a vertical one.
The $Z$-Blocks $B_i$ with $i \in \{3,4,5,6\}$, are relevant for
\Cref{remarkfunc}.}
\label{fig:notpuente}
\end{figure}

Formally, given 2 blocks $B_1,B_2$ diagonally connected (notation example in \Cref{fig:notpuente}) we define a $Z$-bridge for a word $Z \in \{H,V\}^{h,v,k}$ as the tuple $T = (P,c)$ where $P = (P(0), \hdots, P(k))$ is a $Z$-path connecting $B_1^+,B_2^+$ or $B_1^-,B_2^-$ and $c$ is a configuration such that $c(x) = \begin{cases}
    3 & \text{ if } x \in P, \\
    0 & \text{ otherwise. }
\end{cases}$

Instead, if $c(P(0)) = 4$, we call $T$ a source $Z$-bridge. And if for some $l \in [k], c(P(l)) = 2$ and $c(P(j)) = 0, \forall j \geq l$ , we say $T$ is a bridge with a sink at $l$. 

This is well-defined since, by \Cref{rem:pathconect}, a path can always be defined that connects the diagonally connected blocks $B_1$ and $B_2$. A bridge $T$ connecting positive cells, i.e. when $P(0)$ and $P(k)$ both ar positive source, will be referred to as a \emph{positive} bridge. 
In the same way, if $P$ is connecting negative source cells, i.e. $P(0)$ and $P(k)$ are negative sources, it is called a \emph{negative} bridge. 

Now, we show that bridges can transmit signals. This is a key element of our construction. The idea is that a signal can go through the bridge following the structure of the $Z$-path contained in $T$. Furthermore, as illustrated in \Cref{fig:bridge_basura}, even when the bridges are surrounded by cells with values lower than 4, the bridge will still function to transmit the signal.



\begin{lemma}\label{lem:funcbridge}
    Let $Z \in \{H,V\}^{h,v,k}$ and let us consider a source $Z$-bridge $T =(P,c)$ connecting $B_1^+,B_2^+$ from two diagonally connected blocks $B_1,B_2$, with $P = \{P(0),...,P(k)\}$.  If we consider the configuration where every cell outside the bridge has value less or equal to 3, i.e if we define $\Tilde{c}(x) = \begin{cases}
        c(x) & \text{ if } x \in P, \\
        s(x) & \text{ if } x \in \overline{B_1}\setminus P\end{cases}$ where $s$ is any configuration with values strictly less than 4, then we have that $F(\Tilde{c}(x),Z)(P(k))=4.$
\end{lemma}

\begin{figure}[t]
\centering
\begin{tikzpicture}[scale=.29]

    \draw  [draw opacity=0][fill={rgb, 255:red, 255; green, 0; blue,  0 }  ,fill opacity=0.1 ] (0,15) -- (1,15) -- (1,14) -- (0,14) -- cycle ;

    \draw  [draw opacity=0][fill={rgb, 255:red, 255; green, 0; blue,  0 }  ,fill opacity=0.1 ] (11,15) -- (13,15) -- (13,14) -- (11,14) -- cycle ;

    \draw  [draw opacity=0][fill={rgb, 255:red, 255; green, 0; blue,  0 }  ,fill opacity=0.1 ] (22,15) -- (24,15) -- (24,13) -- (22,13) -- cycle ;

    \draw  [draw opacity=0][fill={rgb, 255:red, 255; green, 0; blue,  0 }  ,fill opacity=0.1 ] (33,15) -- (35,15) -- (35,12) -- (33,12) -- cycle ;

    \draw  [draw opacity=0][fill={rgb, 255:red, 255; green, 0; blue,  0 }  ,fill opacity=0.1 ] (33,6) -- (36,6) -- (36,3) -- (33,3) -- cycle ;

    \draw  [draw opacity=0][fill={rgb, 255:red, 255; green, 0; blue,  0 }  ,fill opacity=0.1 ] (22,6) -- (26,6) -- (26,3) -- (22,3) -- cycle ;
    
    \draw  [draw opacity=0][fill={rgb, 255:red, 255; green, 0; blue,  0 }  ,fill opacity=0.1 ] (11,6) -- (16,6) -- (16,3) -- (11,3) -- cycle ;

    \draw  [draw opacity=0][fill={rgb, 255:red, 255; green, 0; blue,  0 }  ,fill opacity=0.1 ] (0,6) -- (5,6) -- (5,2) -- (0,2) -- cycle ;

    \draw [draw opacity=0][fill={rgb, 255:red, 0; green, 0; blue, 255}, fill opacity=0.2] (0,14) -- (2,14) -- (2,15) -- (0,15) -- cycle;
    
    \draw [draw opacity=0][fill={rgb, 255:red, 0; green, 0; blue, 255}, fill opacity=0.2] (1,14) -- (1,12) -- (2,12) -- (2,14) -- cycle;
    
    \draw [draw opacity=0][fill={rgb, 255:red, 0; green, 0; blue, 255}, fill opacity=0.2] (2,13) -- (5,13) -- (5,12) -- (2,12) -- cycle;
    
    \draw [draw opacity=0][fill={rgb, 255:red, 0; green, 0; blue, 255}, fill opacity=0.2] (4,12) -- (5,12) -- (5,11) -- (4,11) -- cycle;

    \draw  [draw opacity=0][fill={rgb, 255:red, 255; green, 0; blue,  0 }  ,fill opacity=0.5 ] (0,15) -- (1,15) -- (1,14) -- (0,14) -- cycle;
    \draw [draw opacity=0, fill={rgb, 255:red, 255; green, 0; blue, 0}, fill opacity=0.5] (4,3) -- (5,3) -- (5,2) -- (4,2) -- cycle;
    \draw [draw opacity=0, fill={rgb, 255:red, 255; green, 0; blue, 0}, fill opacity=0.5] (12,15) -- (13,15) -- (13,14) -- (12,14) -- cycle;
    \draw [draw opacity=0, fill={rgb, 255:red, 255; green, 0; blue, 0}, fill opacity=0.5] (15,4) -- (16,4) -- (16,3) -- (15,3) -- cycle;
    \draw [draw opacity=0, fill={rgb, 255:red, 255; green, 0; blue, 0}, fill opacity=0.5] (23,14) -- (24,14) -- (24,13) -- (23,13) -- cycle;
    \draw [draw opacity=0, fill={rgb, 255:red, 255; green, 0; blue, 0}, fill opacity=0.5] (25,4) -- (26,4) -- (26,3) -- (25,3) -- cycle;
    \draw [draw opacity=0, fill={rgb, 255:red, 255; green, 0; blue, 0}, fill opacity=0.5] (34,13) -- (35,13) -- (35,12) -- (34,12) -- cycle;
    \draw [draw opacity=0, fill={rgb, 255:red, 255; green, 0; blue, 0}, fill opacity=0.5] (35,4) -- (36,4) -- (36,3) -- (35,3) -- cycle;
    
    \draw[-{Stealth[width=2mm,length=2mm]}, line width=0.5mm] (4,8.8) -- (4,6.2);

    \draw[line width=0.1mm] (0, 9) grid (8, 15);
    \draw[line width=0.6mm] (0,9) -- (0,15); 
    \draw[line width=0.6mm] (4,9) -- (4,15); 
    \draw[line width=0.6mm] (8,9) -- (8,15); 
    \draw[line width=0.6mm] (0,9) -- (8,9);
    \draw[line width=0.6mm] (0,12) -- (8,12);
    \draw[line width=0.6mm] (0,15) -- (8,15);
    
    \setcounter{rows}{15}
    \setcounter{row}{0}
    \setrow {4,3,3,2,3}
    \setrow {1,3,3,3, }
    \setrow {3,3,3,3,3}
    \setrow { ,3,2,3,3};
    
    \draw[- stealth] (8.2,12) -- (10.8,12);

    \draw[line width=0.1mm] (11, 9) grid (19, 15);
    \draw[line width=0.6mm] (11,9) -- (11,15); 
    \draw[line width=0.6mm] (15,9) -- (15,15); 
    \draw[line width=0.6mm] (19,9) -- (19,15); 
    \draw[line width=0.6mm] (11,9) -- (19,9);
    \draw[line width=0.6mm] (11,12) -- (19,12);
    \draw[line width=0.6mm] (11,15) -- (19,15);

    \setcounter{cols}{11}
    \setcounter{rows}{15}
    \setcounter{row}{0}
    \setrow {2,4,3,2,3}
    \setrow {1,3,3,3, }
    \setrow {3,3,3,3,3}
    \setrow { ,3,2,3,3};

    \draw[- stealth] (19.2,12) -- (21.8,12);

    \draw[line width=0.1mm] (22, 9) grid (30, 15);
    \draw[line width=0.6mm] (22,9) -- (22,15); 
    \draw[line width=0.6mm] (26,9) -- (26,15); 
    \draw[line width=0.6mm] (30,9) -- (30,15); 
    \draw[line width=0.6mm] (22,9) -- (30,9);
    \draw[line width=0.6mm] (22,12) -- (30,12);
    \draw[line width=0.6mm] (22,15) -- (30,15);

    \setcounter{cols}{22}
    \setcounter{rows}{15}
    \setcounter{row}{0}
    \setrow {2,2,3,2,3}
    \setrow {1,4,3,3, }
    \setrow {3,3,3,3,3}
    \setrow { ,3,2,3,3};

    \draw[- stealth] (30.2,12) -- (32.8,12);

    \draw[line width=0.1mm] (33, 9) grid (41, 15);
    \draw[line width=0.6mm] (33,9) -- (33,15); 
    \draw[line width=0.6mm] (37,9) -- (37,15); 
    \draw[line width=0.6mm] (41,9) -- (41,15); 
    \draw[line width=0.6mm] (33,9) -- (41,9);
    \draw[line width=0.6mm] (33,12) -- (41,12);
    \draw[line width=0.6mm] (33,15) -- (41,15);

    \setcounter{cols}{33}
    \setcounter{rows}{15}
    \setcounter{row}{0}
    \setrow {2,3,3,2,3}
    \setrow {1,2,3,3, }
    \setrow {3,4,3,3,3}
    \setrow { ,3,2,3,3};

    \draw[- stealth] (37,8.8) -- (37,6.2);

    \draw[line width=0.1mm] (0, 0) grid (8, 6);
    \draw[line width=0.6mm] (0,0) -- (0,6); 
    \draw[line width=0.6mm] (4,0) -- (4,6); 
    \draw[line width=0.6mm] (8,0) -- (8,6); 
    \draw[line width=0.6mm] (0,0) -- (8,0);
    \draw[line width=0.6mm] (0,3) -- (8,3);
    \draw[line width=0.6mm] (0,6) -- (8,6);

    \setcounter{cols}{33}
    \setcounter{rows}{6}
    \setcounter{row}{0}
    \setrow {2,3,3,2,3}
    \setrow {1,2,3,3, }
    \setrow {4,2,4,3,3}
    \setrow { ,3,2,3,3};

    \draw[- stealth] (32.8,3) -- (30.2,3);

    \draw[line width=0.1mm] (11, 0) grid (19, 6);
    \draw[line width=0.6mm] (11,0) -- (11,6); 
    \draw[line width=0.6mm] (15,0) -- (15,6); 
    \draw[line width=0.6mm] (19,0) -- (19,6); 
    \draw[line width=0.6mm] (11,0) -- (19,0);
    \draw[line width=0.6mm] (11,3) -- (19,3);
    \draw[line width=0.6mm] (11,6) -- (19,6);

    \setcounter{cols}{22}
    \setcounter{rows}{6}
    \setcounter{row}{0}
    \setrow {2,3,3,2,3}
    \setrow {1,2,3,3, }
    \setrow {2,4,2,4,3}
    \setrow { ,3,2,3,3};
    
    \draw[- stealth] (21.8,3) -- (19.2,3);

    \draw[line width=0.1mm] (22, 0) grid (30, 6);
    \draw[line width=0.6mm] (22,0) -- (22,6); 
    \draw[line width=0.6mm] (26,0) -- (26,6); 
    \draw[line width=0.6mm] (30,0) -- (30,6); 
    \draw[line width=0.6mm] (22,0) -- (30,0);
    \draw[line width=0.6mm] (22,3) -- (30,3);
    \draw[line width=0.6mm] (22,6) -- (30,6);

    \setcounter{cols}{11}
    \setcounter{rows}{6}
    \setcounter{row}{0}
    \setrow {2,3,3,2,3}
    \setrow {1,2,3,3, }
    \setrow {3,2,4,2,4}
    \setrow { ,3,2,3,3};

    \draw[- stealth] (10.8,3) -- (8.2,3);

   \draw[line width=0.1mm] (33, 0) grid (41, 6);
    \draw[line width=0.6mm] (33,0) -- (33,6); 
    \draw[line width=0.6mm] (37,0) -- (37,6); 
    \draw[line width=0.6mm] (41,0) -- (41,6); 
    \draw[line width=0.6mm] (33,0) -- (41,0);
    \draw[line width=0.6mm] (33,3) -- (41,3);
    \draw[line width=0.6mm] (33,6) -- (41,6);

    \setcounter{cols}{0}
    \setcounter{rows}{6}
    \setcounter{row}{0}
    \setrow {2,3,3,2,3}
    \setrow {1,2,4,3,1}
    \setrow {3,2,2,2,2}
    \setrow { ,3,3,3,4};

\setcounter{cols}{0}

    \draw (3.5,16.6) node [anchor=north west,scale=0.8][inner sep=0.75pt]    {$c$};
     
    \draw (0,-0.5) node [anchor=north west,scale=0.8][inner sep=0.75pt]    {$F(c,Z_1,7) = F^1(c,Z_1)$};
   
    \draw (13,-0.5) node [anchor=north west,scale=0.8][inner sep=0.75pt]    {$F(c,Z_1,6)$};

    \draw (24,-0.5) node [anchor=north west,scale=0.8][inner sep=0.75pt]    {$F(c,Z_1,5)$};
  
    \draw (35,-0.5) node [anchor=north west,scale=0.8][inner sep=0.75pt]    {$F(c,Z_1,4)$};

    \draw (35,16.5) node [anchor=north west,scale=0.8][inner sep=0.75pt]    {$F(c,Z_1,3)$};

    \draw (13,16.5) node [anchor=north west,scale=0.8][inner sep=0.75pt]    {$F(c,Z_1,1)$};

    \draw (24,16.5) node [anchor=north west,scale=0.8][inner sep=0.75pt]    {$F(c,Z_1,2)$};

\end{tikzpicture}
\caption{Functioning of a bridge with a configuration as in \Cref{lem:funcbridge} (Every cell outside the bridge has value less than 3) for the word $Z_1 = HVVHHHV$. The figure shows how the configuration evolves through each of the transitions until completing a full cycle of $Z_1$. In blue, in the leftmost figure, the cells that are part of the $Z$-path defining the bridge connecting the $Z$-blocks (top-left and bottom-right) are shown. The signal traveling from one block to another is highlighted in red, and the cells outside the light red rectangle remain unchanged with respect to the initial configuration.}
\label{fig:bridge_basura}
\end{figure}
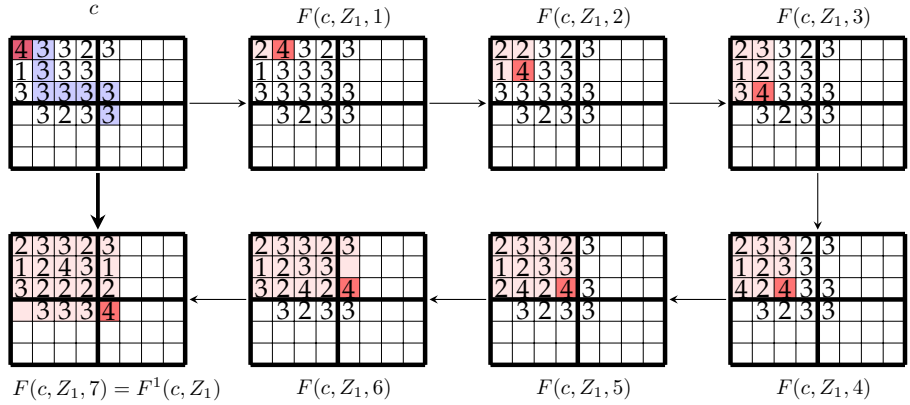

\begin{figure}[t]
\centering
\begin{tikzpicture}[scale=.29]

    \draw [draw opacity=0][fill={rgb, 255:red, 0; green, 0; blue, 255}, fill opacity=0.2] (0,14) -- (2,14) -- (2,15) -- (0,15) -- cycle;

    \draw [draw opacity=0][fill={rgb, 255:red, 0; green, 0; blue, 255}, fill opacity=0.2] (0,13) -- (0,12) -- (1,12) -- (1,13) -- cycle;

    \draw [draw opacity=0][fill={rgb, 255:red, 0; green, 0; blue, 255}, fill opacity=0.2] (1,14) -- (1,12) -- (2,12) -- (2,14) -- cycle;
    
    \draw [draw opacity=0][fill={rgb, 255:red, 0; green, 0; blue, 255}, fill opacity=0.2] (2,13) -- (5,13) -- (5,12) -- (2,12) -- cycle;
    
    \draw [draw opacity=0][fill={rgb, 255:red, 0; green, 0; blue, 255}, fill opacity=0.2] (4,12) -- (5,12) -- (5,11) -- (4,11) -- cycle;

    \draw [draw opacity=0][fill={rgb, 255:red, 0; green, 0; blue, 255}, fill opacity=0.2] (4,14) -- (5,14) -- (5,13) -- (4,13) -- cycle;

    \draw[line width=0.1mm] (0, 9) grid (8, 15);
    \draw[line width=0.6mm] (0,9) -- (0,15); 
    \draw[line width=0.6mm] (4,9) -- (4,15); 
    \draw[line width=0.6mm] (8,9) -- (8,15); 
    \draw[line width=0.6mm] (0,9) -- (8,9);
    \draw[line width=0.6mm] (0,12) -- (8,12);
    \draw[line width=0.6mm] (0,15) -- (8,15);
    
    \setcounter{rows}{15}
    \setcounter{row}{0}
    \setrow {4,3, , , }
    \setrow { ,3, , , }
    \setrow { ,3,3,3,3}
    \setrow { , , , ,3};

\end{tikzpicture}
\caption{Affected neighbors of a bridge $T$, whose configuration is shown, for the word $Z_1 = HVVHHHV$. In blue, $N(T)$, which coincides with the cells that gain a grain, as illustrated in \Cref{fig:ciclo}.}
\label{fig:bridge_NT}
\end{figure}
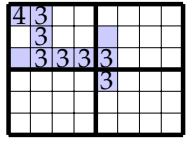

\begin{proof}
Let $\alpha,\beta$ be the constants defining the $Z$-path $P$ as in \Cref{def:defbridge}. By induction on $i,$ we show that for all $i,$ $F(\Tilde{c},Z,i)(P(i)) = 4$ and $F(\Tilde{c},Z,i)(P(j)) = 3$ for each $j>i$. For $i=0$ the result holds directly by definition. 

Let $i\geq 1$ and let us assume $F(\Tilde{c},Z,i-1)(P(i-1)) = 4$ and $F(\Tilde{c},Z,i-1)(P(j)) = 3$ for each $j>i-1$. 

From the bridge definition we have $P(i) = \alpha Z(i) [Z(i) = H]  + \beta Z(i) [Z(i) = V] + P(i-1)$. Suppose $Z(i) = H$, then $P(i) = P(i-1) + \alpha H$. Since we have: \begin{itemize}
    \item $F(\Tilde{c},Z,i-1)(P(i-1)) = 4$
    \item $F(\Tilde{c},Z,i-1)(P(i)) = 3$
    \item $P(i) \in N_h(P(i-1))$
    \item $Z(i)=H$
\end{itemize}
we conclude by definition that $F(\Tilde{c},Z,i)(P(i)) = 4$ and $F(\Tilde{c},Z,i)(P(i-1)) = 2.$

Observe that $c(P(j)) = 3$ for all $j > i$ follows from the fact that in $F(\Tilde{c}, Z, i)$, since each signal can move at most one cell per transition, either vertically or horizontally, all cells located at a horizontal distance greater than $|Z(1)\ldots Z(i)|_H$ (i.e., the number of $H$ transitions in $Z$ up to its $i$-th letter) and a vertical distance greater than $|Z(1)\ldots Z(i)|_V$ (i.e., the number of $V$ transitions in $Z$ up to its $i$-th letter) stay in the same state as in $\Tilde{c}$. Since every cell $c(P(j)) = 3$ for $j > i$ lies beyond these distances, their values remain unchanged.

\end{proof}

\begin{remark}\label{remarkfunc}
In this lemma we assume that the configuration around the bridge was initially in state $0$, but the same result holds if we assume that initially there are no other signals in the $2\times2$ square that defines the figure, nor in the blocks $\{B_i\}_{i\in \{3,4,5,6\}}$.
\end{remark}


With this, we have that the bridges allow us to move signals across blocks. Moreover, we can see that bridges are the minimal structure to carry signals, in the sense that, adding grains to adjacent cells outside the bridge does not affect its functionality (as long we do not add more than 3). 


In the following, we present two technical lemmas that we will use in the next sections. First, it is shown as a corollary that two bridges in the same direction can transport signals across multiple blocks, even when there are surrounding cells with values greater than zero. This latter property will be useful for defining more advanced structures that require moving signals through regions containing other bridges. Then, it is established how bridges with a sink interrupt the propagation of signals, yet are able to increase the value of a cell along the defining path from 2 to 3, as shown in \Cref{fig:lemma4}. This property will be useful later in defining a switch-type structure.

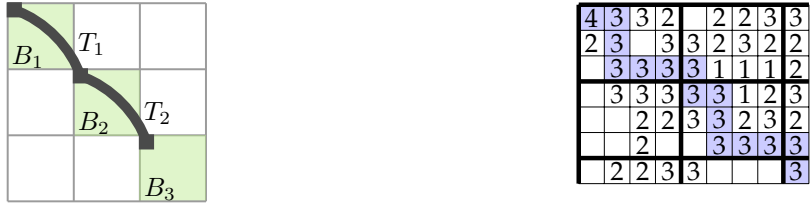
\begin{figure}[t]
\begin{minipage}[t][0.2\textheight][c]{0.5\textwidth}
\centering
\tikzset{every picture/.style={line width=0.75pt}} 

\begin{tikzpicture}[x=0.5pt,y=0.5pt,yscale=-1,xscale=1]

\draw [color={rgb, 255:red, 155; green, 155; blue, 155 }  ,draw opacity=1 ][line width=0.75]    (200,150) -- (50,150) ;
\draw [color={rgb, 255:red, 155; green, 155; blue, 155 }  ,draw opacity=1 ][line width=0.75]    (200,50) -- (200,200) ;
\draw [color={rgb, 255:red, 155; green, 155; blue, 155 }  ,draw opacity=1 ][line width=0.75]    (150,50) -- (150,200) ;
\draw  [draw opacity=0][fill={rgb, 255:red, 184; green, 233; blue, 134 }  ,fill opacity=0.43 ] (100,100) -- (150,100) -- (150,150) -- (100,150) -- cycle ;
\draw  [draw opacity=0][fill={rgb, 255:red, 184; green, 233; blue, 134 }  ,fill opacity=0.43 ] (150,150) -- (200,150) -- (200,200) -- (150,200) -- cycle ;
\draw  [draw opacity=0][fill={rgb, 255:red, 184; green, 233; blue, 134 }  ,fill opacity=0.43 ] (50,50) -- (100,50) -- (100,100) -- (50,100) -- cycle ;
\draw [color={rgb, 255:red, 155; green, 155; blue, 155 }  ,draw opacity=1 ][line width=0.75]    (100,50) -- (100,200) ;
\draw [color={rgb, 255:red, 74; green, 74; blue, 74 }  ,draw opacity=1 ][line width=3.75]    (105,105) .. controls (129.2,113.5) and (148.2,134.5) .. (155,155) ;
\draw [color={rgb, 255:red, 155; green, 155; blue, 155 }  ,draw opacity=1 ][line width=0.75]    (50,50) -- (50,200) ;
\draw [color={rgb, 255:red, 155; green, 155; blue, 155 }  ,draw opacity=1 ][line width=0.75]    (200,50) -- (50,50) ;
\draw [color={rgb, 255:red, 155; green, 155; blue, 155 }  ,draw opacity=1 ][line width=0.75]    (200,100) -- (50,100) ;
\draw  [color={rgb, 255:red, 74; green, 74; blue, 74 }  ,draw opacity=1 ][fill={rgb, 255:red, 74; green, 74; blue, 74 }  ,fill opacity=1 ] (50,50) -- (60,50) -- (60,60) -- (50,60) -- cycle ;
\draw  [color={rgb, 255:red, 74; green, 74; blue, 74 }  ,draw opacity=1 ][fill={rgb, 255:red, 74; green, 74; blue, 74 }  ,fill opacity=1 ] (150,150) -- (160,150) -- (160,160) -- (150,160) -- cycle ;
\draw  [color={rgb, 255:red, 74; green, 74; blue, 74 }  ,draw opacity=1 ][fill={rgb, 255:red, 74; green, 74; blue, 74 }  ,fill opacity=1 ] (100,100) -- (110,100) -- (110,110) -- (100,110) -- cycle ;
\draw [color={rgb, 255:red, 74; green, 74; blue, 74 }  ,draw opacity=1 ][line width=3.75]    (55,55) .. controls (79.2,63.5) and (98.2,84.5) .. (105,105) ;
\draw [color={rgb, 255:red, 155; green, 155; blue, 155 }  ,draw opacity=1 ][line width=0.75]    (200,200) -- (50,200) ;

\draw (101,70.4) node [anchor=north west][inner sep=0.75pt]    {$T_{1}$};
\draw (151,122.4) node [anchor=north west][inner sep=0.75pt]    {$T_{2}$};
\draw (151,180.4) node [anchor=north west][inner sep=0.75pt]    {$B_{3}$};
\draw (101,130.4) node [anchor=north west][inner sep=0.75pt]    {$B_{2}$};
\draw (51,80.4) node [anchor=north west][inner sep=0.75pt]    {$B_{1}$};

\end{tikzpicture}\\
\end{minipage}
\begin{minipage}[t][0.2\textheight][c]{0.5\textwidth}
\centering
\begin{tikzpicture}[scale=.34]
    
    \draw  [draw opacity=0][fill={rgb, 255:red, 0; green, 0; blue,  255 }  ,fill opacity=0.2 ] (0,6) -- (2,6) -- (2,7) -- (0,7) -- cycle ;

    \draw  [draw opacity=0][fill={rgb, 255:red, 0; green, 0; blue,  255 }  ,fill opacity=0.2 ] (1,6) -- (1,4) -- (2,4) -- (2,6) -- cycle ;

    \draw  [draw opacity=0][fill={rgb, 255:red, 0; green, 0; blue,  255 }  ,fill opacity=0.2 ] (2,5) -- (5,5) -- (5,4) -- (2,4) -- cycle ;

    \draw  [draw opacity=0][fill={rgb, 255:red, 0; green, 0; blue,  255 }  ,fill opacity=0.2 ] (4,4) -- (6,4) -- (6,3) -- (4,3) -- cycle ;

    \draw  [draw opacity=0][fill={rgb, 255:red, 0; green, 0; blue,  255 }  ,fill opacity=0.2 ] (5,3) -- (6,3) -- (6,1) -- (5,1) -- cycle ;

    \draw  [draw opacity=0][fill={rgb, 255:red, 0; green, 0; blue,  255 }  ,fill opacity=0.2 ] (6,2) -- (9,2) -- (9,1) -- (6,1) -- cycle ;

    \draw  [draw opacity=0][fill={rgb, 255:red, 0; green, 0; blue,  255 }  ,fill opacity=0.2 ] (8,1) -- (9,1) -- (9,0) -- (8,0) -- cycle ;

    \draw[line width=0.1mm] (0, 0) grid (9, 7);
    \draw[line width=0.6mm] (0,0) -- (0,7); 
    \draw[line width=0.6mm] (4,0) -- (4,7); 
    \draw[line width=0.6mm] (8,0) -- (8,7); 
    \draw[line width=0.6mm] (0,1) -- (9,1);
    \draw[line width=0.6mm] (0,4) -- (9,4);
    \draw[line width=0.6mm] (0,7) -- (9,7);

    \setcounter{cols}{0}
    \setcounter{rows}{7}
    \setcounter{row}{0}
    \setrow {4,3,3,2, ,2,2,3,3}
    \setrow {2,3, ,3,3,2,3,2,2}
    \setrow { ,3,3,3,3,1,1,1,2}
    \setrow { ,3,3,3,3,3,1,2,3}
    \setrow { , ,2,2,3,3,2,3,2}
    \setrow { , ,2, , ,3,3,3,3}
    \setrow { ,2,2,3,3, , , ,3}
    \setrow { , , , , };

\setcounter{cols}{0}

\end{tikzpicture}
\end{minipage}
\caption{On the left, the structure for \Cref{lem:funcbridge2} is shown: $Z$-blocks $B_1$, $B_2$, and $B_3$ are diagonally connected, with two bridges $T_1$ and $T_2$ connecting $B_1$ to $B_2$, and $B_2$ to $B_3$, respectively. On the right, an example is given for the word $Z = HVVHHHV$. Thick black lines indicate the division between blocks. We consider two bridges $T_i$ for $i = 1,2$, where $T_1$ connects the top-left $Z$-block ($B_1$) with the central $Z$-block ($B_2$), and $T_2$ connects the central $Z$-block with the bottom-right block ($B_3$, not fully shown in the figure). The cells that belong to either bridge are highlighted in blue. All cells outside the bridges have values strictly less than 4.}
\label{fig:lemma3}
\end{figure}

\begin{corollary}\label{lem:funcbridge2}
    Let $Z \in \{H,V\}^{h,v,k}$ and let us consider 2 $Z$-bridges
    $T_i=(P_i,c_i)$ with $i \in \{1,2\}$ connecting blocks $B_1,B_2$ and
    $B_2,B_3$ respectively as in \Cref{fig:lemma3}. If we consider the
    configuration $\Tilde{c}(x)$ where 
    \[
        \Tilde{c}(x) = \begin{cases}
        4 & \text{ if } x=P_1(0), \\
        3 & \text{ if } x \in (P_1 \cup P_2)\setminus P_1(0), \\
        s(x) & \text{ if } x \in (\overline{B_1} \cup \overline{B_2} \cup \overline{B_3})\setminus (P_1\cup P_2)\end{cases}
    \]
    with $s(x)$ strictly less than 4, we have that $F^2(\Tilde{c}(x),Z)(P_2(k))=4.$
\end{corollary}
\begin{proof}
    We note that this structure satisfies the hypotheses of \Cref{lem:funcbridge} for a bridge $T$ defined for the word $ZZ$ (that is, $Z$ concatenated with itself). Therefore, the conclusion follows identically. 
\end{proof}

This notion of combining bridges will be generalized in the next section, but we treat this case separately because it is stronger than our previous assumptions regarding the functionality of combining bridges. In particular, this situation requires special attention because, unlike before, now there may be garbage outside the bridges. This result follows directly from the proof of \Cref{lem:funcbridge}.

Let $Z \in \{H,V\}^{v,h,k}$ a word and $T =(P,c)$ a $Z$-bridge. We define $N(P)$, called the \emph{affected neighbor} of $T$, as the set of cells which gain a cell over the transition of a cycle when $T$ is transmitting a signal (See \Cref{fig:bridge_NT}). 

Formally, given a bridge $T = (P,c)$ for a word $Z$ and $I = \{i\,|\,Z(i) = H\}$
and $J = \{j\,|\,Z(j) = V\}$, we define 
\[
  N(T) = P \cup\left[\bigcup_{i\in I}N_h(P(i))\right] \cup \left[\bigcup_{j\in J}N_v(P(j))\right]
\]

\begin{lemma}\label{lem:funcbridgesink}
    Let $Z \in \{H,V\}^{h,v,k}$ and let us consider a source $Z$-bridge $T =(P,c)$ with a sink at $l$ connecting $B_1^+,B_2^+$ from two diagonally connected blocks $B_1,B_2$, with $P = \{P(0),...,P(k)\}$.  If we consider the configuration where every cell outside the bridge has value less or equal to 3, and every cell at $N(T)$ has value less than 3, i.e if we define $\Tilde{c}(x) = \begin{cases}
        c(x) & \text{ if } x \in P, \\
        s(x) & \text{ if } x \notin P\end{cases}$ where $s$ is any configuration with values strictly less than 4 outside $N(T)$ and strictly less than 3 at $N(T)$, then we have that $F(\Tilde{c}(x),Z)(P(l))=3.$
\end{lemma}

That is, bridges with a Sink interrupt the propagation of signals through the bridge, which will be useful for constructing switches.

\begin{proof}
    Proof follows analogous to previous lemma except we have $F(\Tilde{c},Z,l)(P(l)) = F(\Tilde{c},Z,l-1)(P(l)) +1=3$, which ensures with our assumption for $N(P)$ that no cell has a value of 4 (or greater) at step $l$ and $F(\Tilde{c},Z)(x) \leq 3,  \forall x \in N(P)$.
\end{proof}

\begin{figure}[t]
\centering
\resizebox{\textwidth}{!}{%
\begin{tikzpicture}[scale=.3]

    \draw  [draw opacity=0][fill={rgb, 255:red, 255; green, 0; blue,  0 }  ,fill opacity=0.5 ] (0,15) -- (1,15) -- (1,14) -- (0,14) -- cycle;
    
    \draw [draw opacity=0, fill={rgb, 255:red, 255; green, 0; blue, 0}, fill opacity=0.5] (12,15) -- (13,15) -- (13,14) -- (12,14) -- cycle;

    \draw [draw opacity=0, fill={rgb, 255:red, 255; green, 0; blue, 0}, fill opacity=0.5] (23,14) -- (24,14) -- (24,13) -- (23,13) -- cycle;

    \draw[line width=0.1mm] (0, 9) grid (8, 15);
    \draw[line width=0.6mm] (0,9) -- (0,15); 
    \draw[line width=0.6mm] (4,9) -- (4,15); 
    \draw[line width=0.6mm] (8,9) -- (8,15); 
    \draw[line width=0.6mm] (0,9) -- (8,9);
    \draw[line width=0.6mm] (0,12) -- (8,12);
    \draw[line width=0.6mm] (0,15) -- (8,15);
    \setcounter{cols}{0}
    \setcounter{rows}{15}
    \setcounter{row}{0}
    \setrow {4,3,3, , }
    \setrow {3,3,3,3, }
    \setrow { ,2, , , }
    \setrow {3, ,3, , };

    \draw[line width=0.1mm] (11, 9) grid (19, 15);
    \draw[line width=0.6mm] (11,9) -- (11,15); 
    \draw[line width=0.6mm] (15,9) -- (15,15); 
    \draw[line width=0.6mm] (19,9) -- (19,15); 
    \draw[line width=0.6mm] (11,9) -- (19,9);
    \draw[line width=0.6mm] (11,12) -- (19,12);
    \draw[line width=0.6mm] (11,15) -- (19,15);

    \setcounter{cols}{11}
    \setcounter{rows}{15}
    \setcounter{row}{0}
    \setrow {2,4,3, , }
    \setrow {3,3,3,3, }
    \setrow { ,2, , , }
    \setrow {3, ,3, , };

    \draw[line width=0.1mm] (22, 9) grid (30, 15);
    \draw[line width=0.6mm] (22,9) -- (22,15); 
    \draw[line width=0.6mm] (26,9) -- (26,15); 
    \draw[line width=0.6mm] (30,9) -- (30,15); 
    \draw[line width=0.6mm] (22,9) -- (30,9);
    \draw[line width=0.6mm] (22,12) -- (30,12);
    \draw[line width=0.6mm] (22,15) -- (30,15);

    \setcounter{cols}{22}
    \setcounter{rows}{15}
    \setcounter{row}{0}
    \setrow {2,2,3, , }
    \setrow {3,4,3,3, }
    \setrow { ,2, , , }
    \setrow {3, ,3, , };

    \draw[line width=0.1mm] (33, 9) grid (41, 15);
    \draw[line width=0.6mm] (33,9) -- (33,15); 
    \draw[line width=0.6mm] (37,9) -- (37,15); 
    \draw[line width=0.6mm] (41,9) -- (41,15); 
    \draw[line width=0.6mm] (33,9) -- (41,9);
    \draw[line width=0.6mm] (33,12) -- (41,12);
    \draw[line width=0.6mm] (33,15) -- (41,15);

    \setcounter{cols}{33}
    \setcounter{rows}{15}
    \setcounter{row}{0}
    \setrow {2,3,3, , }
    \setrow {3,2,3,3, }
    \setrow { ,3, , , }
    \setrow {3, ,3, , };

\setcounter{cols}{0}
    \draw (3.5,16.6) node [anchor=north west,scale=0.8][inner sep=0.75pt]    {$c$};
    
    \draw (35,16.5) node [anchor=north west,scale=0.8][inner sep=0.75pt]    {$F(c,Z_1,3)$};

    \draw (13,16.5) node [anchor=north west,scale=0.8][inner sep=0.75pt]    {$F(c,Z_1,1)$};

    \draw (24,16.5) node [anchor=north west,scale=0.8][inner sep=0.75pt]    {$F(c,Z_1,2)$};

\end{tikzpicture}
}
\caption{Structure and functionality for \Cref{lem:funcbridgesink}, corresponding to the word $Z_1 = HVVHHHV$, and a bridge $T = (P,c)$ with a sink at 3. Thick black lines indicate the division between blocks. Cells outside the bridge have values less than 4, and after 3 updates, $P(3)$ reaches value 3, having started at 2.} 
\label{fig:lemma4}
\end{figure}
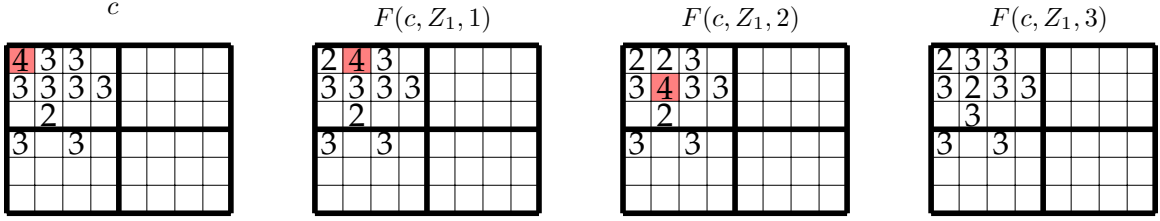


\paragraph{Combining bridges.}
As previously mentioned, the entire circuit is built around bridges, which will serve as the modular structure of the paper. Since bridges depend on the word $Z$, we need to define a formal way to combine them. After that, we will be able to describe every other structure as a union of bridges, this will be called gadget.

As bridges are defined as a tuple consisting of a set and a configuration, we define an operation between configurations to be used in the union of bridges.
Given the configurations $c_1,c_2$ we define $c_1\oplus c_2$ as the
configuration $c$ such that: 
\[
c(x) = \begin{cases}
2 & \text{ if } c_1(x) = 2 \lor c_2(x) = 2\\
\max\{c_1,c_2\}(x) & \text{ if not}
\end{cases}
\]
Given two bridges $T_1 = (P_1,c_1)$ and $T_2 = (P_2,c_2)$. We define $T_1+T_2 := (P_1\cup P_2, c_1 \oplus c_2)$ wich will be called bridge gadget, or simply gadget when the context is clear. We extend the definition to define gadgets obtained from other gadgets, formally, given two gadgets $G_1 = (Q_1,c_1), G_2 = (Q_2,c_2)$. We define the gadget $G = G_1+G_2 :=  (G_1\cup G_2, c_1 \oplus c_2)$.

We saw that bridges are used to transfer signals between 2 blocks. With this operation, we will be able to use bridges to move signals throught any number of blocks (With concatenations as in \cref{fig:operation2}), or even to duplicate a signal, as will be seen in \cref{sec:layer-2}. 

\begin{figure}[t]
\centering
\resizebox{.7\textwidth}{!}{%
\begin{tikzpicture}[scale=.34]

    \draw[line width=0.1mm] (0, 8) grid (9, 15);
    \draw[line width=0.6mm] (0,8) -- (0,15); 
    \draw[line width=0.6mm] (4,8) -- (4,15); 
    \draw[line width=0.6mm] (8,8) -- (8,15); 
    \draw[line width=0.6mm] (0,9) -- (9,9);
    \draw[line width=0.6mm] (0,12) -- (9,12);
    \draw[line width=0.6mm] (0,15) -- (9,15);
    
    \setcounter{rows}{15}
    \setcounter{row}{0}
    \setrow {4,3, , , }
    \setrow { ,3, , , }
    \setrow { ,3,3,3,3}
    \setrow { , , , ,3};
    
    \draw (9.5,12) node [anchor=north west,scale=1][inner sep=0.75pt]    {$\oplus$};
    
    \draw[line width=0.1mm] (11, 8) grid (20, 15);
    \draw[line width=0.6mm] (11,8) -- (11,15); 
    \draw[line width=0.6mm] (15,8) -- (15,15); 
    \draw[line width=0.6mm] (19,8) -- (19,15); 
    \draw[line width=0.6mm] (11,9) -- (20,9);
    \draw[line width=0.6mm] (11,12) -- (20,12);
    \draw[line width=0.6mm] (11,15) -- (20,15);

    \setcounter{cols}{11}
    \setcounter{rows}{15}
    \setcounter{row}{0}
    \setrow { , , , , }
    \setrow { , , , , }
    \setrow { , , , , }
    \setrow { , , , ,3,3}
    \setrow { , , , , ,3}
    \setrow { , , , , ,3,3,3,3}
    \setrow { , , , , , , , ,3}
    \setrow { , , , , };
    
    \draw (20.5,12) node [anchor=north west,scale=1][inner sep=0.75pt]    {$=$};
    
    \draw[line width=0.1mm] (22, 8) grid (31, 15);
    \draw[line width=0.6mm] (22,8) -- (22,15); 
    \draw[line width=0.6mm] (26,8) -- (26,15); 
    \draw[line width=0.6mm] (30,8) -- (30,15); 
    \draw[line width=0.6mm] (22,9) -- (31,9);
    \draw[line width=0.6mm] (22,12) -- (31,12);
    \draw[line width=0.6mm] (22,15) -- (31,15);

    \setcounter{cols}{22}
    \setcounter{rows}{15}
    \setcounter{row}{0}
    \setrow {4,3, , , }
    \setrow { ,3, , , }
    \setrow { ,3,3,3,3}
    \setrow { , , , ,3,3}
    \setrow { , , , , ,3}
    \setrow { , , , , ,3,3,3,3}
    \setrow { , , , , , , , ,3}
    \setrow { , , , , };

\setcounter{cols}{0}

\end{tikzpicture}
}
\caption{Operation $\oplus$ to concatenate two bridges, which will be useful for passing signals through multiple blocks (in this case, two)}
\label{fig:operation2}
\end{figure}
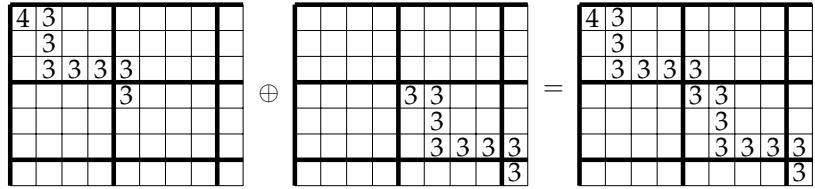

\section{Implementation of Layer 2 Components} \label{sec:layer-2}
\subsection{Semi-wires}\label{Sec_semi_wire}
As we mentioned, concatenating enough bridges can easily be used to achieve a gadget to move signals through any number of $Z-$blocks, this will be achieved with a gadget we call semi-wire. 

Formally, for any sequence of distinct $Z$-blocks ${B_1, \dots, B_{m}}$, such that for all $i \in \{0,...,m-1\}$, $B_i$ is diagonally connected to $B_{i+1}$ and is not diagonally connected to any block $B_j$ with $j \neq i+1,i-1$. Given the positive (resp. negative) bridges $T_i = (P_i,c_i)$, with $P_i$ connecting $B_i$ and $B_{i+1}$, we define a positive (resp. negative) semi-wire as the bridge gadget $(S,c) = \sum_{i\in [0...m-1]} T_i$. We call $B_0^+$ (resp. $B_0^-$) the input of the semi-wire, and $B_m^+$ (resp. $B_m^-$) its output. 

\begin{figure}
\centering
\tikzset{every picture/.style={line width=0.75pt}} 

\begin{tikzpicture}[x=0.6pt,y=0.6pt,yscale=-1,xscale=1]

\draw  [draw opacity=0][fill={rgb, 255:red, 184; green, 233; blue, 134 }  ,fill opacity=0.5 ] (450,150) -- (500,150) -- (500,200) -- (450,200) -- cycle ;
\draw  [draw opacity=0][fill={rgb, 255:red, 184; green, 233; blue, 134 }  ,fill opacity=0.5 ] (400,200) -- (450,200) -- (450,250) -- (400,250) -- cycle ;
\draw  [draw opacity=0][fill={rgb, 255:red, 184; green, 233; blue, 134 }  ,fill opacity=0.5 ] (100,200) -- (150,200) -- (150,250) -- (100,250) -- cycle ;
\draw  [draw opacity=0][fill={rgb, 255:red, 184; green, 233; blue, 134 }  ,fill opacity=0.5 ] (150,250) -- (200,250) -- (200,300) -- (150,300) -- cycle ;
\draw  [draw opacity=0][fill={rgb, 255:red, 208; green, 2; blue, 27 }  ,fill opacity=0.71 ] (200,200) -- (250,200) -- (250,250) -- (200,250) -- cycle ;
\draw  [draw opacity=0][fill={rgb, 255:red, 184; green, 233; blue, 134 }  ,fill opacity=0.5 ] (250,150) -- (300,150) -- (300,200) -- (250,200) -- cycle ;
\draw  [draw opacity=0][fill={rgb, 255:red, 208; green, 2; blue, 27 }  ,fill opacity=0.71 ] (150,150) -- (200,150) -- (200,200) -- (150,200) -- cycle ;
\draw  [draw opacity=0][fill={rgb, 255:red, 184; green, 233; blue, 134 }  ,fill opacity=0.5 ] (100,100) -- (150,100) -- (150,150) -- (100,150) -- cycle ;
\draw [color={rgb, 255:red, 155; green, 155; blue, 155 }  ,draw opacity=1 ]   (100,250) -- (300,250) ;
\draw [color={rgb, 255:red, 155; green, 155; blue, 155 }  ,draw opacity=1 ]   (100,200) -- (300,200) ;
\draw [color={rgb, 255:red, 155; green, 155; blue, 155 }  ,draw opacity=1 ]   (150,100) -- (150,300) ;
\draw [color={rgb, 255:red, 155; green, 155; blue, 155 }  ,draw opacity=1 ]   (200,100) -- (200,300) ;
\draw [color={rgb, 255:red, 155; green, 155; blue, 155 }  ,draw opacity=1 ]   (100,150) -- (300,150) ;
\draw  [color={rgb, 255:red, 155; green, 155; blue, 155 }  ,draw opacity=1 ][line width=0.75]  (100,100) -- (300,100) -- (300,300) -- (100,300) -- cycle ;
\draw [color={rgb, 255:red, 155; green, 155; blue, 155 }  ,draw opacity=1 ]   (250,100) -- (250,300) ;
\draw [color={rgb, 255:red, 74; green, 74; blue, 74 }  ,draw opacity=1 ][line width=3.75]    (155,255) .. controls (179.2,246.5) and (198.2,225.5) .. (205,205) ;
\draw  [color={rgb, 255:red, 74; green, 74; blue, 74 }  ,draw opacity=1 ][fill={rgb, 255:red, 74; green, 74; blue, 74 }  ,fill opacity=1 ] (150,260) -- (160,260) -- (160,250) -- (150,250) -- cycle ;
\draw  [color={rgb, 255:red, 74; green, 74; blue, 74 }  ,draw opacity=1 ][fill={rgb, 255:red, 74; green, 74; blue, 74 }  ,fill opacity=1 ] (150,150) -- (160,150) -- (160,160) -- (150,160) -- cycle ;
\draw [color={rgb, 255:red, 74; green, 74; blue, 74 }  ,draw opacity=1 ][line width=3.75]    (105,105) .. controls (129.2,113.5) and (148.2,134.5) .. (155,155) ;
\draw  [color={rgb, 255:red, 74; green, 74; blue, 74 }  ,draw opacity=1 ][fill={rgb, 255:red, 74; green, 74; blue, 74 }  ,fill opacity=1 ] (100,110) -- (110,110) -- (110,100) -- (100,100) -- cycle ;
\draw [color={rgb, 255:red, 74; green, 74; blue, 74 }  ,draw opacity=1 ][line width=3.75]    (155,155) .. controls (130.8,163.5) and (111.8,184.5) .. (105,205) ;
\draw  [color={rgb, 255:red, 74; green, 74; blue, 74 }  ,draw opacity=1 ][fill={rgb, 255:red, 74; green, 74; blue, 74 }  ,fill opacity=1 ] (100,210) -- (110,210) -- (110,200) -- (100,200) -- cycle ;
\draw [color={rgb, 255:red, 74; green, 74; blue, 74 }  ,draw opacity=1 ][line width=3.75]    (105,205) .. controls (129.2,213.5) and (148.2,234.5) .. (155,255) ;
\draw  [color={rgb, 255:red, 74; green, 74; blue, 74 }  ,draw opacity=1 ][fill={rgb, 255:red, 74; green, 74; blue, 74 }  ,fill opacity=1 ] (200,210) -- (210,210) -- (210,200) -- (200,200) -- cycle ;
\draw  [color={rgb, 255:red, 74; green, 74; blue, 74 }  ,draw opacity=1 ][fill={rgb, 255:red, 74; green, 74; blue, 74 }  ,fill opacity=1 ] (250,160) -- (260,160) -- (260,150) -- (250,150) -- cycle ;
\draw [color={rgb, 255:red, 74; green, 74; blue, 74 }  ,draw opacity=1 ][line width=3.75]    (205,205) .. controls (229.2,196.5) and (248.2,175.5) .. (255,155) ;
\draw  [draw opacity=0][fill={rgb, 255:red, 184; green, 233; blue, 134 }  ,fill opacity=0.5 ] (500,200) -- (550,200) -- (550,250) -- (500,250) -- cycle ;
\draw  [draw opacity=0][fill={rgb, 255:red, 184; green, 233; blue, 134 }  ,fill opacity=0.5 ] (350,150) -- (400,150) -- (400,200) -- (350,200) -- cycle ;
\draw [color={rgb, 255:red, 155; green, 155; blue, 155 }  ,draw opacity=1 ]   (350,250) -- (550,250) ;
\draw [color={rgb, 255:red, 155; green, 155; blue, 155 }  ,draw opacity=1 ]   (400,150) -- (400,250) ;
\draw [color={rgb, 255:red, 155; green, 155; blue, 155 }  ,draw opacity=1 ]   (450,150) -- (450,250) ;
\draw [color={rgb, 255:red, 155; green, 155; blue, 155 }  ,draw opacity=1 ]   (350,200) -- (550,200) ;
\draw  [color={rgb, 255:red, 155; green, 155; blue, 155 }  ,draw opacity=1 ][line width=0.75]  (350,150) -- (550,150) -- (550,250) -- (350,250) -- cycle ;
\draw  [color={rgb, 255:red, 74; green, 74; blue, 74 }  ,draw opacity=1 ][fill={rgb, 255:red, 74; green, 74; blue, 74 }  ,fill opacity=1 ] (400,200) -- (410,200) -- (410,210) -- (400,210) -- cycle ;
\draw [color={rgb, 255:red, 74; green, 74; blue, 74 }  ,draw opacity=1 ][line width=3.75]    (355,155) .. controls (379.2,163.5) and (398.2,184.5) .. (405,205) ;
\draw  [color={rgb, 255:red, 74; green, 74; blue, 74 }  ,draw opacity=1 ][fill={rgb, 255:red, 74; green, 74; blue, 74 }  ,fill opacity=1 ] (350,160) -- (360,160) -- (360,150) -- (350,150) -- cycle ;
\draw  [color={rgb, 255:red, 74; green, 74; blue, 74 }  ,draw opacity=1 ][fill={rgb, 255:red, 74; green, 74; blue, 74 }  ,fill opacity=1 ] (500,210) -- (510,210) -- (510,200) -- (500,200) -- cycle ;
\draw  [color={rgb, 255:red, 74; green, 74; blue, 74 }  ,draw opacity=1 ][fill={rgb, 255:red, 74; green, 74; blue, 74 }  ,fill opacity=1 ] (450,160) -- (460,160) -- (460,150) -- (450,150) -- cycle ;
\draw [color={rgb, 255:red, 74; green, 74; blue, 74 }  ,draw opacity=1 ][line width=3.75]    (405,205) .. controls (429.2,196.5) and (448.2,175.5) .. (455,155) ;
\draw [color={rgb, 255:red, 155; green, 155; blue, 155 }  ,draw opacity=1 ]   (500,150) -- (500,175.5) -- (500,250) ;
\draw [color={rgb, 255:red, 74; green, 74; blue, 74 }  ,draw opacity=1 ][line width=3.75]    (455,155) .. controls (479.2,163.5) and (498.2,184.5) .. (505,205) ;

\draw (109,122.4) node [anchor=north west][inner sep=0.75pt]    {$B_{1}$};
\draw (159,172.4) node [anchor=north west][inner sep=0.75pt]    {$B_{2}$};
\draw (109,222.4) node [anchor=north west][inner sep=0.75pt]    {$B_{3}$};
\draw (159,272.4) node [anchor=north west][inner sep=0.75pt]    {$B_{4}$};
\draw (221,222.4) node [anchor=north west][inner sep=0.75pt]    {$B_{5}$};
\draw (271,172.4) node [anchor=north west][inner sep=0.75pt]    {$B_{6}$};
\draw (359,172.4) node [anchor=north west][inner sep=0.75pt]    {$B_{1}$};
\draw (409,222.4) node [anchor=north west][inner sep=0.75pt]    {$B_{2}$};
\draw (521,222.4) node [anchor=north west][inner sep=0.75pt]    {$B_{4}$};
\draw (459,172.4) node [anchor=north west][inner sep=0.75pt]    {$B_{3}$};

\end{tikzpicture}
\caption{Invalid sequence for a semi-wire (left), as blocks $B_2$ and $B_5$ (shown in red) are diagonally connected but are not consecutive in the enumeration defining the semi-wire, which could allow signals to travel directly from $B_2$ to $B_5$, essentially accelerating signal propagation through the circuit. This is problematic, since—as we will see in later sections—the circuit is delay-sensitive. Valid sequence for a semi-wire (right)}
\label{fig:semiwireej}
\end{figure}
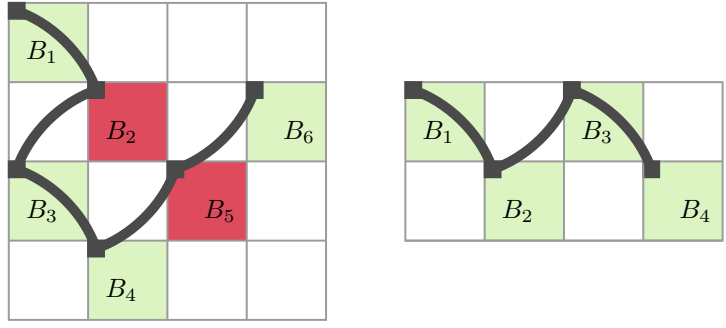

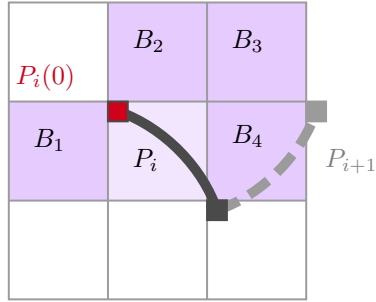
\begin{figure}
\centering
\tikzset{every picture/.style={line width=0.75pt}} 

\begin{tikzpicture}[x=0.75pt,y=0.75pt,yscale=-1,xscale=1]

\draw  [draw opacity=0][fill={rgb, 255:red, 144; green, 19; blue, 254 }  ,fill opacity=0.11 ] (130,119.85) -- (180,119.85) -- (180,169.85) -- (130,169.85) -- cycle ;
\draw  [draw opacity=0][fill={rgb, 255:red, 144; green, 19; blue, 254 }  ,fill opacity=0.22 ] (130,70) -- (180,70) -- (180,120) -- (130,120) -- cycle ;
\draw  [draw opacity=0][fill={rgb, 255:red, 144; green, 19; blue, 254 }  ,fill opacity=0.22 ] (180,70) -- (230,70) -- (230,120) -- (180,120) -- cycle ;
\draw  [draw opacity=0][fill={rgb, 255:red, 144; green, 19; blue, 254 }  ,fill opacity=0.22 ] (180,119.85) -- (230,119.85) -- (230,169.85) -- (180,169.85) -- cycle ;
\draw  [draw opacity=0][fill={rgb, 255:red, 144; green, 19; blue, 254 }  ,fill opacity=0.22 ] (80,120) -- (130,120) -- (130,170) -- (80,170) -- cycle ;
\draw  [color={rgb, 255:red, 155; green, 155; blue, 155 }  ,draw opacity=1 ][fill={rgb, 255:red, 155; green, 155; blue, 155 }  ,fill opacity=1 ] (230,130) -- (240,130) -- (240,120) -- (230,120) -- cycle ;
\draw [color={rgb, 255:red, 155; green, 155; blue, 155 }  ,draw opacity=1 ][line width=3.75]  [dash pattern={on 9pt off 4.5pt}]  (185,175) .. controls (209.2,166.5) and (228.2,145.5) .. (235,125) ;
\draw [color={rgb, 255:red, 155; green, 155; blue, 155 }  ,draw opacity=1 ]   (80,120) -- (230,120) ;
\draw [color={rgb, 255:red, 155; green, 155; blue, 155 }  ,draw opacity=1 ]   (80,170) -- (230,170) ;
\draw [color={rgb, 255:red, 155; green, 155; blue, 155 }  ,draw opacity=1 ]   (180,70) -- (180,220) ;
\draw [color={rgb, 255:red, 155; green, 155; blue, 155 }  ,draw opacity=1 ]   (130,70) -- (130,220) ;
\draw  [color={rgb, 255:red, 74; green, 74; blue, 74 }  ,draw opacity=1 ][fill={rgb, 255:red, 74; green, 74; blue, 74 }  ,fill opacity=1 ] (180,169.85) -- (190,169.85) -- (190,179.85) -- (180,179.85) -- cycle ;
\draw [color={rgb, 255:red, 74; green, 74; blue, 74 }  ,draw opacity=1 ][line width=3.75]    (135,124.85) .. controls (159.2,133.35) and (178.2,154.35) .. (185,174.85) ;
\draw  [color={rgb, 255:red, 74; green, 74; blue, 74 }  ,draw opacity=1 ][fill={rgb, 255:red, 208; green, 2; blue, 27 }  ,fill opacity=1 ] (130,130) -- (140,130) -- (140,120) -- (130,120) -- cycle ;
\draw  [color={rgb, 255:red, 155; green, 155; blue, 155 }  ,draw opacity=1 ][line width=0.75]  (80,70) -- (230,70) -- (230,219.85) -- (80,219.85) -- cycle ;

\draw (141,142.4) node [anchor=north west][inner sep=0.75pt]  [color={rgb, 255:red, 0; green, 0; blue, 0 }  ,opacity=1 ]  {$P_{i}$};
\draw (82,99.4) node [anchor=north west][inner sep=0.75pt]  [color={rgb, 255:red, 208; green, 2; blue, 27 }  ,opacity=1 ]  {$P_{i}( 0)$};
\draw (91,132.4) node [anchor=north west][inner sep=0.75pt]  [color={rgb, 255:red, 0; green, 0; blue, 0 }  ,opacity=1 ]  {$B_{1}$};
\draw (141,82.4) node [anchor=north west][inner sep=0.75pt]  [color={rgb, 255:red, 0; green, 0; blue, 0 }  ,opacity=1 ]  {$B_{2}$};
\draw (191,82.4) node [anchor=north west][inner sep=0.75pt]  [color={rgb, 255:red, 0; green, 0; blue, 0 }  ,opacity=1 ]  {$B_{3}$};
\draw (191,130.4) node [anchor=north west][inner sep=0.75pt]  [color={rgb, 255:red, 0; green, 0; blue, 0 }  ,opacity=1 ]  {$B_{4}$};
\draw (238,142.4) node [anchor=north west][inner sep=0.75pt]  [color={rgb, 255:red, 128; green, 128; blue, 128 }  ,opacity=1 ]  {$P_{i+1}$};

\end{tikzpicture}
\caption{Affected blocks for a bridge $P_i$ over a cicle. In purple, the blocks adjacent to \(P_i\) that can gain grains during the cycle starting with a signal at \(P_i(0)\) are shown. In gray, the orientation for the path \(P_{i+1}\) that could cause it to gain grains during the cycle.}
\label{fig:afectadas}
\end{figure}

The next lemma states how signals travel through semi-wires. Roughly speaking, condition 1 assures that after every cycle, a signal arrives at the next $Z$-block and condition 2 verifies that the remaining part of the semi-wire remains unchanged so it can transmit the information.

\begin{lemma}
Let $(S,c)$ be a semi-wire gadget such that $T_0$ is a source bridge. Then, for each $i \in [m]$,\begin{enumerate}
     \item $F^i(c,Z)(B_i^+) = 4$
     \item $F^i(c,Z)(P_j(l)) = 3,\, \forall j\geq i,l \in [k]$
 \end{enumerate}
\end{lemma}


\begin{proof}
    Let $Z \in \{H,V\}^{h,v,k}$. Similarly to the proof of \Cref{lem:funcbridge} we proceed by induction over $i$. The base case holds directly. Let us assume that, for some $i \in [m],$ we have:
    \begin{enumerate}
         \item $c^i(B_i^+) = 4$
         \item $c^i(P_j(l)) = 3,\, \forall j\geq i,l \in [k]$
    \end{enumerate}
    where $c^i:= F^i(c,Z)$. We need to show that:
    \begin{enumerate}
         \item $F(c^i,k)(B_{i+1}^+) = 4$
         \item $F(c^i,k)(P_j(l)) = 3,\, \forall j\geq i+1,l \in [k]$
    \end{enumerate}

Observe (1) follows directly from inductive hypothesis and \Cref{lem:funcbridge}. 


To verify (2), we analyze every step \(1 < l < k\) of the cycle. We suppose $P_i$ has the form of \Cref{fig:afectadas}. From \Cref{lem:funcbridge}, we know that transition \(l\), \(P_i(l)\) topples. We claim that no cell on a path \(P_j\) with \(j > i\) gains a grain. First, we note that around the bridge $P_i$, the only cells that could gain a grain during the cycle are those within the purple blocks in \Cref{fig:afectadas}. We call these blocks affected blocks. Given our assumption that there are no blocks in the sequence of the semi-wire connected diagonally, we have that bridges \(P_j\) with \(j > i+1\) cannot contain any cell in the marked blocks.

Additionally, bridge $P_{i+1}$ could only contain a cell adjacent to $P_i$ if the cells in $P_{i+1}$ are located in $B_4$. This is only possible if $P_{i+1}$ takes the form shown in gray in \Cref{fig:afectadas}. In this case, we note that, due to the construction of bridges, the cells in $P_{i+1}$ that could belong to $B_4$ are located at least one column away from the first column of $B_4$. In particular, we have that $N(P_i)\cap P_{i+1} = P_{i+1}(0)$, which ensures no cell at $P_{i+1}$ gains a grain during step $l$.

Thus, we have that no cell $P_j(l)$ with $j\geq i+1,l \in [k]$ gains a grain during the cycle. The lemma holds.
\end{proof}

\subsection{Semi-duplications}
We need to be able to create a bridge gadget that carries a signal to 2 different blocks. Formally, given 3 blocks $(B_1,B_2,B_3)$ included in a $3\times2$ block shape as in \Cref{fig:duplic_p}, and the positive (resp. neg.) bridges $T_1,T_2$ connecting $B_1,B_2$ and $B_1,B_3$, we call positive (resp. neg.) semi-duplicator to the gadget $(D,c)=T_1+T_2$. We call $B_1^+$ (resp. $B_1^-$) the input of the semi-duplicator and $B_2^+,B_3^+$ (resp. $B_2^-,B_3^-$) the outputs of the semi-duplicator. 

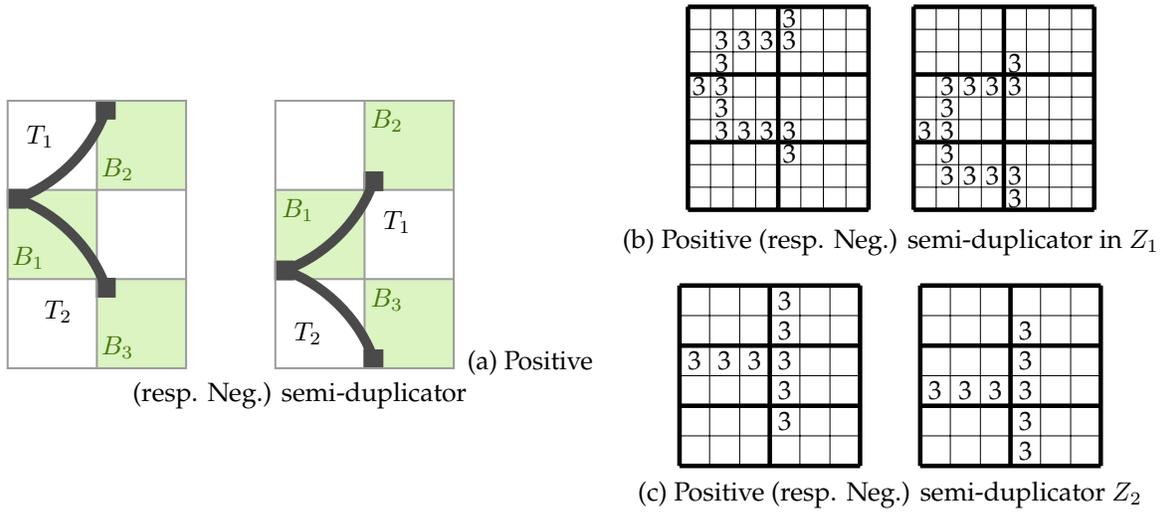
\begin{figure}
\begin{minipage}[t][0.35\textheight][c]{0.5\textwidth}
\centering
\tikzset{every picture/.style={line width=0.75pt}} 

\begin{tikzpicture}[scale = .9,x=0.75pt,y=0.75pt,yscale=-1,xscale=1]

\draw  [draw opacity=0][fill={rgb, 255:red, 184; green, 233; blue, 134 }  ,fill opacity=0.5 ] (250,150) -- (300,150) -- (300,200) -- (250,200) -- cycle ;
\draw  [draw opacity=0][fill={rgb, 255:red, 184; green, 233; blue, 134 }  ,fill opacity=0.5 ] (200,100) -- (250,100) -- (250,150) -- (200,150) -- cycle ;
\draw  [draw opacity=0][fill={rgb, 255:red, 184; green, 233; blue, 134 }  ,fill opacity=0.5 ] (250,50) -- (300,50) -- (300,100) -- (250,100) -- cycle ;
\draw  [draw opacity=0][fill={rgb, 255:red, 184; green, 233; blue, 134 }  ,fill opacity=0.5 ] (100,149.85) -- (150,149.85) -- (150,199.85) -- (100,199.85) -- cycle ;
\draw  [draw opacity=0][fill={rgb, 255:red, 184; green, 233; blue, 134 }  ,fill opacity=0.5 ] (50,100) -- (100,100) -- (100,150) -- (50,150) -- cycle ;
\draw  [draw opacity=0][fill={rgb, 255:red, 184; green, 233; blue, 134 }  ,fill opacity=0.5 ] (100,50) -- (150,50) -- (150,100) -- (100,100) -- cycle ;
\draw [color={rgb, 255:red, 155; green, 155; blue, 155 }  ,draw opacity=1 ]   (200,100) -- (300,100) ;
\draw  [color={rgb, 255:red, 74; green, 74; blue, 74 }  ,draw opacity=1 ][fill={rgb, 255:red, 74; green, 74; blue, 74 }  ,fill opacity=1 ] (250,100.15) -- (260,100.15) -- (260,90.15) -- (250,90.15) -- cycle ;
\draw [color={rgb, 255:red, 155; green, 155; blue, 155 }  ,draw opacity=1 ]   (200,150) -- (300,150) ;
\draw [color={rgb, 255:red, 155; green, 155; blue, 155 }  ,draw opacity=1 ]   (250,50) -- (250,130.3) -- (250,200) ;
\draw [color={rgb, 255:red, 155; green, 155; blue, 155 }  ,draw opacity=1 ]   (50,150) -- (150,150) ;
\draw  [color={rgb, 255:red, 74; green, 74; blue, 74 }  ,draw opacity=1 ][fill={rgb, 255:red, 74; green, 74; blue, 74 }  ,fill opacity=1 ] (100,149.85) -- (110,149.85) -- (110,159.85) -- (100,159.85) -- cycle ;
\draw [color={rgb, 255:red, 155; green, 155; blue, 155 }  ,draw opacity=1 ]   (50,100) -- (150,100) ;
\draw [color={rgb, 255:red, 155; green, 155; blue, 155 }  ,draw opacity=1 ]   (100,200) -- (100,119.7) -- (100,50) ;
\draw [color={rgb, 255:red, 74; green, 74; blue, 74 }  ,draw opacity=1 ][line width=3.75]    (55,105) .. controls (79.2,96.5) and (98.2,75.5) .. (105,55) ;
\draw [color={rgb, 255:red, 74; green, 74; blue, 74 }  ,draw opacity=1 ][line width=3.75]    (55,104.85) .. controls (79.2,113.35) and (98.2,134.35) .. (105,154.85) ;
\draw [color={rgb, 255:red, 74; green, 74; blue, 74 }  ,draw opacity=1 ][line width=3.75]    (205,145.15) .. controls (229.2,136.65) and (248.2,115.65) .. (255,95.15) ;
\draw [color={rgb, 255:red, 74; green, 74; blue, 74 }  ,draw opacity=1 ][line width=3.75]    (205,145) .. controls (229.2,153.5) and (248.2,174.5) .. (255,195) ;
\draw  [color={rgb, 255:red, 74; green, 74; blue, 74 }  ,draw opacity=1 ][fill={rgb, 255:red, 74; green, 74; blue, 74 }  ,fill opacity=1 ] (250,200) -- (260,200) -- (260,190) -- (250,190) -- cycle ;
\draw  [color={rgb, 255:red, 74; green, 74; blue, 74 }  ,draw opacity=1 ][fill={rgb, 255:red, 74; green, 74; blue, 74 }  ,fill opacity=1 ] (100,50) -- (110,50) -- (110,60) -- (100,60) -- cycle ;
\draw  [color={rgb, 255:red, 74; green, 74; blue, 74 }  ,draw opacity=1 ][fill={rgb, 255:red, 74; green, 74; blue, 74 }  ,fill opacity=1 ] (200,150.15) -- (210,150.15) -- (210,140.15) -- (200,140.15) -- cycle ;
\draw  [color={rgb, 255:red, 74; green, 74; blue, 74 }  ,draw opacity=1 ][fill={rgb, 255:red, 74; green, 74; blue, 74 }  ,fill opacity=1 ] (50,110) -- (60,110) -- (60,100) -- (50,100) -- cycle ;
\draw  [color={rgb, 255:red, 155; green, 155; blue, 155 }  ,draw opacity=1 ][line width=0.75]  (200,50) -- (300,50) -- (300,200) -- (200,200) -- cycle ;
\draw  [color={rgb, 255:red, 155; green, 155; blue, 155 }  ,draw opacity=1 ][line width=0.75]  (50,50) -- (150,50) -- (150,200) -- (50,200) -- cycle ;

\draw (101,80.4) node [anchor=north west][inner sep=0.75pt]  [color={rgb, 255:red, 65; green, 117; blue, 5 }  ,opacity=1 ]  {$B_{2}$};
\draw (51,130.4) node [anchor=north west][inner sep=0.75pt]  [color={rgb, 255:red, 65; green, 117; blue, 5 }  ,opacity=1 ]  {$B_{1}$};
\draw (101,180.4) node [anchor=north west][inner sep=0.75pt]  [color={rgb, 255:red, 65; green, 117; blue, 5 }  ,opacity=1 ]  {$B_{3}$};
\draw (202,103.4) node [anchor=north west][inner sep=0.75pt]  [color={rgb, 255:red, 65; green, 117; blue, 5 }  ,opacity=1 ]  {$B_{1}$};
\draw (252,53.4) node [anchor=north west][inner sep=0.75pt]  [color={rgb, 255:red, 65; green, 117; blue, 5 }  ,opacity=1 ]  {$B_{2}$};
\draw (252,153.4) node [anchor=north west][inner sep=0.75pt]  [color={rgb, 255:red, 65; green, 117; blue, 5 }  ,opacity=1 ]  {$B_{3}$};
\draw (59,62.4) node [anchor=north west][inner sep=0.75pt]    {$T_{1}$};
\draw (69,160.4) node [anchor=north west][inner sep=0.75pt]    {$T_{2}$};
\draw (209,172.4) node [anchor=north west][inner sep=0.75pt]    {$T_{2}$};
\draw (259,110.4) node [anchor=north west][inner sep=0.75pt]    {$T_{1}$};

\end{tikzpicture}
(a) Positive (resp. Neg.) semi-duplicator
\end{minipage}
\begin{minipage}[t][0.35\textheight][c]{0.5\textwidth}
\centering
\begin{tikzpicture}[scale=.3]

\draw[line width=0.1mm] (0, 0) grid (8, 9);
\draw[line width=0.6mm] (0,0) -- (0,9); 
\draw[line width=0.6mm] (4,0) -- (4,9); 
\draw[line width=0.6mm] (8,0) -- (8,9); 
\draw[line width=0.6mm] (0,0) -- (8,0);
\draw[line width=0.6mm] (0,3) -- (8,3);
\draw[line width=0.6mm] (0,6) -- (8,6);
\draw[line width=0.6mm] (0,9) -- (8,9);

\setcounter{rows}{9}
\setcounter{row}{0}
\setrow { , , , ,3}
\setrow { ,3,3,3,3}
\setrow { ,3, , , }
\setrow {3,3, , , }
\setrow { ,3, , , }
\setrow { ,3,3,3,3}
\setrow { , , , ,3};


\draw[line width=0.1mm] (10, 0) grid (18, 9);
\draw[line width=0.6mm] (10,0) -- (10,9); 
\draw[line width=0.6mm] (14,0) -- (14,9); 
\draw[line width=0.6mm] (18,0) -- (18,9); 
\draw[line width=0.6mm] (10,0) -- (18,0);
\draw[line width=0.6mm] (10,3) -- (18,3);
\draw[line width=0.6mm] (10,6) -- (18,6);
\draw[line width=0.6mm] (10,9) -- (18,9);

\setcounter{cols}{10}
\setcounter{rows}{9}
\setcounter{row}{2}
\setrow { , , , ,3}
\setrow { ,3,3,3,3}
\setrow { ,3, , , }
\setrow {3,3, , , }
\setrow { ,3, , , }
\setrow { ,3,3,3,3}
\setrow { , , , ,3};

\end{tikzpicture}

(b) Positive (resp. Neg.) semi-duplicator in $Z_1$\\
~~\\

\begin{tikzpicture}[scale=.4]

\draw[line width=0.1mm] (0, 0) grid (6, 6);
\draw[line width=0.6mm] (0,0) -- (0,6); 
\draw[line width=0.6mm] (3,0) -- (3,6); 
\draw[line width=0.6mm] (6,0) -- (6,6); 
\draw[line width=0.6mm] (0,0) -- (6,0);
\draw[line width=0.6mm] (0,2) -- (6,2);
\draw[line width=0.6mm] (0,4) -- (6,4);
\draw[line width=0.6mm] (0,6) -- (6,6);

\setcounter{cols}{0}
\setcounter{rows}{6}
\setcounter{row}{0}
\setrow { , , ,3}
\setrow { , , ,3}
\setrow {3,3,3,3}
\setrow { , , ,3}
\setrow { , , ,3};

\draw[line width=0.1mm] (8, 0) grid (14, 6);
\draw[line width=0.6mm] (8,0) -- (8,6); 
\draw[line width=0.6mm] (11,0) -- (11,6); 
\draw[line width=0.6mm] (14,0) -- (14,6); 
\draw[line width=0.6mm] (8,0) -- (14,0);
\draw[line width=0.6mm] (8,2) -- (14,2);
\draw[line width=0.6mm] (8,4) -- (14,4);
\draw[line width=0.6mm] (8,6) -- (14,6);

\setcounter{cols}{8}
\setcounter{rows}{6}
\setcounter{row}{1}
\setrow { , , ,3}
\setrow { , , ,3}
\setrow {3,3,3,3}
\setrow { , , ,3}
\setrow { , , ,3};

\setcounter{cols}{0}

\end{tikzpicture}

(c) Positive (resp. Neg.) semi-duplicator $Z_2$\\
\end{minipage}
\caption{Semi-duplicator Gadget. On the left, the gadget is shown using the standard notation we’ve been using for bridges and the $Z$-blocks involved in the structure. A bridge $T_1$ connecting $B_1$ and $B_2$ and $T_2$ connecting $B_1$ and $B_3$. On the right, particular cases of the gadget are illustrated. The positive version of the gadget is always shown on the left of each subfigure.}
\label{fig:duplic_p}
\end{figure}

The functionality of the structure follows as a corollary of \Cref{lem:funcbridge} that states that a bridge gadget formed by 2 source bridges with its source at the same cell, produces both  bridges to carry a signal, essentially duplicating it:

\begin{corollary}\label{lem:funcduplic}
Let $D=T_1+T_2$ a positive semi-duplicator gadget with blocks $(B_1,B_2,B_3).$ If $T_1$ or $T_2$ are source bridges then $F(c,Z)(B_2^+)=F(c,Z)(B_3^+)=4$.
\end{corollary}

\begin{proof}
    Using \Cref{lem:funcbridge} to both bridges $T_1,T_2$, both work as intended making $F(c,Z)(B_2^+)=F(c,Z)(B_3^+)=4$.
\end{proof}

\subsection{Semi-merges}
Now we show how to merge semi-wires. In order to do this, we need to be able to create a gadget that carries a signal from 2 different blocks to one. Formally, given 3 blocks $(B_1,B_2,B_3)$ included in a $3\times2$ block shape as in \Cref{fig:merg_p}, and the positive bridges (resp. neg.) $T_1,T_2$ connecting $B_1,B_3$ and $B_2,B_3$, we call positive (resp. neg.) semi-merge to the gadget $(M,c)= T_1+T_2$. We call $B_1^+,B_2^+$ (resp. $B_1^-,B_2^-$) the input of the semi-merge and $B_3^+$ (resp. $B_3^-$) the outputs of the semi-merge.

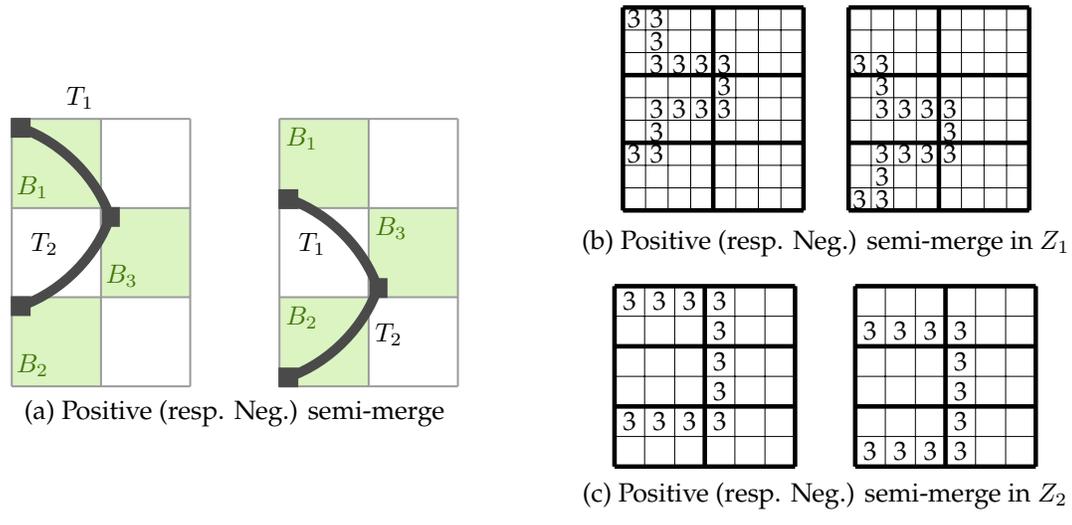
\begin{figure}
\begin{minipage}[t][0.35\textheight][c]{0.5\textwidth}
\centering
\tikzset{every picture/.style={line width=0.75pt}} 

\begin{tikzpicture}[scale = .9,x=0.75pt,y=0.75pt,yscale=-1,xscale=1]

\draw  [draw opacity=0][fill={rgb, 255:red, 184; green, 233; blue, 134 }  ,fill opacity=0.5 ] (200,150) -- (250,150) -- (250,200) -- (200,200) -- cycle ;
\draw  [draw opacity=0][fill={rgb, 255:red, 184; green, 233; blue, 134 }  ,fill opacity=0.5 ] (250,100) -- (300,100) -- (300,150) -- (250,150) -- cycle ;
\draw  [draw opacity=0][fill={rgb, 255:red, 184; green, 233; blue, 134 }  ,fill opacity=0.5 ] (200,50) -- (250,50) -- (250,100) -- (200,100) -- cycle ;
\draw  [draw opacity=0][fill={rgb, 255:red, 184; green, 233; blue, 134 }  ,fill opacity=0.5 ] (50,150) -- (100,150) -- (100,200) -- (50,200) -- cycle ;
\draw  [draw opacity=0][fill={rgb, 255:red, 184; green, 233; blue, 134 }  ,fill opacity=0.5 ] (100,100) -- (150,100) -- (150,150) -- (100,150) -- cycle ;
\draw  [draw opacity=0][fill={rgb, 255:red, 184; green, 233; blue, 134 }  ,fill opacity=0.5 ] (50,50) -- (100,50) -- (100,100) -- (50,100) -- cycle ;
\draw [color={rgb, 255:red, 155; green, 155; blue, 155 }  ,draw opacity=1 ]   (50,50) -- (50,130.3) -- (50,200) ;
\draw [color={rgb, 255:red, 155; green, 155; blue, 155 }  ,draw opacity=1 ]   (50,100) -- (150,100) ;
\draw  [color={rgb, 255:red, 74; green, 74; blue, 74 }  ,draw opacity=1 ][fill={rgb, 255:red, 74; green, 74; blue, 74 }  ,fill opacity=1 ] (100,110) -- (110,110) -- (110,100) -- (100,100) -- cycle ;
\draw [color={rgb, 255:red, 155; green, 155; blue, 155 }  ,draw opacity=1 ]   (50,50) -- (150,50) ;
\draw [color={rgb, 255:red, 155; green, 155; blue, 155 }  ,draw opacity=1 ]   (50,150) -- (150,150) ;
\draw [color={rgb, 255:red, 155; green, 155; blue, 155 }  ,draw opacity=1 ]   (100,50) -- (100,130.3) -- (100,200) ;
\draw [color={rgb, 255:red, 155; green, 155; blue, 155 }  ,draw opacity=1 ]   (150,50) -- (150,130.3) -- (150,200) ;
\draw [color={rgb, 255:red, 155; green, 155; blue, 155 }  ,draw opacity=1 ]   (50,200) -- (150,200) ;
\draw [color={rgb, 255:red, 155; green, 155; blue, 155 }  ,draw opacity=1 ]   (200,200) -- (200,119.7) -- (200,50) ;
\draw [color={rgb, 255:red, 155; green, 155; blue, 155 }  ,draw opacity=1 ]   (200,150) -- (300,150) ;
\draw  [color={rgb, 255:red, 74; green, 74; blue, 74 }  ,draw opacity=1 ][fill={rgb, 255:red, 74; green, 74; blue, 74 }  ,fill opacity=1 ] (250,139.85) -- (260,139.85) -- (260,149.85) -- (250,149.85) -- cycle ;
\draw [color={rgb, 255:red, 155; green, 155; blue, 155 }  ,draw opacity=1 ]   (200,200) -- (300,200) ;
\draw [color={rgb, 255:red, 155; green, 155; blue, 155 }  ,draw opacity=1 ]   (200,100) -- (300,100) ;
\draw [color={rgb, 255:red, 155; green, 155; blue, 155 }  ,draw opacity=1 ]   (250,200) -- (250,119.7) -- (250,50) ;
\draw [color={rgb, 255:red, 155; green, 155; blue, 155 }  ,draw opacity=1 ]   (300,200) -- (300,119.7) -- (300,50) ;
\draw [color={rgb, 255:red, 155; green, 155; blue, 155 }  ,draw opacity=1 ]   (200,50) -- (300,50) ;
\draw [color={rgb, 255:red, 74; green, 74; blue, 74 }  ,draw opacity=1 ][line width=3.75]    (205,194.85) .. controls (229.2,186.35) and (248.2,165.35) .. (255,144.85) ;
\draw [color={rgb, 255:red, 74; green, 74; blue, 74 }  ,draw opacity=1 ][line width=3.75]    (205,94.85) .. controls (229.2,103.35) and (248.2,124.35) .. (255,144.85) ;
\draw [color={rgb, 255:red, 74; green, 74; blue, 74 }  ,draw opacity=1 ][line width=3.75]    (55,155) .. controls (79.2,146.5) and (98.2,125.5) .. (105,105) ;
\draw [color={rgb, 255:red, 74; green, 74; blue, 74 }  ,draw opacity=1 ][line width=3.75]    (55,55) .. controls (79.2,63.5) and (98.2,84.5) .. (105,105) ;
\draw  [color={rgb, 255:red, 74; green, 74; blue, 74 }  ,draw opacity=1 ][fill={rgb, 255:red, 74; green, 74; blue, 74 }  ,fill opacity=1 ] (50,60) -- (60,60) -- (60,50) -- (50,50) -- cycle ;
\draw  [color={rgb, 255:red, 74; green, 74; blue, 74 }  ,draw opacity=1 ][fill={rgb, 255:red, 74; green, 74; blue, 74 }  ,fill opacity=1 ] (200,190) -- (210,190) -- (210,200) -- (200,200) -- cycle ;
\draw  [color={rgb, 255:red, 74; green, 74; blue, 74 }  ,draw opacity=1 ][fill={rgb, 255:red, 74; green, 74; blue, 74 }  ,fill opacity=1 ] (50,160) -- (60,160) -- (60,150) -- (50,150) -- cycle ;
\draw  [color={rgb, 255:red, 74; green, 74; blue, 74 }  ,draw opacity=1 ][fill={rgb, 255:red, 74; green, 74; blue, 74 }  ,fill opacity=1 ] (200,100) -- (210,100) -- (210,90) -- (200,90) -- cycle ;

\draw (51,80.4) node [anchor=north west][inner sep=0.75pt]  [color={rgb, 255:red, 65; green, 117; blue, 5 }  ,opacity=1 ]  {$B_{1}$};
\draw (101,130.4) node [anchor=north west][inner sep=0.75pt]  [color={rgb, 255:red, 65; green, 117; blue, 5 }  ,opacity=1 ]  {$B_{3}$};
\draw (51,180.4) node [anchor=north west][inner sep=0.75pt]  [color={rgb, 255:red, 65; green, 117; blue, 5 }  ,opacity=1 ]  {$B_{2}$};
\draw (202,53.4) node [anchor=north west][inner sep=0.75pt]  [color={rgb, 255:red, 65; green, 117; blue, 5 }  ,opacity=1 ]  {$B_{1}$};
\draw (252,103.4) node [anchor=north west][inner sep=0.75pt]  [color={rgb, 255:red, 65; green, 117; blue, 5 }  ,opacity=1 ]  {$B_{3}$};
\draw (202,153.4) node [anchor=north west][inner sep=0.75pt]  [color={rgb, 255:red, 65; green, 117; blue, 5 }  ,opacity=1 ]  {$B_{2}$};
\draw (79,30.4) node [anchor=north west][inner sep=0.75pt]    {$T_{1}$};
\draw (252,163.25) node [anchor=north west][inner sep=0.75pt]    {$T_{2}$};
\draw (209,112.4) node [anchor=north west][inner sep=0.75pt]    {$T_{1}$};
\draw (59,112.4) node [anchor=north west][inner sep=0.75pt]    {$T_{2}$};

\end{tikzpicture}\\
(a) Positive (resp. Neg.) semi-merge
\end{minipage}
\begin{minipage}[t][0.35\textheight][c]{0.5\textwidth}
\centering
\begin{tikzpicture}[scale=.3]
    
    \draw[line width=0.1mm] (0, 0) grid (8, 9);
    \draw[line width=0.6mm] (0,0) -- (0,9); 
    \draw[line width=0.6mm] (4,0) -- (4,9); 
    \draw[line width=0.6mm] (8,0) -- (8,9); 
    \draw[line width=0.6mm] (0,0) -- (8,0);
    \draw[line width=0.6mm] (0,3) -- (8,3);
    \draw[line width=0.6mm] (0,6) -- (8,6);
    \draw[line width=0.6mm] (0,9) -- (8,9);

    \setcounter{rows}{9}
    \setcounter{row}{0}
    \setrow {3,3, , , }
    \setrow { ,3, , , }
    \setrow { ,3,3,3,3}
    \setrow { , , , ,3}
    \setrow { ,3,3,3,3}
    \setrow { ,3, , , }
    \setrow {3,3, , , };


\setcounter{cols}{10}
\setcounter{rows}{9}
\setcounter{row}{2}

\draw[line width=0.1mm] (10, 0) grid (18, 9);
\draw[line width=0.6mm] (10,0) -- (10,9); 
\draw[line width=0.6mm] (14,0) -- (14,9); 
\draw[line width=0.6mm] (18,0) -- (18,9); 
\draw[line width=0.6mm] (10,0) -- (18,0);
\draw[line width=0.6mm] (10,3) -- (18,3);
\draw[line width=0.6mm] (10,6) -- (18,6);
\draw[line width=0.6mm] (10,9) -- (18,9);

    \setcounter{rows}{9}
    \setcounter{row}{2}
    \setrow {3,3, , , }
    \setrow { ,3, , , }
    \setrow { ,3,3,3,3}
    \setrow { , , , ,3}
    \setrow { ,3,3,3,3}
    \setrow { ,3, , , }
    \setrow {3,3, , , };

    \end{tikzpicture}

(b) Positive (resp. Neg.) semi-merge in $Z_1$\\
~~\\

\begin{tikzpicture}[scale=.4]

    \draw[line width=0.1mm] (0, 0) grid (6, 6);
    \draw[line width=0.6mm] (0,0) -- (0,6); 
    \draw[line width=0.6mm] (3,0) -- (3,6); 
    \draw[line width=0.6mm] (6,0) -- (6,6); 
    \draw[line width=0.6mm] (0,0) -- (6,0);
    \draw[line width=0.6mm] (0,2) -- (6,2);
    \draw[line width=0.6mm] (0,4) -- (6,4);
    \draw[line width=0.6mm] (0,6) -- (6,6);

    \setcounter{cols}{0}
    \setcounter{rows}{6}
    \setcounter{row}{0}
    \setrow {3,3,3,3}
    \setrow { , , ,3}
    \setrow { , , ,3}
    \setrow { , , ,3}
    \setrow {3,3,3,3};

\draw[line width=0.1mm] (8, 0) grid (14, 6);
\draw[line width=0.6mm] (8,0) -- (8,6); 
\draw[line width=0.6mm] (11,0) -- (11,6); 
\draw[line width=0.6mm] (14,0) -- (14,6); 
\draw[line width=0.6mm] (8,0) -- (14,0);
\draw[line width=0.6mm] (8,2) -- (14,2);
\draw[line width=0.6mm] (8,4) -- (14,4);
\draw[line width=0.6mm] (8,6) -- (14,6);

\setcounter{cols}{8}
\setcounter{rows}{6}
\setcounter{row}{1}

    \setrow {3,3,3,3}
    \setrow { , , ,3}
    \setrow { , , ,3}
    \setrow { , , ,3}
    \setrow {3,3,3,3};
  
    \end{tikzpicture}

    (c) Positive (resp. Neg.) semi-merge in $Z_2$\\
\end{minipage}
\caption{Bridge merging. On the left, the gadget is shown using the standard notation we’ve been using for bridges and the $Z$-blocks involved in the structure. A bridge $T_1$ connecting $B_1$ and $B_3$ and $T_2$ connecting $B_2$ and $B_3$. On the right, particular cases of the gadget are illustrated. The positive version of the gadget is always shown on the left of each subfigure.}
\label{fig:merg_p}
\end{figure}

As the semi-duplicator, the functionality of the structure follows as a corollary of \Cref{lem:funcbridge}, it states that this gadget, formed by two bridges with the same endcell, can carry a signal to that cell independently:

\begin{corollary}   
Let $(M,c)=T_1+T_2$ be a positive semi-merge in $(B_1,B_2,B_3).$ If $T_1$ or $T_2$ are source bridges, then $F(c,Z)(B_3^+) =4$.
\end{corollary} 

\begin{proof}
  As a direct application of \Cref{lem:funcbridge} to both bridges $T_1$ and $T_2,$ we have that $F(c,Z)(B_3^+)=4$.
\end{proof}

\subsection{Semi-crossings}
We describe in this section a way to cross semi-wires of distinct polarity with a gadget we call semi-crossing. This gadget will be a key feature in the construction of a wire crossing gadget, as it will allow a positive (resp. negative) signal to travel through negative bridges without allowing them to retain a signal after the cycle.

\begin{figure}
\begin{minipage}[t][0.35\textheight][c]{0.5\textwidth}
\centering
\input{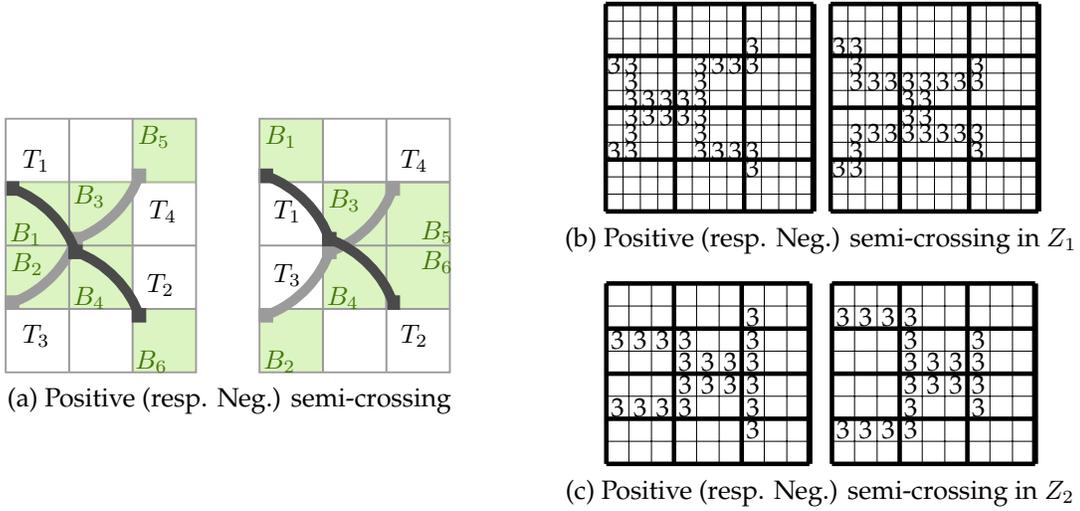}\\
(a) Positive (resp. Neg.) semi-crossing
\end{minipage}
\begin{minipage}[t][0.35\textheight][c]{0.5\textwidth}
\centering
\begin{tikzpicture}[scale=.23]

  \begin{scope}
    
    \draw[line width=0.1mm] (0, 0) grid (12, 12);
    \draw[line width=0.6mm] (0,0) -- (0,12); 
    \draw[line width=0.6mm] (4,0) -- (4,12); 
    \draw[line width=0.6mm] (8,0) -- (8,12); 
    \draw[line width=0.8mm] (12,0) -- (12,12); 
        
    \draw[line width=0.6mm] (0,0) -- (12,0);
    \draw[line width=0.6mm] (0,3) -- (12,3);
    \draw[line width=0.6mm] (0,6) -- (12,6);
    \draw[line width=0.6mm] (0,9) -- (12,9);
    \draw[line width=0.6mm] (0,12) -- (12,12);

    \setcounter{cols}{0}
    \setcounter{rows}{12}
    \setcounter{row}{2}
    \setrow { , , , , , , , ,3}
    \setrow {3,3, , , ,3,3,3,3}
    \setrow { ,3, , , ,3, , , }
    \setrow { ,3,3,3,3,3, , , }
    \setrow { ,3,3,3,3,3, , , }
    \setrow { ,3, , , ,3, , , }
    \setrow {3,3, , , ,3,3,3,3}
    \setrow { , , , , , , , ,3};
  \end{scope}

  \begin{scope}
\draw[line width=0.1mm] (13, 0) grid (25, 12);
\draw[line width=0.6mm] (13,0) -- (13,12); 
\draw[line width=0.6mm] (17,0) -- (17,12); 
\draw[line width=0.6mm] (21,0) -- (21,12); 
\draw[line width=0.8mm] (25,0) -- (25,12); 

\draw[line width=0.6mm] (13,0) -- (25,0);
\draw[line width=0.6mm] (13,3) -- (25,3);
\draw[line width=0.6mm] (13,6) -- (25,6);
\draw[line width=0.6mm] (13,9) -- (25,9);
\draw[line width=0.6mm] (13,12) -- (25,12);
    \setcounter{cols}{13}
    \setcounter{row}{2}
    \setrow {3,3, , , , , , , }
    \setrow { ,3, , , , , , ,3}
    \setrow { ,3,3,3,3,3,3,3,3}
    \setrow { , , , ,3,3, , , }
    \setrow { , , , ,3,3, , , }
    \setrow { ,3,3,3,3,3,3,3,3}
    \setrow { ,3, , , , , , ,3}
    \setrow {3,3, , , , , , , };
  \end{scope}

    \end{tikzpicture}
    
(b) Positive (resp. Neg.) semi-crossing in $Z_1$\\
~~\\

\begin{tikzpicture}[scale=.3]

  \begin{scope}
    
    \draw[line width=0.1mm] (0, 0) grid (9, 8);
    \draw[line width=0.6mm] (0,0) -- (0,8); 
    \draw[line width=0.6mm] (3,0) -- (3,8); 
    \draw[line width=0.6mm] (6,0) -- (6,8); 
    \draw[line width=0.8mm] (9,0) -- (9,8); 
        
    \draw[line width=0.6mm] (0,0) -- (9,0);
    \draw[line width=0.6mm] (0,2) -- (9,2);
    \draw[line width=0.6mm] (0,4) -- (9,4);
    \draw[line width=0.6mm] (0,6) -- (9,6);
    \draw[line width=0.6mm] (0,8) -- (9,8);
    
    \setcounter{cols}{0}
    \setcounter{rows}{8}
    \setcounter{row}{0}
    \setrow { , , , , , , , , }
    \setrow { , , , , , ,3, , }
    \setrow {3,3,3,3, , ,3, , }
    \setrow { , , ,3,3,3,3, , }
    \setrow { , , ,3,3,3,3, , }
    \setrow {3,3,3,3, , ,3, , }
    \setrow { , , , , , ,3, , }
    \setrow { , , , , , , , , };
  \end{scope}

  \begin{scope}
    
    \draw[line width=0.1mm] (10, 0) grid (19, 8);
    \draw[line width=0.6mm] (10,0) -- (10,8); 
    \draw[line width=0.6mm] (13,0) -- (13,8); 
    \draw[line width=0.6mm] (16,0) -- (16,8); 
    \draw[line width=0.8mm] (19,0) -- (19,8); 
            
    \draw[line width=0.6mm] (10,0) -- (19,0);
    \draw[line width=0.6mm] (10,2) -- (19,2);
    \draw[line width=0.6mm] (10,4) -- (19,4);
    \draw[line width=0.6mm] (10,6) -- (19,6);
    \draw[line width=0.6mm] (10,8) -- (19,8);

    \setcounter{cols}{10}
    \setcounter{rows}{8}
    \setcounter{row}{0}
    \setrow { , , , , , , , , }

    \setrow {3,3,3,3, , , , , }
    \setrow { , , ,3, , ,3, , }    
    \setrow { , , ,3,3,3,3, , }
    \setrow { , , ,3,3,3,3, , }
    \setrow { , , ,3, , ,3, , }
    \setrow {3,3,3,3, , , , , };
  \end{scope}

    \end{tikzpicture}
    
(c) Positive (resp. Neg.) semi-crossing in $Z_2$\\
~~\\
\end{minipage}
\caption{Bridge semi-crossing. On the left, the gadget is shown using the standard notation we’ve been using for bridges and the $Z$-blocks involved in the structure. Bridges $T_1,T_2$ connecting $B_1$ to $B_4$ and $B_4$ to $B_6$; Bridges $T_3,T_4$ connecting $B_2$ to $B_3$ and $B_3$ to $B_5$. On the right, particular cases of the gadget are illustrated. The positive version of the gadget is always shown on the left.}
\label{fig:semicr_p}
\end{figure}

\begin{figure}
\centering
\begin{tikzpicture}[scale=.25]

  \begin{scope}
    
    \draw[line width=0.1mm] (0, 0) grid (12, 12);
    \draw[line width=0.6mm] (0,0) -- (0,12); 
    \draw[line width=0.6mm] (4,0) -- (4,12); 
    \draw[line width=0.6mm] (8,0) -- (8,12); 
    \draw[line width=0.8mm] (12,0) -- (12,12); 
        
    \draw[line width=0.6mm] (0,0) -- (12,0);
    \draw[line width=0.6mm] (0,3) -- (12,3);
    \draw[line width=0.6mm] (0,6) -- (12,6);
    \draw[line width=0.6mm] (0,9) -- (12,9);
    \draw[line width=0.6mm] (0,12) -- (12,12);

    \setcounter{cols}{0}
    \setcounter{rows}{12}
    \setcounter{row}{2}
    \setrow { , , , , , , , ,3}
    \setrow {4,3, , , ,3,3,3,3}
    \setrow { ,3, , , ,3, , , }
    \setrow { ,3,3,3,3,3, , , }
    \setrow { ,3,3,3,3,3, , , }
    \setrow { ,3, , , ,3, , , }
    \setrow {3,3, , , ,3,3,3,3}
    \setrow { , , , , , , , ,3};
  \end{scope}

  \begin{scope}
\draw[line width=0.1mm] (14, 0) grid (26, 12);
\draw[line width=0.6mm] (14,0) -- (14,12); 
\draw[line width=0.6mm] (18,0) -- (18,12); 
\draw[line width=0.6mm] (22,0) -- (22,12); 
\draw[line width=0.8mm] (26,0) -- (26,12); 

\draw[line width=0.6mm] (14,0) -- (26,0);
\draw[line width=0.6mm] (14,3) -- (26,3);
\draw[line width=0.6mm] (14,6) -- (26,6);
\draw[line width=0.6mm] (14,9) -- (26,9);
\draw[line width=0.6mm] (14,12) -- (26,12);
    \setcounter{cols}{14}
    \setcounter{row}{2}
    \setrow { ,1, , , , , , ,3}
    \setrow {2,3, , , ,3,3,3,3}
    \setrow { ,2, , ,1,4, , , }
    \setrow {1,3,3,3,2,2, , , }
    \setrow { ,3,3,2,2,4, , , }
    \setrow { ,3, ,1, ,2, , , }
    \setrow {3,3, , , ,4,3,3,3}
    \setrow { , , , , , , , ,3};
  \end{scope}

  \begin{scope}
\draw[line width=0.1mm] (28, 0) grid (40, 12);
\draw[line width=0.6mm] (28,0) -- (28,12); 
\draw[line width=0.6mm] (32,0) -- (32,12); 
\draw[line width=0.6mm] (36,0) -- (36,12); 
\draw[line width=0.8mm] (40,0) -- (40,12); 

\draw[line width=0.6mm] (28,0) -- (40,0);
\draw[line width=0.6mm] (28,3) -- (40,3);
\draw[line width=0.6mm] (28,6) -- (40,6);
\draw[line width=0.6mm] (28,9) -- (40,9);
\draw[line width=0.6mm] (28,12) -- (40,12);
    \setcounter{cols}{28}
    \setcounter{row}{2}
    \setrow { ,1, , , , , , ,3}
    \setrow {2,3, , , ,3,3,3,3}
    \setrow { ,2, , ,2,2,1, , }
    \setrow {1,3,3,3,2,2, , , }
    \setrow { ,3,3,2,3,2,1, , }
    \setrow { ,3, ,1, ,2, , ,1}
    \setrow {3,3, , ,1,3,3,2,2}
    \setrow { , , , , , , , ,4};
  \end{scope}
    
    \draw (28,-0.5) node [anchor=north west,scale=0.8][inner sep=0.75pt]    {$F(c,Z_1,14) = F^2(c,Z_1)$};

    \draw (5.5,-0.5) node [anchor=north west,scale=0.8][inner sep=0.75pt]    {$c$};

    \draw (17.5,-0.5) node [anchor=north west,scale=0.8][inner sep=0.75pt]    {$F(c,Z_1,8)$};

    \end{tikzpicture}
\caption{Example of the positive semi-crossing functionality for the word $Z_1 = HVVHHHV$. At the configuration $F(c, Z_1, 8)$, two signals coexist with the same vertical parity; however, as intended, in $F^2(c, Z_1)$ only the signal at the corresponding output remains.}
\label{fig:semicrsvpp}
\end{figure}
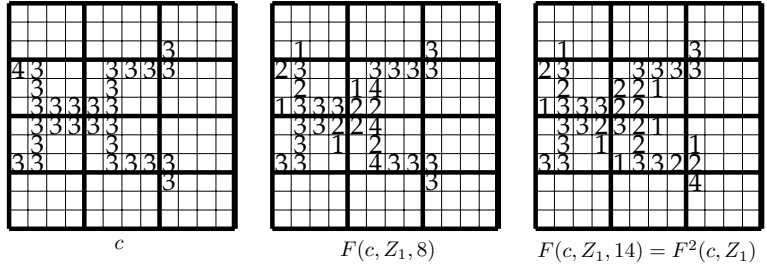

Consider 6 blocks $(B_1,...,B_6)$ included in a $4\times 3$ block shape as in \Cref{fig:semicr_p}, positive bridges $T_1,T_2$ connecting $B_1,B_4$ and $B_4,B_6$ and negative bridges $T_3,T_4$ connecting $B_2,B_3$ and $B_3,B_5$. We define a positive semi-crossing as the bridge gadget $(Q,c) = T_1+T_2+T_3+T_4$. We call $B_1^+,B_2^-$ the inputs of the semi-crossing, and $B_5^-,B_6^+$ its outputs. Its functionality follows from the following lemma:

\begin{lemma}\label{lem:switch}
Let $(Q,c) = T_1+T_2+T_3+T_4$ be a positive (resp. negative) semi-crossing with blocks $(B_1,...,B_6).$  The following conditions hold:

\begin{itemize}
    \item If $T_1$ is a source bridge, then $F^2(c,Z)(B_6^+) =4$ and $F^i(c,Z)(B_5^-) \neq 4, \forall i$
    \item If $T_3$ is a source bridge, then $F^2(c,Z)(B_5^-) =4$ and $F^i(c,Z)(B_6^+) \neq 4, \forall i$
\end{itemize} 
\end{lemma}

For the proof of the previous lemma, we need a property related to the "parity" of signals. For a configuration \(c\), we say that \(c\) has the \emph{vertical parity property} (vpp) if, for every pair of signals \(x, y\) in \(c\), we have that, \(x - y = (a, 2b)\) for some \(a, b \in \mathbb{Z}\). Observe that this property is preserved by the dynamics. In fact we have the following lemma:

\begin{lemma}\label{lem:vpp}
Let $c$ be a configuration with the vpp. Then, $F(c,H),F(c,V)$ have the vpp.
\end{lemma}

\begin{proof}
Given a configuration $c$ with the vpp, we have:

1. Since no row without a signal can gain one over transitions \(F(\cdot, H)\), the configuration \(F(c, H)\) has the vpp 

2. Since the vpp holds on $c$, only the rows of a given parity contain signals; assume these are all columns with even vertical coordinate. Through a transition $F(\cdot,V)$, all signals topple and distribute grains to columns of odd parity. Because the columns of odd parity initially contained no signals, all columns with even vertical parity on $F(c,V)$ are signal-free, and therefore the vpp holds.

\end{proof}

Now we proceed with the proof of \Cref{lem:switch}.
\begin{proof} (of \Cref{lem:switch})
    We need to show that that the signal will not duplicate and arrive to $B_5^-$ \textbf{and} $B_6^+$ at the end of the second cycle. 

    If $c(B_1^+) = 4$, from \Cref{lem:funcbridge2}, we have $F^2(c,Z)(B_6^+) =4$. We note from the structure that $B_5^- = B_6^+ + 2vV + 1$, which says $B_5^-, B_6^+$ have different vertical parity. Since at $c$ all signals have the same parity, and $F^2(c,Z)(B_6^+) =4$, we conclude $F^2(c,Z)(B_5^+) \neq 4$.
\end{proof}

We define the negative semi-crossing gadget as shown in \Cref{fig:semicr_p}. The difference between the polarities of the upper input and the bottom one: the positive semi-crossing has the positive signal input on top, while the negative semi-crossing has it below. The functionality of the negative semi-crossing is analogous, and we do not provide a proof here.

\subsection{Switches}
In this section we define a gadget called switch gadget. Roughly, this gadget will function as a positive semi-crossing gadget with one key difference: It will have only one output, of its respective polarity, which will be activated by a signal only if the opposite input has received a signal.

\begin{figure}[ht]
\begin{minipage}[t][0.35\textheight][c]{0.5\textwidth}
\centering
\input{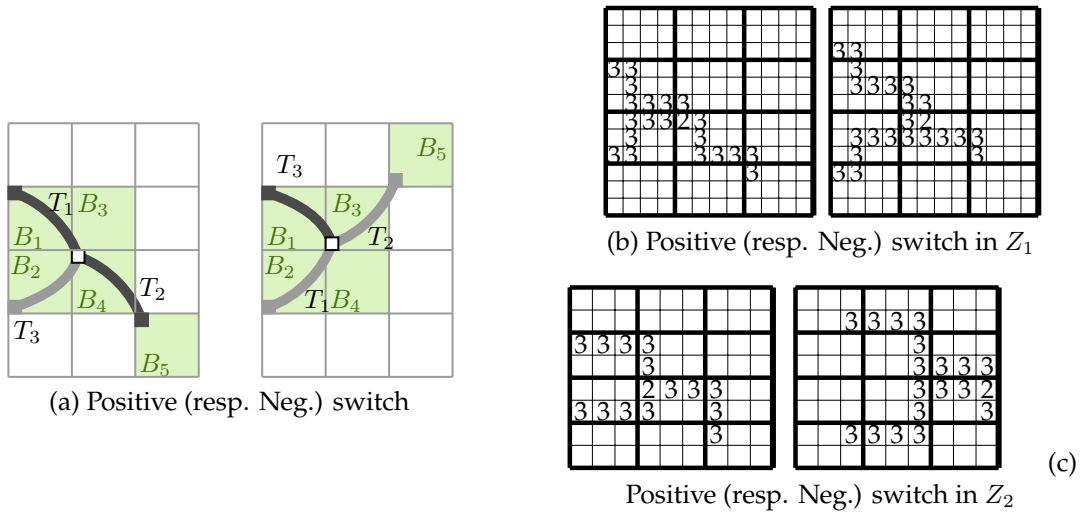}\\
(a) Positive (resp. Neg.) switch
\end{minipage}
\begin{minipage}[t][0.35\textheight][c]{0.5\textwidth}
\centering
\begin{tikzpicture}[scale=.23]

  \begin{scope}
    
    \draw[line width=0.1mm] (0, 0) grid (12, 12);
    \draw[line width=0.6mm] (0,0) -- (0,12); 
    \draw[line width=0.6mm] (4,0) -- (4,12); 
    \draw[line width=0.6mm] (8,0) -- (8,12); 
    \draw[line width=0.8mm] (12,0) -- (12,12); 
        
    \draw[line width=0.6mm] (0,0) -- (12,0);
    \draw[line width=0.6mm] (0,3) -- (12,3);
    \draw[line width=0.6mm] (0,6) -- (12,6);
    \draw[line width=0.6mm] (0,9) -- (12,9);
    \draw[line width=0.6mm] (0,12) -- (12,12);

    \setcounter{cols}{0}
    \setcounter{rows}{12}
    \setcounter{row}{2}
    \setrow { , , , , , , , , }
    \setrow {3,3, , , , , , , }
    \setrow { ,3, , , , , , , }
    \setrow { ,3,3,3,3, , , , }
    \setrow { ,3,3,3,2,3, , , }
    \setrow { ,3, , , ,3, , , }
    \setrow {3,3, , , ,3,3,3,3}
    \setrow { , , , , , , , ,3};

\draw[line width=0.1mm] (13, 0) grid (25, 12);
\draw[line width=0.6mm] (13,0) -- (13,12); 
\draw[line width=0.6mm] (17,0) -- (17,12); 
\draw[line width=0.6mm] (21,0) -- (21,12); 
\draw[line width=0.8mm] (25,0) -- (25,12); 

\draw[line width=0.6mm] (13,0) -- (25,0);
\draw[line width=0.6mm] (13,3) -- (25,3);
\draw[line width=0.6mm] (13,6) -- (25,6);
\draw[line width=0.6mm] (13,9) -- (25,9);
\draw[line width=0.6mm] (13,12) -- (25,12);

    \setcounter{cols}{13}
    \setcounter{row}{2}
    \setrow {3,3, , , , , , , }
    \setrow { ,3, , , , , , , }
    \setrow { ,3,3,3,3, , , , }
    \setrow { , , , ,3,3, , , }
    \setrow { , , , ,3,2, , , }
    \setrow { ,3,3,3,3,3,3,3,3}
    \setrow { ,3, , , , , , ,3}
    \setrow {3,3, , , , , , , };
  \end{scope}

    \end{tikzpicture}

(b) Positive (resp. Neg.) switch in $Z_1$\\
~~\\

\begin{tikzpicture}[scale=.3]

  \begin{scope}
    
    \draw[line width=0.1mm] (0, 0) grid (9, 8);
    \draw[line width=0.6mm] (0,0) -- (0,8); 
    \draw[line width=0.6mm] (3,0) -- (3,8); 
    \draw[line width=0.6mm] (6,0) -- (6,8); 
    \draw[line width=0.8mm] (9,0) -- (9,8); 
        
    \draw[line width=0.6mm] (0,0) -- (9,0);
    \draw[line width=0.6mm] (0,2) -- (9,2);
    \draw[line width=0.6mm] (0,4) -- (9,4);
    \draw[line width=0.6mm] (0,6) -- (9,6);
    \draw[line width=0.6mm] (0,8) -- (9,8);

    \setcounter{cols}{0}
    \setcounter{rows}{8}
    \setcounter{row}{0}
    \setrow { , , , , , , , , }
    \setrow { , , , , , , , , }
    \setrow {3,3,3,3, , , , , }
    \setrow { , , ,3, , , , , }
    \setrow { , , ,2,3,3,3, , }
    \setrow {3,3,3,3, , ,3, , }
    \setrow { , , , , , ,3, , }
    \setrow { , , , , , , , , };
    
\draw[line width=0.1mm] (10, 0) grid (19, 8);
\draw[line width=0.6mm] (10,0) -- (10,8); 
\draw[line width=0.6mm] (13,0) -- (13,8); 
\draw[line width=0.6mm] (16,0) -- (16,8); 
\draw[line width=0.8mm] (19,0) -- (19,8); 

\draw[line width=0.6mm] (10,0) -- (19,0);
\draw[line width=0.6mm] (10,2) -- (19,2);
\draw[line width=0.6mm] (10,4) -- (19,4);
\draw[line width=0.6mm] (10,6) -- (19,6);
\draw[line width=0.6mm] (10,8) -- (19,8);

    \setcounter{cols}{12}
    \setcounter{rows}{8}
    \setcounter{row}{0}
    \setrow { , , , , , , , , }
    \setrow {3,3,3,3, , , , , }
    \setrow { , , ,3, , , , , }    
    \setrow { , , ,3,3,3,3, , }
    \setrow { , , ,3,3,3,2, , }
    \setrow { , , ,3, , ,3, , }
    \setrow {3,3,3,3, , , , , };
  \end{scope}

    \end{tikzpicture}
(c) Positive (resp. Neg.) switch in $Z_2$\\
~~\\
\end{minipage}
\caption{Switch. On the left, the gadget is shown using the standard notation we’ve been using for bridges and the $Z$-blocks involved in the structure. Bridges $T_1,T_2$ connecting $B_1$ to $B_4$ and $B_4$ to $B_5$ in the positive switch; Bridges $T_1,T_2$ connecting $B_2$ to $B_3$ and $B_3$ to $B_5$ in the negative switch; Bridge with a sink (denoted as a white dot) $T_3$ connecting $B_2$ to $B_3$ in the positive switch; Bridge with a sink (denoted as a white dot) $T_3$ connecting $B_1$ to $B_4$ in the negative switch;. On the right, particular cases of the gadget are illustrated. The positive version of the gadget is always shown on the left. Note that in this case, bridge $T_3$ does not end at the block it would normally reach if it were not a bridge with a sink.}
\label{fig:swich}
\end{figure}

The construction of the positive and negative switches differs slightly but follows the same principle. The polarity of the switch is determined by the polarities of the upper bridges. Like the semi-crossing, its inputs/outputs have different polarities.

\subsubsection{Positive switch}

For the positive switch, we arrange 5 blocks $(B_1,...,B_5)$ included in a $4\times 3$ block shape as in \Cref{fig:swich}, positive bridges $T_1,T_2$ connecting $B_1,B_4$ and $B_4,B_5$ and a negative bridge $T_3 = (P_3,c_3)$ with a sink at $k-1$ connecting $B_2,B_3$. We define a positive off switch as the bridge gadget $(S,c_{\text{off}}) = T_1+T_2+T_3$. We call $B_1^+,B_2^-$ the inputs of the switch, and $B_5^+$ its output.

If the upper input receives a signal, it will travel to the output iff the lower input has already received a signal before. We show the functionality for the positive switch, with the negative one being analogous.

\begin{lemma}\label{lem:funcswitch}
~\\
\begin{itemize}
    \item Given $(S,c_{\text{off}}) = T_1+T_2+T_3$ a positive off switch at $(B_1,...,B_5)$. If $T_1$ is a source bridge, then $F^2(c_{\text{off}},Z)$ has no signals.
    \item Let $(S,c_{\text{on}}) = T_1+T_2+T_3$ the gadget defined as $S(S,c_{\text{off}})$ except $T_3$ is a source bridge with a sink at $k-1$, then $F^2(c_{on},Z)$ satisfies the conditions of \Cref{lem:funcbridge2} for bridges $T_1,T_2$
\end{itemize}
\end{lemma}

\begin{proof}
    Both parts follows from the functionality of bridges with a sink, since every cell at \(N(P_1)\setminus P_1\) and \(N(P_3)\setminus P_3\) in both $c_{on},c_{off}$ has value 0 and $T_1,T_3$ are bridges with sinks. 
    
    (1) follows directly from \Cref{lem:funcbridgesink}.
    
    For (2), \Cref{lem:funcbridgesink} ensures $F(c_{\text{on}},Z)(P_3(k-1)) = 3$. Outside of $P_3(k)$ and $P_3(k-1)$, no cell from $P_1,P_2$ is at $N(P_3)$, ensuring no cell at $P_1,P_2$ outside $P_{3}(k-1)$ gains a grain during the cicle (remeber by definition of the switch $P_{3}(k-1) = P_1(k)$).
    
\end{proof}

\subsubsection{Negative switch}

For the negative switch, given 5 blocks $(B_1,...,B_5)$ included in a $4\times 3$ block shape as in \Cref{fig:swich}, negative bridges $T_1,T_2$ connecting $B_2,B_3$ and $B_3,B_5$ and a positive bridge $T_3 = (P_3,c_3)$ with a sink at $k-1$ connecting $B_1,B_4$. We define a positive off switch as the bridge gadget $(S,c_{\text{off}}) = T_1+T_2+T_3$. We call $B_1^+,B_2^-$ the inputs of the switch, and $B_5^-$ its output.

Functionality of negative switches is analogous to the positive one, having:
\begin{lemma}\label{lem:funcswitchneg}
~\\
\begin{itemize}
    \item Given $(S,c_{\text{off}}) = T_1+T_2+T_3$ a negative off switch at $(B_1,...,B_5)$ except $T_1$ is a source bridge, then $F^2(c_{\text{off}},Z)$ has no signals.
    \item Let $(S,c_{\text{on}}) = T_1+T_2+T_3$ the gadget defined as $S(S,c_{\text{off}})$ except $T_3$ is a source bridge with a sink at $k-1$, then $F^2(c_{on},Z)$ satisfies the conditions of \Cref{lem:funcbridge2} for bridges $T_1,T_2$
\end{itemize}
\end{lemma}

Observe that the case in which two signals arrive at the same time is not interesting to analyze since we can control in our construction that this particular case never happens.

\subsection{Connecting gadgets}\label{sec:connecting}
\fxnote{Shouldn't this come after retarders?}

\subsubsection{Shifted cycles}
We will assume that the words we are working with are not cyclic, that is, words \(Z\) that cannot be decomposed as \(Z = W^i\) with \(W \in \{H, V\}^*\). This assumption is made without loss of generality, because if we wanted to perform the reduction for cyclic words, we could simply perform the reduction for the subword, as the dynamics behave the same. We will make this assumption to ensure that the structure responsible for eliminating uncoordinated signals functions correctly. Let $Z$ be a word and $c$ an initial configuration of a gadget. We define an \emph{uncoordinated signal} as any signal that appears at the input of a bridge (belonging to the gadget) in a configuration different from $F^i(c, Z)$ for some positive integer $i$, i.e., at a time other than the start of a cycle. In general, this kind of signal has proven not to be problematic. Moreover, with a careful construction adapted to each word, it may be possible to ensure that each part of the final construction avoids the presence of such signals. However, for completeness and to simplify the overall functionality, we proceed with the following construction.

Consider a subdivision of the word \(Z\) such that \(Z = XY\), with \(X, Y\) non-empty, and the word \(W = Y|X\), named \textit{shifted $Z$ cycle}. We note that the dynamics of \(W\) are the same as those of \(Z\), except that it starts out of sync. An example of the uncoordinated cycle can be seen in \Cref{fig:coord_problem}, and in the same figure, it can be observed how the bridge ceases to function properly (although in that case, we only have a pure bridge). We will show that the signals cannot move through the semi-wires following the dynamics of \(W\), which demonstrates that we do not need to worry about uncoordinated signals, as they will not arrive its destination.

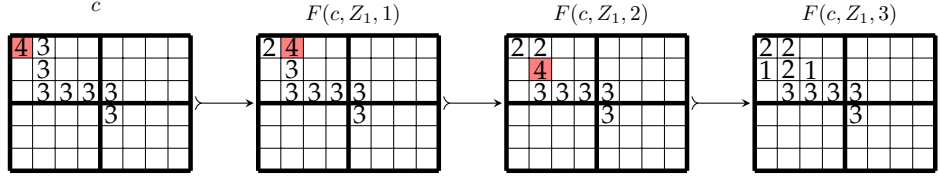
\begin{figure}
\centering
\begin{tikzpicture}[scale=.3]

    \draw  [draw opacity=0][fill={rgb, 255:red, 255; green, 0; blue,  0 }  ,fill opacity=0.5 ] (0,15) -- (1,15) -- (1,14) -- (0,14) -- cycle;
    
    \draw [draw opacity=0, fill={rgb, 255:red, 255; green, 0; blue, 0}, fill opacity=0.5] (12,15) -- (13,15) -- (13,14) -- (12,14) -- cycle;

    \draw [draw opacity=0, fill={rgb, 255:red, 255; green, 0; blue, 0}, fill opacity=0.5] (23,14) -- (24,14) -- (24,13) -- (23,13) -- cycle;

    \draw[line width=0.1mm] (0, 9) grid (8, 15);
    \draw[line width=0.6mm] (0,9) -- (0,15); 
    \draw[line width=0.6mm] (4,9) -- (4,15); 
    \draw[line width=0.6mm] (8,9) -- (8,15); 
    \draw[line width=0.6mm] (0,9) -- (8,9);
    \draw[line width=0.6mm] (0,12) -- (8,12);
    \draw[line width=0.6mm] (0,15) -- (8,15);
    \setcounter{cols}{0}
    \setcounter{rows}{15}
    \setcounter{row}{0}
    \setrow {4,3, , , }
    \setrow { ,3, , , }
    \setrow { ,3,3,3,3}
    \setrow { , , , ,3};
    
    \draw[>- stealth] (8.2,12) -- (10.8,12);

    \draw[line width=0.1mm] (11, 9) grid (19, 15);
    \draw[line width=0.6mm] (11,9) -- (11,15); 
    \draw[line width=0.6mm] (15,9) -- (15,15); 
    \draw[line width=0.6mm] (19,9) -- (19,15); 
    \draw[line width=0.6mm] (11,9) -- (19,9);
    \draw[line width=0.6mm] (11,12) -- (19,12);
    \draw[line width=0.6mm] (11,15) -- (19,15);

    \setcounter{cols}{11}
    \setcounter{rows}{15}
    \setcounter{row}{0}
    \setrow {2,4, , , }
    \setrow { ,3, , , }
    \setrow { ,3,3,3,3}
    \setrow { , , , ,3};

    \draw[>- stealth] (19.2,12) -- (21.8,12);

    \draw[line width=0.1mm] (22, 9) grid (30, 15);
    \draw[line width=0.6mm] (22,9) -- (22,15); 
    \draw[line width=0.6mm] (26,9) -- (26,15); 
    \draw[line width=0.6mm] (30,9) -- (30,15); 
    \draw[line width=0.6mm] (22,9) -- (30,9);
    \draw[line width=0.6mm] (22,12) -- (30,12);
    \draw[line width=0.6mm] (22,15) -- (30,15);

    \setcounter{cols}{22}
    \setcounter{rows}{15}
    \setcounter{row}{0}
    \setrow {2,2, , , }
    \setrow { ,4, , , }
    \setrow { ,3,3,3,3}
    \setrow { , , , ,3};

    \draw[>- stealth] (30.2,12) -- (32.8,12);

    \draw[line width=0.1mm] (33, 9) grid (41, 15);
    \draw[line width=0.6mm] (33,9) -- (33,15); 
    \draw[line width=0.6mm] (37,9) -- (37,15); 
    \draw[line width=0.6mm] (41,9) -- (41,15); 
    \draw[line width=0.6mm] (33,9) -- (41,9);
    \draw[line width=0.6mm] (33,12) -- (41,12);
    \draw[line width=0.6mm] (33,15) -- (41,15);

    \setcounter{cols}{33}
    \setcounter{rows}{15}
    \setcounter{row}{0}
    \setrow {2,2, , , }
    \setrow {1,2,1, , }
    \setrow { ,3,3,3,3}
    \setrow { , , , ,3};

\setcounter{cols}{0}

    \draw (3.5,16.6) node [anchor=north west,scale=0.8][inner sep=0.75pt]    {$c$};
    
    \draw (35,16.5) node [anchor=north west,scale=0.8][inner sep=0.75pt]    {$F(c,Z_1,3)$};

    \draw (13,16.5) node [anchor=north west,scale=0.8][inner sep=0.75pt]    {$F(c,Z_1,1)$};

    \draw (24,16.5) node [anchor=north west,scale=0.8][inner sep=0.75pt]    {$F(c,Z_1,2)$};

\end{tikzpicture}
\caption{Bridge failure under updates following a shifted cycle. The signals in the configuration are highlighted in red. Recalling  \Cref{fig:ciclo}, which indeed shows a cycle for the word $Z_1 = HVVHHHV$, we observe that if we consider the sequence of updates along the shifted cycle $W_1 = HV|HVVHH$, then in $F(W_1,3)$ there are no longer any signals in the configuration. We note that 3 is the smallest index such that $Z_1(i) \neq W_1(i)$, and this is not a coincidence—it is related to the way the bridges were constructed and will be an important aspect in the study of coordinating semi-wires in this section.}
\label{fig:coord_problem}
\end{figure}

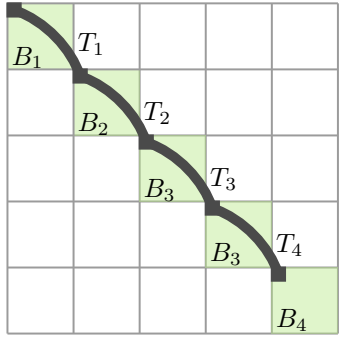
\begin{figure}
\centering
\tikzset{every picture/.style={line width=0.75pt}} 

\begin{tikzpicture}[x=0.5pt,y=0.5pt,yscale=-1,xscale=1]

\draw [color={rgb, 255:red, 155; green, 155; blue, 155 }  ,draw opacity=1 ][line width=0.75]    (300,150) -- (50,150) ;
\draw [color={rgb, 255:red, 155; green, 155; blue, 155 }  ,draw opacity=1 ][line width=0.75]    (300,200) -- (50,200) ;
\draw [color={rgb, 255:red, 155; green, 155; blue, 155 }  ,draw opacity=1 ][line width=0.75]    (300,250) -- (50,250) ;
\draw [color={rgb, 255:red, 155; green, 155; blue, 155 }  ,draw opacity=1 ][line width=0.75]    (300,300) -- (50,300) ;
\draw [color={rgb, 255:red, 155; green, 155; blue, 155 }  ,draw opacity=1 ][line width=0.75]    (300,50) -- (300,300) ;
\draw [color={rgb, 255:red, 155; green, 155; blue, 155 }  ,draw opacity=1 ][line width=0.75]    (150,50) -- (150,300) ;
\draw [color={rgb, 255:red, 155; green, 155; blue, 155 }  ,draw opacity=1 ][line width=0.75]    (200,50) -- (200,300) ;
\draw [color={rgb, 255:red, 155; green, 155; blue, 155 }  ,draw opacity=1 ][line width=0.75]    (250,50) -- (250,300) ;
\draw  [draw opacity=0][fill={rgb, 255:red, 184; green, 233; blue, 134 }  ,fill opacity=0.43 ] (100,100) -- (150,100) -- (150,150) -- (100,150) -- cycle ;
\draw  [draw opacity=0][fill={rgb, 255:red, 184; green, 233; blue, 134 }  ,fill opacity=0.43 ] (150,150) -- (200,150) -- (200,200) -- (150,200) -- cycle ;
\draw  [draw opacity=0][fill={rgb, 255:red, 184; green, 233; blue, 134 }  ,fill opacity=0.43 ] (50,50) -- (100,50) -- (100,100) -- (50,100) -- cycle ;
\draw [color={rgb, 255:red, 155; green, 155; blue, 155 }  ,draw opacity=1 ][line width=0.75]    (100,50) -- (100,300) ;
\draw [color={rgb, 255:red, 74; green, 74; blue, 74 }  ,draw opacity=1 ][line width=3.75]    (105,105) .. controls (129.2,113.5) and (148.2,134.5) .. (155,155) ;
\draw [color={rgb, 255:red, 155; green, 155; blue, 155 }  ,draw opacity=1 ][line width=0.75]    (50,50) -- (50,300) ;
\draw [color={rgb, 255:red, 155; green, 155; blue, 155 }  ,draw opacity=1 ][line width=0.75]    (300,50) -- (50,50) ;
\draw [color={rgb, 255:red, 155; green, 155; blue, 155 }  ,draw opacity=1 ][line width=0.75]    (300,100) -- (50,100) ;
\draw  [color={rgb, 255:red, 74; green, 74; blue, 74 }  ,draw opacity=1 ][fill={rgb, 255:red, 74; green, 74; blue, 74 }  ,fill opacity=1 ] (50,50) -- (60,50) -- (60,60) -- (50,60) -- cycle ;
\draw  [color={rgb, 255:red, 74; green, 74; blue, 74 }  ,draw opacity=1 ][fill={rgb, 255:red, 74; green, 74; blue, 74 }  ,fill opacity=1 ] (150,150) -- (160,150) -- (160,160) -- (150,160) -- cycle ;
\draw  [color={rgb, 255:red, 74; green, 74; blue, 74 }  ,draw opacity=1 ][fill={rgb, 255:red, 74; green, 74; blue, 74 }  ,fill opacity=1 ] (100,100) -- (110,100) -- (110,110) -- (100,110) -- cycle ;
\draw [color={rgb, 255:red, 74; green, 74; blue, 74 }  ,draw opacity=1 ][line width=3.75]    (55,55) .. controls (79.2,63.5) and (98.2,84.5) .. (105,105) ;
\draw  [draw opacity=0][fill={rgb, 255:red, 184; green, 233; blue, 134 }  ,fill opacity=0.43 ] (200,200) -- (250,200) -- (250,250) -- (200,250) -- cycle ;
\draw  [draw opacity=0][fill={rgb, 255:red, 184; green, 233; blue, 134 }  ,fill opacity=0.43 ] (250,250) -- (300,250) -- (300,300) -- (250,300) -- cycle ;
\draw [color={rgb, 255:red, 74; green, 74; blue, 74 }  ,draw opacity=1 ][line width=3.75]    (155,155) .. controls (179.2,163.5) and (198.2,184.5) .. (205,205) ;
\draw  [color={rgb, 255:red, 74; green, 74; blue, 74 }  ,draw opacity=1 ][fill={rgb, 255:red, 74; green, 74; blue, 74 }  ,fill opacity=1 ] (200,200) -- (210,200) -- (210,210) -- (200,210) -- cycle ;
\draw  [color={rgb, 255:red, 74; green, 74; blue, 74 }  ,draw opacity=1 ][fill={rgb, 255:red, 74; green, 74; blue, 74 }  ,fill opacity=1 ] (250,250) -- (260,250) -- (260,260) -- (250,260) -- cycle ;
\draw [color={rgb, 255:red, 74; green, 74; blue, 74 }  ,draw opacity=1 ][line width=3.75]    (205,205) .. controls (229.2,213.5) and (248.2,234.5) .. (255,255) ;

\draw (101,70.4) node [anchor=north west][inner sep=0.75pt]    {$T_{1}$};
\draw (151,122.4) node [anchor=north west][inner sep=0.75pt]    {$T_{2}$};
\draw (151,180.4) node [anchor=north west][inner sep=0.75pt]    {$B_{3}$};
\draw (101,130.4) node [anchor=north west][inner sep=0.75pt]    {$B_{2}$};
\draw (51,80.4) node [anchor=north west][inner sep=0.75pt]    {$B_{1}$};
\draw (201,229.4) node [anchor=north west][inner sep=0.75pt]    {$B_{3}$};
\draw (201,172.4) node [anchor=north west][inner sep=0.75pt]    {$T_{3}$};
\draw (251,222.4) node [anchor=north west][inner sep=0.75pt]    {$T_{4}$};
\draw (251,279.4) node [anchor=north west][inner sep=0.75pt]    {$B_{4}$};

\end{tikzpicture}
\caption{Structure of a coordinating semiwire consisting of 4 bridges $T_i,i \in [4]$}
\label{fig:coord}
\end{figure}

The positive coordinating semi-wire is any semi-wire consisting of a constant number of bridges in the same direction, as shown in \Cref{fig:coord}. It holds that an uncoordinated signal at the input of the semi-wire, will not produce a signal arriving at the output. 

\begin{lemma}
    Given a coordinator semi-wire $(S,c) = \sum_{i\in [m]}T_i$ for bridges $T_i = (P_i,c_i)$, for a word $Z$, with $T_1$ a source bridge. There exists some $i < km$ such that $F(c,w,i)$ has no signals for any $w$ \textit{shifted $Z$ cycle}.
\end{lemma}

\begin{proof}
    At this semiwire, each cell $x$ with value 3 has exactly two adjacent cells also with value 3, except for the initial signal at $c$, $P_1(0)$, which has only one adjacent cell with value 3. After any transition $Q$ (where $Q = H$ or $Q = V$), $P_1(0)$ takes the value 2, and its only neighbor with value 3 could potentially become a signal (although this is not guaranteed, since $w \neq Z$). However, we can ensure that $F(c, Q)$ has at most one signal at $P_1(1)$, which now has only one neighbor with value 3 (since its other neighbor, $P_1(0)$, satisfies $F(c,Q)(P_1(0)) = 2$). Following this reasoning, for every $i < km$, $F(c,w,i)$ has at most one signal. Since $w \neq Z$, there exists some $i$ such that $w(i) \neq Z(i)$. Considering this step over the cycle: in $F(c,w,i-1)$, if there is a signal, the signal cell is horizontally (or vertically) connected to a neighbor with value 3. However, if $w(i) = V$ (or $w(i) = H$, respectively), the transition will cause cells with value less than 3 to receive a signal instead. As a result, $F(c,w,i)$ will have no signal at all.
\end{proof}

\subsubsection{Connecting gadgets}

As mentioned earlier, the circuits we construct are sensitive to the time it takes for signals to reach each of the gadgets, as switches do not function if the signal that is supposed to pass arrives before the one that activates it. All gadgets described above have at least one input cell and at least one output cell, and therefore can be connected by superimposing the output pin of one object with the input pin of another. With this, if $t_1$ is the number of cycles a signal takes to travel from the input pin to the output pin of a gadget, and $t_2$ is the same for a second gadget, then the total \textit{travel time} is simply $t_1 + t_2$. We will refer as the \emph{delay} of a gadget as the time it takes for a signal to travel from its input to the corresponding output. Note that the delay of a structure depends solely on the number of bridges the signal must pass through.

In the functionality of gadgets, we treat them as if they were the only cells with grains in the configuration; however, in the final construction, this will not be the case. It is important to note that, in all instances of this dynamics, the grains located in a $Z$-block surrounded by $Z$-blocks with cells of value 0 can only affect adjacent $Z$-blocks. This is because the dynamics is conservative, and the signals will dissipate into the adjacent cells after a few transitions. Therefore, by leaving a constant number of $Z$-blocks between the structures and the semiwires, we can ensure that all gadgets function in parallel without interfering with each other.

\begin{figure}
\centering
\tikzset{every picture/.style={line width=0.75pt}} 

\begin{tikzpicture}[x=0.5pt,y=0.5pt,yscale=-1,xscale=1]

\draw  [draw opacity=0][fill={rgb, 255:red, 184; green, 233; blue, 134 }  ,fill opacity=0.5 ] (200,150) -- (250,150) -- (250,200) -- (200,200) -- cycle ;
\draw  [draw opacity=0][fill={rgb, 255:red, 184; green, 233; blue, 134 }  ,fill opacity=0.5 ] (150,100) -- (200,100) -- (200,150) -- (150,150) -- cycle ;
\draw  [draw opacity=0][fill={rgb, 255:red, 184; green, 233; blue, 134 }  ,fill opacity=0.5 ] (100,50) -- (150,50) -- (150,100) -- (100,100) -- cycle ;
\draw [color={rgb, 255:red, 155; green, 155; blue, 155 }  ,draw opacity=1 ]   (250,50) -- (250,300) ;
\draw [color={rgb, 255:red, 155; green, 155; blue, 155 }  ,draw opacity=1 ]   (300,50) -- (300,300) ;
\draw [color={rgb, 255:red, 155; green, 155; blue, 155 }  ,draw opacity=1 ]   (150,50) -- (150,300) ;
\draw [color={rgb, 255:red, 155; green, 155; blue, 155 }  ,draw opacity=1 ]   (200,50) -- (200,300) ;
\draw [color={rgb, 255:red, 155; green, 155; blue, 155 }  ,draw opacity=1 ]   (100,150) -- (350,150) ;
\draw [color={rgb, 255:red, 155; green, 155; blue, 155 }  ,draw opacity=1 ]   (100,100) -- (350,100) ;
\draw  [draw opacity=0][fill={rgb, 255:red, 184; green, 233; blue, 134 }  ,fill opacity=0.5 ] (300,248.85) -- (350,248.85) -- (350,298.85) -- (300,298.85) -- cycle ;
\draw  [draw opacity=0][fill={rgb, 255:red, 184; green, 233; blue, 134 }  ,fill opacity=0.5 ] (250,199) -- (300,199) -- (300,249) -- (250,249) -- cycle ;
\draw  [draw opacity=0][fill={rgb, 255:red, 184; green, 233; blue, 134 }  ,fill opacity=0.5 ] (300,149) -- (350,149) -- (350,199) -- (300,199) -- cycle ;
\draw [color={rgb, 255:red, 155; green, 155; blue, 155 }  ,draw opacity=1 ]   (100,250) -- (350,250) ;
\draw  [color={rgb, 255:red, 74; green, 74; blue, 74 }  ,draw opacity=1 ][fill={rgb, 255:red, 74; green, 74; blue, 74 }  ,fill opacity=1 ] (300,248.85) -- (310,248.85) -- (310,258.85) -- (300,258.85) -- cycle ;
\draw [color={rgb, 255:red, 155; green, 155; blue, 155 }  ,draw opacity=1 ]   (100,200) -- (350,200) ;
\draw [color={rgb, 255:red, 155; green, 155; blue, 155 }  ,draw opacity=1 ]   (300,299) -- (300,218.7) -- (300,149) ;
\draw [color={rgb, 255:red, 74; green, 74; blue, 74 }  ,draw opacity=1 ][line width=3.75]    (255,204) .. controls (279.2,195.5) and (298.2,174.5) .. (305,154) ;
\draw [color={rgb, 255:red, 74; green, 74; blue, 74 }  ,draw opacity=1 ][line width=3.75]    (255,203.85) .. controls (279.2,212.35) and (298.2,233.35) .. (305,253.85) ;
\draw  [color={rgb, 255:red, 74; green, 74; blue, 74 }  ,draw opacity=1 ][fill={rgb, 255:red, 74; green, 74; blue, 74 }  ,fill opacity=1 ] (300,149) -- (310,149) -- (310,159) -- (300,159) -- cycle ;
\draw  [color={rgb, 255:red, 74; green, 74; blue, 74 }  ,draw opacity=1 ][fill={rgb, 255:red, 74; green, 74; blue, 74 }  ,fill opacity=1 ] (250,209) -- (260,209) -- (260,199) -- (250,199) -- cycle ;
\draw  [color={rgb, 255:red, 155; green, 155; blue, 155 }  ,draw opacity=1 ][line width=0.75]  (100,50) -- (350,50) -- (350,299) -- (100,299) -- cycle ;
\draw [color={rgb, 255:red, 74; green, 74; blue, 74 }  ,draw opacity=1 ][line width=3.75]    (205,154) .. controls (229.2,162.5) and (248.2,183.5) .. (255,204) ;
\draw  [color={rgb, 255:red, 74; green, 74; blue, 74 }  ,draw opacity=1 ][fill={rgb, 255:red, 74; green, 74; blue, 74 }  ,fill opacity=1 ] (200,149) -- (210,149) -- (210,159) -- (200,159) -- cycle ;
\draw [color={rgb, 255:red, 74; green, 74; blue, 74 }  ,draw opacity=1 ][line width=3.75]    (155,105) .. controls (179.2,113.5) and (198.2,134.5) .. (205,155) ;
\draw  [color={rgb, 255:red, 74; green, 74; blue, 74 }  ,draw opacity=1 ][fill={rgb, 255:red, 74; green, 74; blue, 74 }  ,fill opacity=1 ] (150,100) -- (160,100) -- (160,110) -- (150,110) -- cycle ;
\draw [color={rgb, 255:red, 74; green, 74; blue, 74 }  ,draw opacity=1 ][line width=3.75]    (105,55) .. controls (129.2,63.5) and (148.2,84.5) .. (155,105) ;
\draw  [color={rgb, 255:red, 74; green, 74; blue, 74 }  ,draw opacity=1 ][fill={rgb, 255:red, 74; green, 74; blue, 74 }  ,fill opacity=1 ] (100,50) -- (110,50) -- (110,60) -- (100,60) -- cycle ;

\draw (259,161.4) node [anchor=north west][inner sep=0.75pt]    {$T_{1}$};
\draw (269,259.4) node [anchor=north west][inner sep=0.75pt]    {$T_{2}$};
\draw (126,112.4) node [anchor=north west][inner sep=0.75pt]    {$S$};

\end{tikzpicture}
\caption{Correct connection of a coordinator semi-wire $S$ to a semi-duplicator composed of the bridges $T_1,T_2$.}
\label{fig:conection_duplic}
\end{figure}
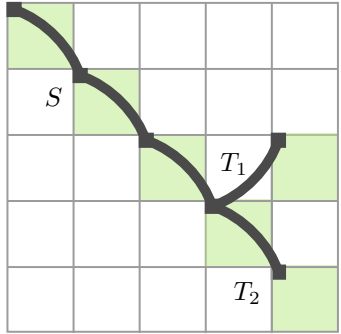

On the other hand, we have the following theorem, wich refeers to the correct combination of semi-wires and the semi-duplicator (of the same polarity), but can be generalized to any other gadget: 
\begin{theorem}
    Let $S$ be the coordinator semiwire whose output coincides with the input of the semi-duplicator composed of the bridges $T_1$ and $T_2$, oriented as shown in \Cref{fig:conection_duplic}, both with the same polarity. If the first bridge of the semiwire is a source bridge, then after a number of cycles equal to the length of the semiwire plus one, there will be two signals at the output of the semi-duplicator.
\end{theorem}
\begin{proof}
    The cells of the duplicator do not interfere with the correct functioning of the semiwire. Therefore, after as many cycles as the number of bridges composing it, there will be exactly one signal (due to the behavior of the coordinator semiwire) at the input of the duplicator. At that point, since only one signal is present, it follows by the same argument used in \Cref{lem:funcduplic} that the signal reaches the outputs.
\end{proof}

This theorem can be generalized to all the gadgets we need, We only need to additionaly ensure that the gadgets with two inputs receive the semi-wire of their upper input from above and the semi-wire of their lower input from below, as shown in \Cref{fig:conection}.

\begin{figure}
\centering
\input{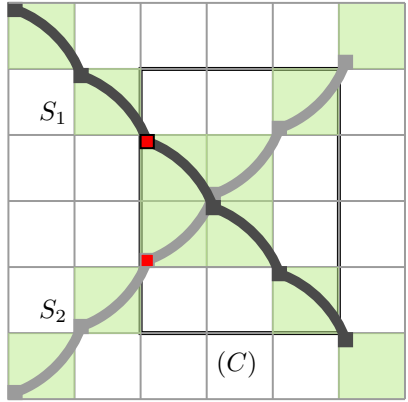}
\caption{Correct connection of two semiwires $S_1$ and $S_2$, each consisting of at least two bridges, is shown in the figure, connected to a semicrossing Gadget $C$, which is contained within the eight $Z$-blocks marked by the black rectangle. The cells marked in red indicate positions that may contain a signal, depending on whether the corresponding semiwire carries one.}
\label{fig:conection}
\end{figure}

It has been shown that the gadgets correctly deliver the signals to the outputs, but we still need to address the cases where a signal additionally exits through an input or when there are "uncoordinated" signals, which will be explained in this section.

\subsubsection{Backward signal travel}

The same construction from coordinator semi-wire ensures that coordinate signals cannot travel backward through coordinating wires, leaving only the problem of uncoordinated signals traveling backward along the semi-wires. To address this, by combining the previously described structures, we can build a positive diode, as shown in \Cref{fig:diode}, which operates in such a way that only coordinated signals can travel from the input to the output, while signals at the output cannot reach the input. Analogously, we can construct a negative diode. As in the previous section, backward-traveling signals within the construction are not inherently problematic. In fact, for many words, it is impossible for such signals to propagate beyond a limited number of blocks. Nevertheless, for completeness and simplicity, we choose to work with diodes.

\begin{figure}
\centering
\input{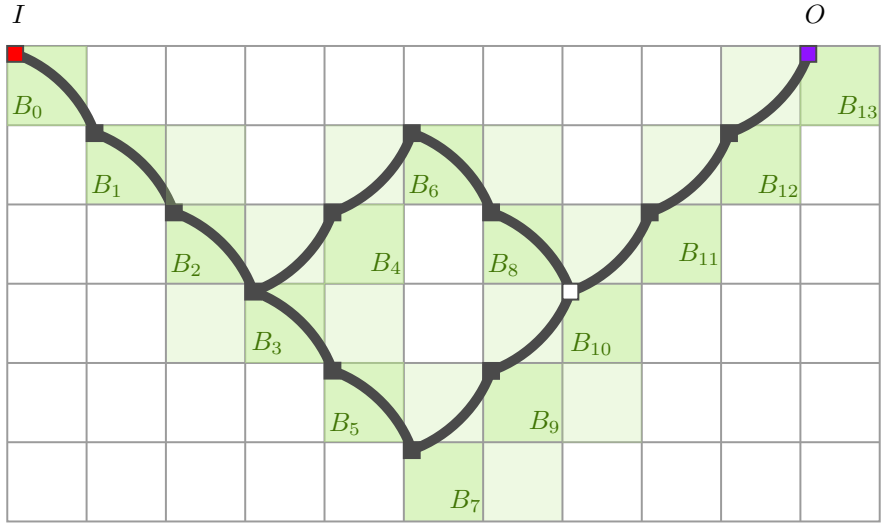}
\caption{Positive diode with input cell $I$ (in red) and output cell $O$ (in purple), formed by positive bridges connecting the blocks $\{B_i\}_{i \in [0...13]}$ using the usual notation. The white square represents a cell with value 2 in $B_{10}^+$, which can be obtained by defining the gadget so that the bridge connecting $B_8$ and $B_{10}$ is a bridge with a sink at its endpoint. If a signal arrives at $I$, it will travel to $B_3^+$, where it will be duplicated, producing two signals that reach $B_{10}^+$. With two signals arriving at $B_{10}^+$, even if its initial value is 2, a signal will be generated at $B_{10}^+$ and will travel to $O$. On the other hand, if a signal starts at $O$ and manages to pass through $B_{11}$, it will necessarily be consumed by $B_{10}^+$.}
\label{fig:diode}
\end{figure}

\subsection{Retarders}
\begin{figure}
\centering
\input{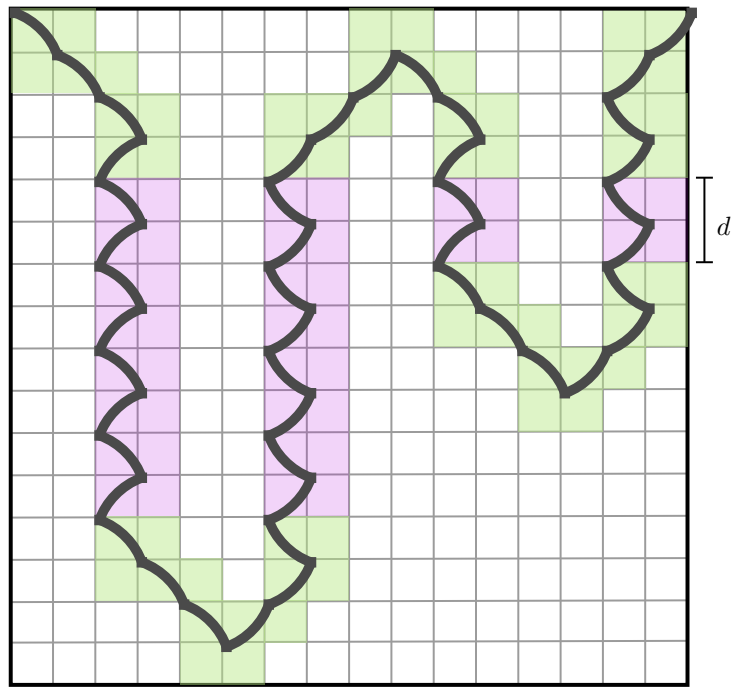}
\caption{Positive Retarder formed by bridges connecting diagonally adjacent blocks. The marked $Z$-blocks indicate those that may contain grains in this gadget. The purple blocks represent sections of the retarder that could vary. In this case, the value $d$ could be increased, to make the right section of the retarder to be "extended" (At most to match the length of the left section), allowing the retarder to have a bigger delay.}
\label{fig:retarder_positivo}
\end{figure}

Retarders will be used to delay a signal as it travels from left to right. This can be achieved using a structure like the one shown in \Cref{fig:retarder_positivo}. Those are very important in the construction, since the switches we defined require that the activating signal arrives before the signal that is supposed to pass. Note that if we constrain the retarder to a space of $L \times L$ $Z-$blocks, the number of bridges that can be placed within the purple section is $\frac{L}{4}(L - 4)$ (Since the first and last two rows are reserved for the U-turns of the semiwires, and only $\frac{1}{4}$ of the remaining blocks are used in bridges), allowing the delay introduced by the retarder to grow quadratically with respect to its width and height. In the construction we describe, all retarders will have the same size; however, the delays they produce will differ. Using fixed-size retarders for any given circuit greatly simplifies the overall layout.

\section{Layer 4: The Logspace Constructor} \label{sec:layer-4}

\fxnote[inline,nomargin]{Add some introductory text, explain we are adapting the
original constructor, etc.}

For the logarithmic-space construction we follow an analysis similar to that of \cite{modanese2024embedding}. 
The complexity of the construction now lies in the fact that the constructor must determine in which diagonally adjacent $Z$-blocks a bridge must be placed. This introduces an additional layer of complexity where the constructor is also required to explicitly build the bridges between such blocks, compared to the previous construction, where it directly adds cells onto the $2\times2$ subdivision if necessary.

Consider a instance of the CVP with $m$ gates. The description of $F$ is column by column. At each step, we will keep track of the coordinates we are currently at and return them as needed. The size of the structure is polynomial in $m$, so we can store the coordinates of $F$ using logarithmic space. Before initializing the construction, we must compute $D$ and store it. Since $D = \Theta(m^4)$, this can easily be done using logarithmic space by examining the end of the input to obtain $m$ (which we will also store in logarithmic space).

Before describing each section, we must define the inputs. Knowing $D$, since each wire maintains a fixed height in each section that only depends on $D$, we only need to construct the wires $\{x_1, \dots, x_n\}$ at the corresponding heights, assigning positive or negative signals as appropriate.

Given a circuit encoded node by node, the construction of $F$ will proceed tile by tile. At each point in the construction, since all sections have the same size, we can calculate the relative coordinates of each section and, from there, assemble the wires and crossings at the corresponding positions.

To construct section $T_i$, we need to know:
\begin{itemize}
    \item The number of wires in the upper part.
    \item The delays for the crossings.
    \item The locations of the crossings.
    \item The corresponding gate.
\end{itemize}

The number of wires in the upper part comes directly from the section number. As we saw in the description of $F$, the number of wires entering the upper part is $n + i - 1$, and the number exiting is $n + i$. The position and length of the wire that descends to the gate depend only on the number of wires in the section, which is known.

The delays required for each crossing and NAND gate can be computed in logarithmic space by sequentially calculating the entry delay of the tile, $d_T = imO(D)$, the entry delay of each crossing $C_j$, $d_{C_j} = d_T + jO(D)$, and the corresponding delay for each NAND gate $N_k$ of the crossing, $d_{N_k} = d_{C_j} + kO(D)$. At each step, we will store these three quantities, but since they are all polynomial in $m$, they can be stored in logarithmic space.

The positions of the crossings are fixed in $T_i$; we only need to determine where crossings will actually occur. This can be obtained directly from the encoding since we saw that each gate has its two inputs $y_1, y_2$ in the encoding. Thus, all wires below these two must have crossings.

The type of gate in the section can be easily determined from the input given the encoding, and the necessary delay for its NAND gates is $iO(D)$, so it can also be computed. Its position in the section is fixed, so by calculating the delays of its NAND gates, we can construct it without issue.

Finally, we describe the coordinates of the target cell (which corresponds to the output cell of the last gate), and we write the waiting time in unary, which can be obtained from the exit time of the last tile. This time is stored in binary, but we write it in unary for the output.

\section{Conclusion}

To conclude, in this section we present the complete proof of the theorem, with explicit references to the required constructions.

\begin{theorem}
  If $Z \in \{ H,V \}^+$ is any word that contains at least one $H$ and one $V$,
  then the fungal sandpile prediction problem associated with $Z$ is
  $\P$-complete.
\end{theorem}

\begin{proof}
The proof proceeds via an efficient reduction from the CVP to the fungal sandpile prediction problem associated with $Z$. Given an instance $C$ of the CVP, we compute the triple $(R,x,T)$.

The region $R$ is a finite grid with configuration $c$ that follows the same tile scheme described in \textbf{Layer 4} of \cite{modanese2024embedding}, with the following modifications:
\begin{itemize}
    \item The subdivision is performed into blocks whose size depends on $Z$ (instead of a $2 \times 2$ subdivision), as described in \Cref{sec:main}.
    \item All gadgets introduced in \cite{modanese2024embedding} are replaced by the generalized gadgets defined in \Cref{sec:layer-2}, and the delay is redefined as in that section.
\end{itemize}

The parameters $x$ and $T$ are chosen analogously to \cite{modanese2024embedding}. Taking into account the adaptations above. The correctness of the reduction follows directly from the functionality of each gadget of stated in \Cref{sec:layer-2} and the overall correctness stated in \cite{modanese2024embedding}.

Finally, the reduction is computable in logarithmic space by the arguments given in \Cref{sec:layer-4}.
\end{proof}
 
\bibliographystyle{splncs04}
\bibliography{sn-bibliography}


%


\end{document}